\documentclass[a4paper,UKenglish,numberwithinsect,cleveref, autoref, thm-restate]{lipics-v2019}

\usepackage[justification=centering]{caption}
\usepackage{subcaption}
\usepackage[ruled, linesnumbered]{algorithm2e}
\usepackage{xspace}
\usepackage{todonotes}

\newtheorem{observation}{Observation}[section]

\date{}

\renewcommand{\paragraph}{\subparagraph*}
\allowdisplaybreaks

\newcommand{\OO}{\mathcal{O}}

\newcommand{\yes}{\textsf{Yes}}
\newcommand{\no}{\textsf{No}}

\newcommand{\NP}{\textsf{NP}}

\newcommand{\FPT}{\textsf{FPT}}

\newcommand{\NPH}{\textsf{NP}-hard}

\newcommand{\paraH}{\textsf{para-NP}-hard}

\newcommand{\cc}{\mathsf{mcc}}

\newcommand{\bGridEm}{{\sc $k \times r$-Grid Embedding}}
\newcommand{\gridEm}{{\sc Grid Embedding}}
\newcommand{\recGrid}{$k\times r$ rectangular grid graph}
\newcommand{\source}{\textsf{Source}}
\newcommand{\sink}{\textsf{Sink}}
\newcommand{\pw}{\mathtt{pw}}
\newcommand{\partition}{{\sc $3$-Partition}}
\newcommand{\degr}{\mathrm{deg}}
\newcommand{\fc}{f_\mathsf{col}}
\newcommand{\fr}{f_\mathsf{row}}
\newcommand{\num}{\mathsf{num}}
\newcommand{\ad}{\mathsf{Adj}}
\newcommand{\frl}{\mathsf{freq_{left}}}
\newcommand{\frr}{\mathsf{freq_{right}}}
\newcommand{\frc}{\mathsf{freq_{cen}}}
\newcommand{\frb}{\mathsf{freq_{boun}}}
\newcommand{\join}{\mathsf{join}}
\newcommand{\len}{\mathsf{length}}
\newcommand{\bre}{\mathsf{breath}}
\newcommand{\C}{\mathsf{Comp}}
\newcommand{\td}{\mathtt{td}}

\title{Grid Recognition: Classical and Parameterized Computational Perspectives} 

\titlerunning{Grid Recognition} 

\author{Siddharth Gupta}{Ben-Gurion University of the Negev, Israel}{siddhart@post.bgu.ac.il}{}{}

\author{Guy Sa'ar}{Ben-Gurion University of the Negev, Israel}{saag@post.bgu.ac.il}{}{}

\author{Meirav Zehavi}{Ben-Gurion University of the Negev, Israel}{meiravze@bgu.ac.il}{}{}

\authorrunning{S.Gupta, G. Sa'ar and M. Zehavi} 

\Copyright{Siddharth Gupta, Guy Sa'ar and Meirav Zehavi} 



\nolinenumbers 

\hideLIPIcs  

\EventLongTitle{The 32nd International Symposium on Algorithms and Computation}
\EventShortTitle{ISAAC 2021}
\EventAcronym{ISAAC}
\EventYear{2021}
\EventDate{December 6–8, 2021}
\EventLocation{Japan}
\EventLogo{}

\begin{document}

\maketitle

\begin{abstract}
Grid graphs, and, more generally, $k\times r$ grid graphs, form one of the most basic classes of geometric graphs. Over the past few decades, a large body of works studied the (in)tractability of various computational problems on grid graphs, which often yield substantially faster algorithms than general graphs. Unfortunately, the recognition of a grid graph (given a graph $G$, decide whether it is a grid graph) is particularly hard---it was shown to be NP-hard even on trees of pathwidth 3 already in 1987. Yet, in this paper, we provide several positive results in this regard in the framework of parameterized complexity (additionally, we present new and complementary hardness results). Specifically, our contribution is threefold.
First, we show that the problem is fixed-parameter tractable (FPT) parameterized by $k+\mathsf {mcc}$ where $\mathsf{mcc}$ is the maximum size of a connected component of $G$. This also implies that the problem is FPT parameterized by $\td+k$ where $
\td$ is the treedepth of $G$ (to be compared with the hardness for pathwidth 2 where $k=3$). (We note that when $k$ and $r$ are unrestricted, the problem is trivially FPT parameterized by $\td$.) Further, we derive as a corollary that strip packing is FPT with respect to the height of the strip plus the maximum of the dimensions of the packed rectangles, which was previously only known to be in XP. 
Second, we present a new parameterization, denoted $a_G$, relating graph distance to geometric distance, which may be of independent interest. We show that the problem is para-NP-hard parameterized by $a_G$, but FPT parameterized by $a_G$ on trees, as well as FPT parameterized by $k+a_G$.
Third, we show that the recognition of $k\times r$ grid graphs is NP-hard on graphs of pathwidth 2 where $k=3$. Moreover,  when $k$ and $r$ are unrestricted, we show that the problem is NP-hard on trees of pathwidth 2, but trivially solvable in polynomial time on graphs of pathwidth 1.
\end{abstract}

\clearpage
\pagenumbering{arabic} 

 
\section{Introduction}\label{sec:intro}

Geometrically, a {\em grid graph} is a graph that can be drawn on the Euclidean plane so that all vertices are drawn on points having positive integer coordinates, and  all edges are drawn as axis-parallel straight line segments of length 1;\footnote{Some papers in the literature use the term grid graphs to refer to {\em induced} grid graphs, where we require also that every pair of vertices at distance $1$ from each other have an edge between them.} when the maximum $x$-coordinate is at most $r$ and the maximum $y$-coordinate is at most $k$, we may use the term {\em $k\times r$ grid graph} (see Figure~\ref{fi:rectangleGraph}). Grid graphs form one of the simplest and most intuitive classes of geometric graphs.  Over the past few decades, algorithmic research of grid graphs yielded a large body of works on the tractability or intractability of various computational problems when restricted to grid graphs (e.g., see \cite{DBLP:conf/stoc/ChuzhoyKN18,DBLP:conf/icalp/ChuzhoyKN18,allender2006grid,itai1982hamilton,umans1997hamiltonian,DBLP:conf/stoc/BergBKMZ18,DBLP:journals/siamcomp/ArkinBDFMS05} for a few examples). Even for problems that remain \NP-hard\ on grid graphs, we know of practical algorithms for instances of moderate size (e.g., the {\sc Steiner Tree} problem on grid graphs is \NP-hard~\cite{garey1977rectilinear}, but admits practical algorithms~\cite{ganley1999computing,zachariasen2001catalog}). Thus, the recognition of a graph as a grid graph unlocks highly efficient tools for its analysis. In practice, grid graphs can represent layouts or environments, and have found applications in several fields, such as VLSI design \cite{sait1999vlsi}, motion planning \cite{handbookMovement} and routing~\cite{sturtevant2012benchmarks}.  Indeed, grid graphs naturally arise to represent entities and the connections between them in {\em existing} layouts or environments. However, often we are given just a (combinatorial) graph $G$---i.e., we are given entities and the  connections desired to have between them, and we are to construct the layout or environment; specifically, we wish to test whether $G$ is a grid graph (where if it is so, realize it as such a graph).  Equivalently, the recognition of a grid graph can be viewed as an embedding problem, where a given graph is to be embedded within a rectangular solid grid.

Accordingly, the problem of recognizing (as well as realizing) grid graphs is a basic recognition problem in Graph Drawing. In what follows, we discuss only recognition---however, it would be clear that all of our results hold also for realization (with the same time complexity in case of algorithms). Formally, in the \gridEm\ problem, we are given a (simple, undirected) $n$-vertex graph $G$, and need to decide whether it is a grid graph. In many cases, taking into account physical constraints, compactness or visual clarity, we would like to not only have a grid graph, but also restrict its dimensions. This yields the \bGridEm\ problem, where given an $n$-vertex graph $G$ and positive integers $k,r\in\mathbb{N}$, we need to decide whether $G$ is a $k\times r$ grid graph. Notice that \gridEm\ is the special case of \bGridEm\ where $k=r=n$ (which virtually means that no dimension restriction is posed).

The \gridEm\ problem has been proven to be \NP-hard\ already in 1987, even on trees of pathwidth $3$~\cite{bhatt1987complexity}. Shortly afterwards, it has been proven to be \NP-hard\ even on binary trees~\cite{gregori1989unit}. On the positive side, there is research on practical algorithms for this problem~\cite{beck2020puzzling}. The related
upward planarity testing and rectilinear planarity testing problems are also known to be \NP-hard~\cite{garg2001computational}, as well as HV-planarity testing even on graphs of maximum degree $3$~\cite{didimo2014complexity}. We remark that when the embedding is fixed, i.e., the clockwise order of the edges is given for each vertex, the situation becomes drastically easier computationally; then, for example, a rectangular drawing of a plane graph of maximum  degree $3$, as well as an orthogonal drawing without bends of a plane graph of maximum  degree $3$, were shown to be computable in linear time in \cite{rahman1998rectangular} and \cite{rahman2006orthogonal}, respectively.

In this paper, we study the classical and parameterized complexity of the \gridEm\ and \bGridEm\ problems. To the best of our knowledge, this is the first time that these problems are studied from the perspective of parameterized complexity. Let $\Pi$ be an NP-hard problem. In the framework of parameterized complexity, each instance of $\Pi$ is associated with a {\em parameter} $k$. Here, the goal is to confine the combinatorial explosion in the running time of an algorithm for $\Pi$ to depend only on $k$. 
Formally, we say that $\Pi$ is {\em fixed-parameter tractable (\FPT{})} if any instance $(I, k)$ of $\Pi$ is solvable in time $f(k)\cdot |I|^{\OO(1)}$, where $f$ is an arbitrary computable function of $k$.  Notably, this means that whenever $f(k)=|I|^{\OO(1)}$, the problem is solvable in polynomial time.
Nowadays, Parameterized Complexity supplies a rich toolkit to design \FPT{} algorithms as well as to prove that some problems are unlikely to be \FPT{} \cite{DBLP:series/txcs/DowneyF13,DBLP:books/sp/CyganFKLMPPS15,fomin2019kernelization}. In particular, the term para-\NP-hard refers to problems that are \NP-hard\ even when the parameter is fixed (i.e., a constant, which is not part of the input), which implies that they are not \FPT\ unless P=\NP.

Research at the intersection of graph drawing and parameterized complexity (and parameterized algorithms in particular) is in its infancy. Most (in particular, the early efforts) have been directed at variants of the classic {\sc Crossing Minimization} problem, introduced by Turán in 1940~\cite{doi:10.1002/jgt.3190010105}, parameterized by the number of crossings (see, e.g., \cite{DBLP:journals/jcss/Grohe04,DBLP:conf/stoc/KawarabayashiR07,DBLP:journals/algorithmica/DujmovicFKLMNRRWW08,DBLP:conf/compgeom/HlinenyD16,hs-ecnpvc-19,DBLP:conf/compgeom/KluteN18}). However, in the past few years, there is an increasing interest in the analysis of a variety of other problems in graph drawing from the perspective of parameterized complexity (see, e.g., \cite{DBLP:conf/gd/BhoreGMN19,DBLP:conf/compgeom/AgrawalGM0Z19,DBLP:conf/wads/HalldorssonKST07,DBLP:conf/esa/Chan04,DBLP:journals/ijfcs/HealyL06,DBLP:journals/jgaa/BannisterCE18,DBLP:conf/isaac/DidimoL98,DBLP:journals/corr/abs-1908-05015,DBLP:journals/algorithmica/BlasiusKRW14,DBLP:journals/jcss/DidimoLP19,DBLP:conf/iwpec/LozzoEG019,DBLP:conf/wg/LozzoEG018} and the upcoming Dagstuhl seminar~\cite{DagstuhlGD}).

\subsection*{Our Contribution and Main Proof Ideas} 

\noindent{\bf\em I. Parameterized Complexity: Maximum Connected Component Size.} Our contribution is threefold. First,
we prove that \bGridEm\ is \FPT\ parameterized by $\mathsf{mcc}+k$. Here, the idea of the proof is first to recognize all possible embeddings of any choice of connected components or parts of connected components of $G$ into $k\times \cc(G)$ grids, called blocks. These blocks then serve as vertices of a new digraph, where there is an arc from one vertex to another if and only if the corresponding blocks can be placed one after the other. After that, we also guess which blocks should occur at least once in the solution, as well as a spanning tree of the underlying undirected graph of the graph induced on them. This then leads us to a formulation of an Integer Linear Program (ILP), where we ensure that each connected component is used as many times as it is in the input, and that overall we get an Eulerian trail in the graph---having such a trail allows us to place the blocks one after the other, so that every pair of consecutive blocks are compatible. The ILP can then be solved using known tools.

\begin{restatable}{theorem}{mccK}\label{thm:mccK}
\bGridEm\ is \FPT\ parameterized by $\mathsf{mcc}+k$ where $\mathsf{mcc}$ is the maximum size of a connected component in the input graph.
\end{restatable}

One almost immediate corollary of this theorem concerns the {\sc 2-Strip Packing} problem. In this problem, we are given a set of $n$ rectangles $S$, and positive integers $k,W\in\mathbb{N}$, and the objective is to decide whether all the rectangles in $S$ can be packed in a rectangle (called a {\em strip}) of dimensions $k\times W$. In \cite{DBLP:journals/tcs/AshokKMS17}, it was shown that if  the maximum of the dimensions of the input rectangles, denoted by $\ell$, is fixed (i.e., a constant independent of the input), then the problem is \FPT\ parameterized by $k$. Specifically, running time of $8^{k\ell}n^{\OO(\ell^2)}W$ was attained, which is not \FPT\ with respect to $k+\ell$. Thus, the question whether the problem is \FPT\ parameterized by $k+\ell$ remained open. By a simple reduction from \bGridEm, we resolve this question as a corollary of our Theorem \ref{thm:mccK}.

\begin{restatable}{corollary}{stripPacking}
{\sc 2-Strip Packing} is \FPT\ parameterized by $\ell+k$ where $\ell$ is the maximum of the dimensions of the input rectangles.
\end{restatable}

We remark that in case $k$ and $r$ are unrestricted, the problem is trivially \FPT\ with respect to $\mathsf{mcc}$, since one can embed each connected component (using brute-force) individually.

\begin{restatable}{observation}{mccUnrestricted}
\gridEm\ is \FPT\ parameterized by $\mathsf{mcc}$ where $\mathsf{mcc}$ is the maximum size of a connected component in the input graph.
\end{restatable}

As a corollary of our theorem and this observation, we obtain that \bGridEm\ is \FPT\ parameterized by $\td+k$, and \gridEm\ is \FPT\ parameterized by $\td$, where $\td$ is the treedepth of the input graph. This finding is of interest when contrasted with the hardness of these problems when pathwidth (and hence also treewidth) equals $2$ and $k=3$ or unrestricted. Thus, this also charts a tractability border between pathwidth and treedepth.

\begin{restatable}{corollary}{treedepthK}
\bGridEm\ is \FPT\ parameterized by $\td+k$, and \gridEm\ is \FPT\ parameterized by $\td$, where $\td$ is the treedepth of the input graph.
\end{restatable}

\medskip
\noindent{\bf\em II. Parameterized Complexity: Difference Between Graph and Geometric Distances.} Secondly, we introduce a new parameterization that relates graph distance to geometric distance, and may be of independent interest. Roughly speaking, the rationale behind this parameterization is to bound the difference between them, so that graph distances may act as approximate indicators to geometric distances. In particular, vertices that are close in the graph, are to be close in the embedding, and vertices that are distant in the graph, are to be distant in the embedding as well. Specifically, with respect to an embedding $f$ of $G$ in a grid, we define the {\em grid distance} between any two vertices as the distance between them in $f$ in L1 norm. Then, we define the measure of {\em distance approximation} of $f$ as the maximum of the difference between the graph distance (in $G$) and the grid distance of two vertices, taken over all pairs of vertices in $G$.  Here, it is implicitly assumed that $G$ is connected. Then, the parameter $a_G$ is the minimum distance approximation $a_f$ of any embedding $f$ of $G$ in a (possibly $k\times r$) grid, defined as $|V(G)|$ if no such embedding exists. A more formal definition as well as motivation is given in Section \ref{sec:prelims}.

We first prove that the problems are para-\NP-hard\ parameterized by $a_G$. This reduction is quite technical. On a high level, we present a construction of ``blocks'' that are embedded in a grid-like fashion, where we place an outer ``frame'' of the form of a grid to guarantee that the boundary (which is a cycle) of each of these blocks must be embedded as a square. Each variable is associated with a column of blocks, and each clause is associated with a row of blocks. Within each block, we place two gadgets, one which transmits information in a row-like fashion, to ensure that the clause corresponding to the row has at least one literal that is satisfied, and the other (which is very different than the first) transmits information in a column-like fashion, to ensure consistency between all blocks corresponding to the same variable (i.e., that all of them will be embedded internally in a way that represents only truth, or only false). For the sake of clarity, we split the reduction into two, and use as an intermediate problem a new problem that we call the {\sc Batteries} problem.

\begin{restatable}{theorem}{distanceHard}\label{thm:distanceHard}
\gridEm\ (and hence also \bGridEm) is para-\NP-hard\ parameterized by $a_G$.
\end{restatable}

When we enrich the parameterization by $k$, then \bGridEm\ problem becomes \FPT. (Recall that parameterized by $k$ alone, the problem is para-\NP-hard). The idea of the proof is to partition a rectangular solid $k \times r$ grid in which we embed our graph into blocks of size $k \times (a_G+k)$, and ``guess'' one vertex that is to be embedded in the leftmost column of the leftmost block. Then, the crux is in the observation that, for every vertex, the block in which it should be placed is ``almost'' fixed---that is, we can determine two consecutive blocks in which the vertex may be placed, and then we only have a choice of one among them. This, in turn, leads us to the design of an iterative procedure that traverses the blocks from left to right, and stores, among other information, which vertices were used in the previous block.

\begin{restatable}{theorem}{distanceFPT}\label{thm:distanceFPT}
\bGridEm\ is \FPT\ parameterized by $a_G+k$.
\end{restatable}

Lastly, we prove that when restricted to trees, the problems become \FPT\ parameterized by $a_G$ alone. Here, a crucial ingredient is to understand the structure of the tree, including a bound on the number of vertices of degree at least 3 in the tree that split it to ``large'' subtrees. For this, one of the central insights is that, with respect to an internal vertex $v$ and any two ``large'' subtrees attached to it (there can be up to four subtrees attached to it), in order not to exceed the allowed difference between the graph and geometric distances, one of the subtrees must be embedded in the ``opposite'' direction of the other (so, both are embedded roughly on the same vertical or horizontal line in opposite sides). Now, for an internal vertex of degree at least 3, there must be two attached subtrees that are not embedded in this fashion (as a line can only accommodate two subtrees), which leads us to the conclusion that all but two of the attached subtrees are small. Making use of this ingredient, we argue that a dynamic programming procedure (somewhat similar to the one mentioned for the previous theorem but much more involved) can be used.

\begin{restatable}{theorem}{distanceTrees}\label{thm:distanceTrees}
\bGridEm\ (and hence also \gridEm) on trees is \FPT\ parameterized by $a_G$.
\end{restatable}

\medskip
\noindent{\bf\em III. Classical Complexity.} Lastly, we extend current knowledge of the classical complexity of \gridEm\ and \bGridEm\ at several fronts. Here, we begin by developing a refinement the classic reduction from {\sc Not-All-Equal 3SAT} in \cite{bhatt1987complexity} (which asserted hardness on trees of pathwidth 3) to derive the following result. While the reduction itself is similar, our proof is more involved and requires, in particular, new inductive arguments. 

\begin{restatable}{theorem}{hardnessUnrestricted}\label{thm:hardnessUnR}
\gridEm\ is \NP-hard even on trees of pathwidth 2. Thus, it is \paraH\ parameterized by $\pw$, where $\pw$ is the pathwidth of the input graph.
\end{restatable}

Because \gridEm\ is a special case of \bGridEm, the above theorem has the following result as an immediate corollary.

\begin{corollary}
\bGridEm\ is \NP-hard even on trees of pathwidth 2.
\end{corollary}

In particular, now the hardness result is {\em tight} with respect to pathwidth due to the following simple observation.

\begin{observation}
\gridEm\ is solvable in polynomial time on graphs of pathwidth~1.
\end{observation}

Additionally, we show that \bGridEm\ is \NP-hard\ on graphs of pathwidth 2 even when $k=3$. Here, we give a reduction from {\sc 3-Partition} (whose objective is to partition a set of numbers encoded in unary into sets of size 3 that sum up to the same number), where the idea is to encode ``containers'' by special identical connected components whose embedding is essentially fixed, and then each number as a simple path on a corresponding number of vertices.

\begin{restatable}{theorem}{hardnessK}\label{thm:hardnessK}
\bGridEm\ is \NPH\ even on graphs of pathwidth $2$ when $k=3$. Thus, it is \paraH\ parameterized by $k + \pw$, where $\pw$ is the pathwidth of the input graph.
\end{restatable}


\section{Preliminaries}\label{sec:prelims}
For $k\in \mathbb{N}$, denote $[k] = \{1,2,\ldots, k\}$, and for $i,j\in \mathbb{N}$, denote $[i,j] = \{i, i+1, \ldots, j\}$. Given a set $W$ of integers, $\sum W$ denotes the sum of the integers in $W$. Given two multisets $X = \{x_1, x_2, \ldots, x_k\}$ and $Y = \{y_1, y_2, \ldots, y_l\}$, their \emph{disjoint union} is the multiset $Z = X \uplus Y = \{x_1, x_2 , \ldots, x_k, y_1, y_2, \ldots, y_l\}$. Given a function $g$ defined on a set $W$, we denote the set of images of its elements by $g(W)$. 

\paragraph{\bf{Graphs.}} For other standard notations not explicitly defined here, we refer to the book~\cite{Diestelbook}. Given a graph $G$, we denote its vertex set and edge set by $V(G)$ and $E(G)$, respectively. For a vertex $v \in V(G)$, we denote the degree of $v$ in $G$ by $\degr_G(v)$. The maximum degree of a vertex in $G$ is denoted by $\Delta(G)$. Given a set $V' \subseteq V(G)$, the subgraph of $G$ induced by $V$ is denoted by $G[V']$. Given a path $P$, \emph{size} and \emph{length} denote the number of vertices and edges in $P$, respectively. Given $u,v\in V(G)$, the distance $d(u,v)$ between $u$ and $v$ in $G$ is the length of a shortest path between them in $G$. A \emph{caterpillar} is a tree in which all the vertices are within distance 1 of a central path. Given two graphs $G_1$ and $G_2$, they are called \emph{isomorphic} if there exists a bijective function $g:V(G_1) \rightarrow V(G_2)$ such that any two vertices $u, v \in V(G_1)$ are adjacent in $G_1$ if and only if $g(u)$ and $g(v)$ are adjacent in $G_2$.
Given a graph $G$ and two vertices $u, v \in V(G)$, we define the operation $\join(u,v)$ on $G$ as follows. Delete $u$ and $v$ from $G$ and add a new vertex $w$ to $G$. Attach all the neighbors of $u$ and $v$ to $w$. If $\{u,v\} \in E(G)$, add a \emph{self-loop} on $w$, i.e.~an edge whose both endpoints are $w$.
The \emph{disjoint union} of two graphs $G_1$ and $G_2$ is the graph with vertex set $V(G_1) \uplus V(G_2)$ and edge set $E(G_1) \uplus E(G_2)$. The \emph{pathwidth} and \emph{treedepth} of a graph $G$ are defined as follows.

\begin{definition}[{\bf Pathwidth}]\label{def:pathwidth}
A {\em path decomposition} of a graph $G$ is a sequence $V_1, V_2, \ldots,$ $ V_t$ of subsets of $V(G)$ with the following two properties:
\begin{enumerate}[(i)]
	\item For every edge $\{u,v\} \in E(G)$, there exists an index $i \in [t]$ such that $u, v \in V_i$, and
	\item For every vertex $u \in V(G), \{j~|~ u \in V_j\} = \{i~|~k \leq i \leq l\}$, for some $1 \leq k \leq l \leq t$.
\end{enumerate}
The \emph{width} of the decomposition is one less than the maximum size of any set $V_i$, and the pathwidth $\pw(G)$ of $G$ is the minimum width of any of its path decompositions.
\end{definition}

\begin{definition}[{\bf Treedepth}]
The {\em treedepth} of a graph $G$ is defined as the minimum height of a forest $F$ on the same vertex set as $G$ with the property that every edge in $E(G)$ connects a pair of vertices that have an ancestor-descendant relationship in $F$.
\end{definition}

Based on the definition of pathwidth, we have the following observation about disjoint union of graphs.

\begin{observation}\label{obs:pathwidthDisjoint}
The pathwidth of the disjoint union of two graphs $G_1$ and $G_2$ is max$\{\pw(G_1)$, $\pw(G_2)\}$.
\end{observation}

Based on the definition of treedepth, we have the following observation about graphs of bounded degree.

\begin{observation}[\cite{DBLP:books/daglib/0030491}]\label{obs:treedepthBounded}
A graph of bounded degree and bounded treedepth has bounded size.
\end{observation} 

\paragraph{\bf{Grid Embedding.}} 

Let $f:V(G) \rightarrow \mathbb{N} \times \mathbb{N}$ be a function that maps each vertex $v$ of $G$ to a point $(i,j)$ of an integer grid; then, $i$ and $j$ are also denoted as $\fr(v)$ and $\fc(v)$, respectively, that is, $f(v) = (\fr(v), \fc(v))$.
We now define some basic notions that are needed to define the \bGridEm\ and \gridEm\ problems.

\begin{definition}[{\bf Grid Graph Distance}] \label{def:Grid graph distance}
Let $G$ be an undirected graph. Let $f:V(G)\rightarrow \mathbb{N} \times \mathbb{N}$. Let $u,v\in V(G)$. The {\em grid graph distance of $u$ and $v$ induced by $f$}, denoted by $d_f(u,v)$, is defined to be $d_f(u,v)=|\fr(u)-\fr(v)|+|\fc(u)-\fc(v)|$.
\end{definition}

\begin{definition}[{\bf Grid Graph Embedding}] \label{def:Grid graph embedding}
Let $G$ be an undirected graph. Let $k,r\in \mathbb{N}$ such that $1\leq k,r\leq |V(G)|$. A {\em $k\times r$ grid graph embedding} of $G$ is an injection $f:V(G)\rightarrow [k] \times [r]$, such that for every $\{u,v\}\in E(G)$ it follows that $d_f(u,v)=1$. Moreover, a {\em grid graph embedding} of $G$ is a $|V(G)| \times |V(G)|$ grid graph embedding of $G$. 
\end{definition}

\begin{definition} [{\bf Grid Graph}] \label{def:Grid graph}
Let $G$ be an undirected graph. Let $k,r\in \mathbb{N}$. Then, $G$ is a {\em $k\times r$ grid graph} if there exists a $k\times r$ grid graph embedding of $G$. Moreover, $G$ is a {\em grid graph} if there exists a grid graph embedding of $G$. 
\end{definition}

The \bGridEm\ and the \gridEm\ problems are defined as follows.

\begin{definition} [{\bf Grid Embedding Problem}]\label{def:GridEmbedding}
Given a graph $G$ and two positive integers $k, r$, the \bGridEm\ and \gridEm\ problems ask whether $G$ is a $k\times r$ grid graph or a grid graph, respectively.
\end{definition}

\paragraph{\bf{Distance Approximation Parameter.}} Before we discuss motivation, let us first define the parameter formally. Towards this, we first present the following simple observation.

\begin{observation}\label{lem:dfLeqd}
Let $G$ be a grid graph with a grid embedding $f$, and let $u,v\in V(G)$. Then, $d_f(u,v)\leq d(u,v)$.
\end{observation}

\begin{proof}
We prove by induction on $d(u,v)$. If $d(u,v)=0$ then $u=v$, and we get that $d_f(u,u)=0=d(u,v)$. Now, assume that $d(u,v)=k\geq 1$. Let $P=(a_0=u,a_1\ldots,a_k=v)$ be a path of size $k$ form $u$ to $v$ in $G$. Observe that $P=(a_0=u,a_1\ldots,a_{k-1})$ is a path of size $k-1$ from $u$ to $a_{k-1}$, so we get that $d(u,a_{k-1})\leq k-1$. Therefore, by the inductive hypothesis, we get that $d_f(u,a_{k-1})\leq d(u,a_{k-1})\leq k-1$. Since $\{a_{k-1},v\}\in E(G)$ and by Definition \ref{def:Grid graph embedding}, we get that $d_f(a_{k-1},v)=1$. By the triangle inequality, we get that $d_f(u,v)\leq d_f(u,a_{k-1})+d_f(a_{k-1},v)\leq k-1+1=k=d(u,v)$. 
\end{proof}

This observation implies that when we consider differences of the form $|d(u,v)-d_f(u,v)|$, we can drop the absolute value. Keeping this in mind, we now formally define the distance approximation parameter.

\begin{definition}[{\bf $k\times r$ Distance Approximation Parameter}] \label{def:Grid graph distance app}
Let $G$ be a connected graph, and let $k,r\in \mathbb{N}$. For any $k\times r$ grid graph embedding $f$ of $G$, define $a_f=\max_{u,v\in V(G)}(d(u-v)-d_f(u,v))$. Then,
\begin{itemize}
\item If $G$ is a $k\times r$ grid graph, then $a_G(k,r)=\min\{a_f~|~f$ is a $k\times r$ grid graph embedding~of~$G\}$.
\item Otherwise, $a_G(k,r)=|V(G)|$.
\end{itemize}
\end{definition}

When $k$ and $r$ are clear from context, we write ``distance approximation parameter'' and $a_G$ rather than ``$k\times r$ distance approximation parameter'' and $a_G(k,r)$, respectively. When $k$ and $r$ are unrestricted, $a_G(k,r)=a_G(|V(G)|,|V(G)|)$. See Figure~\ref{fi:gridGraphDistance}. We also remark that whenever $G$ is a $k\times r$ grid graph, then $a_G(k,r)\leq |V(G)|-2$ (because for any grid graph embedding $f$ of $G$ and two different vertices $u,v\in V(G)$, $d(u,v)\leq |V(G)|-1$ and $d_f(u,v)\geq 1$). So, we get the following observation.

\begin{figure}
\centering
\includegraphics[width=0.5\textwidth, page=5]{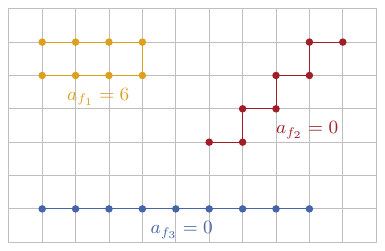}
\caption{Example of a path $P$ on $8$ vertices with three different grid graph embeddings $f_1, f_2$ and $f_3$. Since $a_{f_2} = a_{f_3} = 0$, we get that $a_P = 0$.}\label{fi:gridGraphDistance}
\end{figure}

\begin{observation} \label{obs:disapp bounded}
Let $G=(V,E)$ be a connected grid graph. Let $f$ be a $k\times r$ grid graph embedding of $G$. Then $a_f\leq |V|-2$.
\end{observation} 

Clearly, the definition of this parameter is extendible to other embeddings and other distance measures (e.g., a plane embedding and the Euclidean distance). The motivation behind the consideration of such a parameterization is as follows.  First, we remark that edges between vertices often model interactions or relations between entities. So, vertices adjacent in the graph may be required to be close to each other in the embedding (e.g., so that they can interact more efficiently). On a more general note, vertices closer to each other in the graph may be required to be closer to each other in the embedding. This also works in the opposite way---vertices more distant from each other in the graph may be required to be more distant from each other in the embedding. 

This rationale behind this parameter makes sense in various scenarios. Suppose that vertices represent utilities, factories or organizations, or, very differently, components to be placed on a chip. On the one hand, those that are closer to each other in the graph might need to cooperate more often: they have direct and indirect (through other entities on the path) connections between them; the more ``links on the chain'', the less is directed interaction required. On the other hand, we may have a competitive constraint---we may want these entities to also be ``as far as possible''. In particular, if they are far in the graph, we will take advantage of this to place them far in the embedding (proportionally). For example, these entities may cause pollution, radiation or heat~\cite{bergstra2018effect,garcia2020residential}. Alternatively, in the case of utilities, we may want to cover as large area as we can. Recently, due to the COVID-19 pandemic, many governments around the world world have introduced social distancing. Briefly, social distancing means that people should be physically away from each other, if possible. According to experts, one of the most effective ways to reduce the spread of coronavirus is social distancing~\cite{SocialDis:1,SocialDis:3,SocialDis:2}. Suppose that the vertices represent people, the edges represent social (or other) relations between them, and we want to find a seating arrangement. In order to preserve the social distancing, we would like that people who do not need to be close to each other, to be relatively far away from each other. In another example, suppose that the vertices represent some facilities that ``attract'' people, like stores. Placing the stores far away from each other, if possible, contributes to social distancing.

Besides the above scenarios, there are two more motivating arguments of different flavor. The first argument is that the problem is computationally very hard (in particular, even for trees), so we want to restrict it to get tractability, and this may be a reasonable choice. So even if we do not specifically want distance preservation, seeking those that comply with it is useful since it gives tractability. The second argument is that having such a distance preserving embedding rather than any embedding can yield more efficient algorithms due to special properties that it has. One such property is that computing distances between vertices in such an embedding can be done up to a small error ($a_G$) in constant time.

We remark that the embeddings that our algorithms compute satisfy the conditions of Definition \ref{def:Grid graph embedding}, in particular, the embeddings are planar. Furthermore, we do not need to know the value of $a_G$ in advance, in order to use our algorithm, as we iterate over all the potential values for $a_G$.

\paragraph{\bf{Integer Linear Programming.}} In the {\sc Integer Linear Programming Feasibility} (ILP) problem, the input consists of $t$ variables $x_1, x_2, \ldots, x_t$ and a set of $m$ inequalities of the following form:
	\[\begin{array}{*{9}{@{}c@{}}}
		a_{1,1}x_1 & + & a_{1,2}x_1 &+ & \cdots & + & a_{1,p}x_t & \leq & b_1 \\
		a_{2,1}x_1 & + & a_{2,2}x_2 & + & \cdots & + & a_{2,p}x_t & \leq & b_2 \\
		\vdots    &   & \vdots    &   &        &   & \vdots    &   & \vdots \\
		a_{m,1}x_1 & + & a_{m,2}x_2 & + & \cdots & + & a_{m,p}x_t & \leq & b_m \\
	\end{array}\]
where all coefficients $a_{i_j}$ and $b_i$ are required to integers. The task is to check whether there exists integer values for every variable $x_i$ so that all inequalities are satisfiable. The following theorem about the running time required to solve ILP will be useful in Section~\ref{sec:rows}.

\begin{theorem}[\cite{DBLP:journals/mor/Kannan87,DBLP:journals/mor/Lenstra83,DBLP:journals/combinatorica/FrankT87}]\label{the:runningTimeILP}
An ILP instance of size $n$ with $p$ variables can be solved in time $p^{\OO(p)}\cdot n^{\OO(1)}$. 
\end{theorem}
\section{FPT Algorithm on General Graphs}\label{sec:rows}
In this section, we show that the \bGridEm\ problem is \FPT\ parameterized by $\cc(G) + k$. We first give the definition of a {\em \recGrid}, which will be useful throughout the section.

\begin{definition} [{\bf $k\times r$ Rectangular Grid Graph}] \label{def:rGridGraph}
An undirected graph $H$ is a {\em \recGrid} if there exists a bijection $f:V \rightarrow [k] \times [r]$, such that for every pair of vertices $u, v \in V(H)$, $\{v,u\}\in E(H)$ if and only if $d_f(u,v)=1$. 
\end{definition}

Given a \recGrid\ $H$ and a corresponding bijective function $f$, we define the {\em columns} of $H$ as follows. For every $i \in [k], j \in [r]$, $C_j(H) = \{u \in V(H)| \fc(u) = j\}$. It is easy to see that the vertex set of $H$ can be denoted as the union of columns, i.e. $V(H) = \bigcup_{j=1}^r C_j(H)$. We refer to $C_1(H)$ and $C_r(H)$ as the \emph{left boundary column} and \emph{right boundary column}, respectively, of $H$. 

Given a subgraph $S$ of a $k\times r$ rectangular grid graph $H$, we denote the set of \emph{fully contained} connected components of $S$, i.e. all the connected components of $S$ that either do not intersect the boundary columns of $H$ or intersect both boundary columns of $H$, by ${\cal FC}(S)$. Similarly, we denote the set of \emph{left contained} (\emph{right contained}) connected components of $S$, i.e. all the connected components of $S$ that intersect the left (right) boundary column of $H$ but do not intersect the right (left) boundary column of $H$, by ${\cal LC}(S)$ (${\cal RC}(S)$). Note that the three sets ${\cal FC}(S), {\cal LC}(S)$ and ${\cal RC}(S)$ are pairwise disjoint and $S = {\cal FC}(S) \cup {\cal LC}(S) \cup {\cal RC}(S)$. See Figure~\ref{fi:rectangleGraph}.

\begin{figure}[!t]
\centering
\includegraphics[width=0.4\textwidth, page=2]{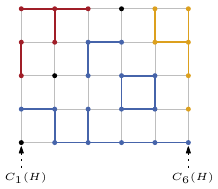}
\caption{A $5 \times 6$ rectangular grid graph $H$, its boundary columns and a subgraph $S$ of $H$ shown by colored vertices and thick colored edges. The blue, red and orange colored connected components belong to ${\cal FC}(S), {\cal LC}(S)$ and ${\cal RC}(S)$, respectively.}\label{fi:rectangleGraph}
\end{figure}

\mccK*

\begin{proof}
The \FPT\ algorithm is based on ILP. To this end, let $G$ be an instance of the \bGridEm\ problem. We first give an overview of the ideas behind the algorithm. 
\paragraph{Overview:} 
Let $H$ be a \recGrid. Let ${\cal B} = \{B_1, B_2, \ldots, B_p\}$ be a partition of $H$ into blocks of size $k \times \cc(G)$ such that $V(B_i) = \bigcup_{j = (i-1)(\cc(G) - 1) + 1}^{i(\cc(G) -1) + 1}C_j(H)$, for each $i \in [p]$ where $p = (r-1)/(\cc(G) - 1)$. We first consider the case where $r-1$ is a multiple of $\cc(G) - 1$. Note that each block $B_i$ is a $k \times \cc(G)$ rectangular grid graph, and for all $i \in [p-1], B_i$ and $B_{i+1}$ share a boundary column. See Figure~\ref{fi:blocks}.

We can restate the \bGridEm\ problem as follows: is $G$ a subgraph of $H$? As $H$ is a planar graph, $G$ must be planar. So for now, assume that $G$ is a planar subgraph of $H$. As the size of any connected component of $G$ is at most $\cc(G)$, any connected component $C$ of $G$ intersects (i) only one block $B_i$ (in particular, it does not intersect either the right or the left boundary of $B_i$), or (ii) exactly two consecutive blocks $B_i$ and $B_{i+1}$ through the right boundary column of $B_i$, or (iii) exactly three consecutive blocks $B_{i-1}, B_i$ and $B_{i+1}$ through the left and right boundary columns of $B_i$. Note that, if any connected component of $G$ intersects three consecutive blocks $B_{i-1}, B_i$ and $B_{i+1}$, then it will intersect the block $B_{i-1}$ ($B_{i+1}$) only at its right (left) boundary column. 

Based on the above observation, we compute the set $\cal S$ of all the possible {\em snapshots} of a $k \times \cc(G)$ rectangular grid graph $R$, i.e. the set of all the subgraphs of $R$, and the left and right \emph{adjacencies} between snapshots. We also find the set, denoted \source, of snapshots which may correspond to $B_1$ and the set, denoted \sink, of snapshots which may correspond to $B_p$. Note that we make this distinction, as except for blocks $B_1$ and $B_p$, all the other blocks share both boundary columns. We then make a directed graph $D$ for every possible combination of a source snapshot $start \in \source$, a sink snapshot $end \in \sink$ and a subset ${\cal S'} \subseteq {\cal S}$ of snapshots we want in our solution as follows. We add to $D$ all the snapshots in $\cal S'$ as vertices and an arc from $S$ to $S'$ if $S$ is ``left adjacent'' to $S'$, for every pair $S, S' \in \cal S'$. Note that a snapshot can be adjacent to itself, so $D$ can have loops. We add two more vertices, one for $start$ and one for $end$, and add arcs from $start$ to all its ``right adjacencies'' and from all the ``left adjacencies'' of $end$ to $end$. We then find (using ILP) the number of times each arc should be duplicated in $D$ to get a new multidigraph $D'$ on the same vertex set as $D$ such that we get a (connected) Eulerian path in $D'$ from $start$ to $end$ of length $p$ and all the connected components of $G$ are covered by the Eulerian path. Finally, we use the path to get the correspondence between the blocks of $H$ and the snapshots in $\cal S'$ with a correct placement from left to right.

Observe that, in the case where $r-1$ is not a multiple of $\cc(G) - 1$, the last block is of size $k \times \ell$, where $\ell = r - \lfloor{r/(\cc(G)-1)}\rfloor$. So, we can change the algorithm by considering \sink\ as the set of valid subgraphs of $k \times \ell$ rectangular grid graph. For sake of simplicity, we give the algorithm considering $r-1$ as a multiple of $\cc(G) - 1$.

\begin{figure}[!t]
\centering
\includegraphics[width=0.5\textwidth, page=3]{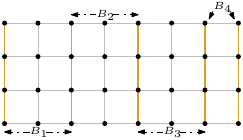}
\caption{A $k \times r$ rectangular grid graph $H$ and the corresponding blocks, where $k=4, r=8$ and $\cc(G)=3$. The boundary columns of the blocks are drawn in orange. Note that the last block is of smaller size than the rest of the blocks.}\label{fi:blocks}
\end{figure}

We now give the algorithm followed by a proof of its correctness.
\paragraph{Algorithm:}
As $G$ must be planar in case we have a yes instance, we first check if $G$ is planar in time $\OO(|V(G|)$ by any of the algorithms given in~\cite{DBLP:journals/jgaa/BoyerM04,DBLP:journals/ejc/FraysseixM12,DBLP:journals/jacm/HopcroftT74}. If the algorithm returns \no, we return \no. Otherwise, we compute the set $\C(G) = \{G_1, G_2, \ldots, G_t\}$ of all \emph{non-isomorphic} connected components of $G$ in time $\OO(|V(G)|)$ using the algorithm given by Hopcroft and Wong~\cite{DBLP:conf/stoc/HopcroftW74}. As the size of any connected component of $G$ is at most $\cc(G)$ and $G$ is planar, $t = 2^{\OO(\cc(G))}$. Let $\num(G_i)$ be the number of times $G_i$ appears in $G$, for every $i \in [t]$. We first compute the sets $\cal S, \source$ and $\sink$ (defined in the overview) in the following manner. We initialize all the three sets to the empty set. For every subgraph $S$ of $R$, if ${\cal FC}(S) \subseteq \C(G)$, then we add $S$ to $\cal S$. Note that, if there exists a connected component of $S$ in ${\cal FC}(S)$ which does not belong to $\C(G)$, then $S$ cannot contribute to a valid solution, i.e. $S$ cannot correspond to a block $B \in \cal B$. For every $S \in \cal S$, if ${\cal LC}(S) \subseteq \C(G)$, then add $S$ to \source\ and if ${\cal RC}(S) \subseteq \C(G)$, then add $S$ to \sink. Note that \source\ is the set of all possible snapshots that can correspond to $B_1$: as $B_1$ does not share its left boundary column with any other block, if there exists a connected component of $S$ in ${\cal LC}(S)$ which does not belong to $\C(G)$, then $S$ cannot contribute as a source, i.e. $S$ cannot correspond to the block $B_1$. A similar argument follows for \sink. For every snapshot $S \in \cal S$ and $i \in [t]$, let $\frc(G_i, S)$ be the number of times $G_i$ appears in ${\cal FC}(S)$. Similarly, for every snapshot $S \in \source$ ($S \in \sink$) and $i \in [t]$, let $\frl(G_i, S)$ ($\frr(G_i, S)$) be the number of times $G_i$ appears in ${\cal LC}(S)$ (${\cal RC}(S)$). Note that $\source, \sink \subseteq \cal S$. As $|E(R)| = \OO(k \cdot \cc(G))$, we have that $|{\cal S}| = 2^{\OO(k \cdot \cc(G))}$.

%
We now find the set $\ad \subseteq {\cal S} \times {\cal S}$ of all possible adjacencies between pairs of snapshots in $\cal S$ in the following manner. We initialize $\ad = \{\}$. Let $R'$ be a $k \times (2\cc(G)-1)$ rectangular grid graph. We partition $R'$ into two blocks $B'_1$ and $ B'_2$ of size $k \times \cc(G)$ such that $V(B'_1) = \bigcup_{i = 1}^{\cc(G)}C_i(R')$ and $V(B'_2) = \bigcup_{i = \cc(G)}^{2\cc(G)-1}C_i(R')$.
For every $i \in \{1,2\}$ and subgraph $S'$ of $R'$, let $S'_i = S'[V(S') \cap V(B'_i)]$ be the subgraph of $S'$ in block $B'_i$. We look at only those subgraphs $S'$ for which both $S'_1$ and $S'_2$ belong to $\cal S$. For every such $S'$, we add the pair $(S'_1, S'_2)$ to $\ad$ if all the connected components of $S'_1 \cup S'_2 = S'$ that intersect both $B'_1$ and $B'_2$ (i.e., intersect column $C_{\cc(G)}(R')$) belong to $\C(G)$. Let ${\cal BC}(S'_1, S'_2)$ be the set of \emph{boundary intersecting} connected components of the pair $(S'_1, S'_2)$, that is, all the connected components of $S'_1 \cup S'_2 = S'$ that intersect the column $C_{\cc(G)}(R')$ but intersect neither $C_1(R')$ nor $C_{2\cc(G)-1}(R')$. For every $i \in [t]$, we denote the number of times $G_i$ appears in ${\cal BC}(S'_1, S'_2)$, by $\frb(G_i, (S'_1, S'_2))$. As we consider every subgraph of $R'$, we get all the possible adjacencies between pairs of snapshots in $\cal S$. Note that ${\cal BC}(S'_1, S'_2)$ intersect neither ${\cal FC}(S'_1)$ nor ${\cal FC}(S'_2)$ but it may intersect ${\cal RC}(S'_1)$ or ${\cal LC}(S'_2)$. See Figure~\ref{fi:boundarySharing}.

\begin{figure}[!t]
\centering
\includegraphics[width=0.5\textwidth, page=4]{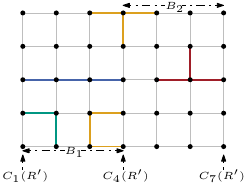}
\caption{A $k \times (2\cc(G)-1)$ rectangular grid graph $R'$, its corresponding blocks and a subgraph $S'$ of $R'$ shown by colored vertices and thick colored edges, where $k=5$ and $\cc(G)=4$. The orange colored connected components belong to ${\cal BC}(S'_1, S'_2)$.}\label{fi:boundarySharing}
\end{figure}

For every pair of snapshots $(start, end)$ such that $start \in \source$ and $end \in \sink$ and a set ${\cal S'} \subseteq {\cal S}$ of snapshots, we create a directed graph $D = D(start, end, {\cal S'})$ as follows. We add all the snapshots in $\cal S'$ as vertices of $D$ and for every pair of snapshots $S, S' \in \cal S'$, if $(S, S') \in \ad$, then add an arc from $S$ to $S'$ in $D$. We then add both $start$ and $end$ as vertices of $D$ and for every snapshot $S \in \cal S$, if $(start, S) \in \ad$, add an arc from $start$ to $S$ in $D$. Similarly, for every snapshot $S \in \cal S'$, if $(S, end) \in \ad$, add an arc from $S$ to $end$ in $D$. For every arc $(S, S') \in E(D)$, let $X(S,S')$ be a variable that corresponds to the number of times the arc $(S,S')$ is duplicated to get the multidigraph $D'$ mentioned in the overview. Then, for the directed graph $D$, the algorithm proceeds as follows.

\begin{itemize}
	\item Find the set $\cal T$ of all spanning trees of the underlying undirected graph of $D$.
	\item For every spanning tree $T \in \cal T$, solve the following ILP to find $X(S,S')$ for every edge $(S,S') \in E(D)$.
		\begin{subequations}
		\begin{gather}
	        \forall S \in V(D) \setminus \{start, end\}: \sum_{(S,S') \in E(D)} X(S,S') = \sum_{(S'',S) \in E(D)} X(S'',S). \label{eq:InDegOutDeg}\\
	        \sum_{(start,S) \in E(D)} X(start,S) = 1. \label{eq:start}\\
	        \sum_{(S,end) \in E(D)} X(S,end) = 1. \label{eq:end}\\
	        \sum_{(S,S') \in E(D)} X(S,S') = p-1. \label{eq:PathLength}\\
	        \forall i \in [t]: \frl(G_i, start) + \sum_{(S,S') \in E(D)} X(S,S') \cdot \frb(G_i, (S,S')) + \nonumber\\
			\sum_{S \in V(D)}\bigg(\sum_{(S,S') \in E(D)} X(S,S')\bigg) \cdot \frc(G_i, S) + \frr(G_i, end) = \num(G_i). \label{eq:countingCC}\\
			\forall (S,S') \in E(T): X(S,S') \geq 1. \label{eq:connectedD}\\
			\forall (S,S') \in E(D) \setminus E(T): X(S,S') \geq 0. \label{eq:normalEdges}
		\end{gather}
		\end{subequations}
		\begin{itemize}
			\item If the ILP returns a feasible solution, then return \yes.
		\end{itemize}
\end{itemize}
Recall that we run the algorithm for every possible $D$. If none of them returns \yes, we return \no.
\paragraph{Correctness:} We start by analyzing the equations. Equation~\ref{eq:connectedD} ensures that the digraph $D'$ is connected, and, in this context, recall that we go over all the possible spanning trees to check all the different possible connectivities between the vertices of $D'$. Equations~\ref{eq:InDegOutDeg},~\ref{eq:start} and~\ref{eq:end} ensure that there exists an Eulerian path in $D$ from $start$ to $end$. Equation~\ref{eq:PathLength} ensures that the total number of edges in $D'$ is $p-1$, which in turn means that the Eulerian path from $start$ to $end$ in $D'$ is of length $p$, which is equal to the number of required blocks. Given a multidigraph $D'$, each connected component of $G$ can contribute to only one set out of ${\cal LC}(start), {\cal RC}(end), {\cal FC}(S)$ and ${\cal B}(S', S'')$, for $S, S', S'' \in {\cal S'}$ such that $(S', S'') \in E(D)$, as there exists no $S \in \cal S'$ such that $(S, start) \in E(D)$ or $(end, S) \in E(D)$. So, Equation~\ref{eq:countingCC} ensures that all the connected components of $G$ are covered by the path exactly once.

Next, we prove that the algorithm is correct. In one direction, assume that the algorithm returns \yes. This means that for some directed graph $D$, ILP assigned integer values for the variables $X(S,S')$, $(S,S') \in E(D)$ such that the corresponding multidigraph $D'$ has an Eulerian path from $start$ to $end$ of size $p$ covering all the connected components of $G$ exactly once. Let $P= (S_1= start, S_2, \ldots, S_{p-1}, S_p=end)$ be the Eulerian path from $start$ to $end$ in $D'$, where $S_i \in {\cal S}'$ for every $i \in [p]$. Then, we can define $V(G) \cap V(B_i) = V(S_i)$, for every $i \in [p]$. Observe that, $G = \bigcup_{i \in [p]}(G \cap B_i)$. Thus $G$ is a subgraph of $H$, i.e. $G$ is a $k \times r$ grid graph.

Conversely, let $G$ be a $k \times r$ grid graph, i.e.~$G$ is a subgraph of $H$. For every $i \in [p]$, let ${\cal S}_H = \{H_1, H_2, \ldots, H_p\}$ be the set of graphs $H_i = G[V(G) \cap V(B_i)]$. For every $i \in [p]$ and $ j \in [p-1]$, observe that $H_i \in \cal S$ and $(H_j, H_{j+1}) \in \ad$. Now, let $P=(v_1, v_2, \ldots, v_p)$ be a directed path where $v_i$ is a vertex corresponding to the graph $H_i$. We create a multidigraph from $P$ by repeating the following procedure. If there exist two vertices $v, v' \in V(P) \setminus \{v_1, v_p\}$ such that the corresponding graphs in ${\cal S}_H$ are isomorphic, then do a $\join(v, v')$ on $P$. Let $D'_H$ be the multidigraph obtained after the above procedure. Note that $D'_H$ has an Eulerian path from $v_1$ to $v_p$ of length $v_p$. Let $D_H$ be the multigraph obtained by $D'_H$ be removing multiple arcs between any pair of vertices by a single directed edge. For every $e \in E(D)$, let $X_H(e)$ be the number of times the arc $e$ appears in $D'_H$. Thus, the algorithm return will return \yes\ for $D_H$ and the corresponding $X_H(e)$ values for every $e \in E(D)$. 

For a given directed graph $D$, $|{\cal T}| = |V(D)|^{|V(D)-2|}$ and number of variables $X(e)$, for every $e \in E(D)$, is $\OO(|V(D)|^2)$. As ${\cal S} = 2^{\OO(k \cdot \cc(G))}$, number of different directed graphs $D$ is $2^{2^{\OO(k \cdot \cc(G))}}$ and $|V(D)| = 2^{\OO(k \cdot \cc(G))}$ for any directed graph $D$. So, by Theorem~\ref{the:runningTimeILP}, the \bGridEm\ problem is \FPT\ parameterized by $\cc(G) + k$.
\end{proof}

The following claim about {\sc 2-Strip Packing} will follow from the Theorem~\ref{thm:mccK}, as we prove below.

\stripPacking*

\begin{proof}
Without loss of generality, we can assume that there does not exist any input rectangle of size $1 \times t$, where $t \leq \ell$, as otherwise we can get an equivalent instance of {\sc 2-Strip Packing} by multiplying each of the dimensions of the input rectangles and of the strip by $2$. For each input rectangle, we create a $\len \times \bre$ rectangular grid graph, where $\len$ and $\bre$ are the dimensions of the input rectangle such that $\len \leq \bre$. Let $G$ be the disjoint union of the graphs corresponding to the input rectangles. Observe that $\cc(G) = \ell^2$, as the size of any graph corresponding to a input rectangle is at most $\ell^2$. So, by Theorem~\ref{thm:mccK}, {\sc 2-Strip Packing} is \FPT\ parameterized by $\ell+k$.
\end{proof}

Note that when $k$ and $r$ are unrestricted, we can embed each connected component individually in \FPT\ time (e.g., by using brute-force), so we directly get the following observation from the Theorem~\ref{thm:mccK}.

\mccUnrestricted*

Given a graph $G$, if $G$ is a ($k \times r$) grid graph then $\Delta(G) \leq 4$. So, by the Observation~\ref{obs:treedepthBounded} and the Theorem~\ref{thm:mccK}, we get the following corollary.

\treedepthK*

\section{Distance Approximation Parameter}\label{sec:distance}

In this section, we consider the distance approximation parameter and prove Theorems~\ref{thm:distanceHard}, ~\ref{thm:distanceFPT} and~\ref{thm:distanceTrees}. We remark that the embeddings that our algorithms compute satisfy the conditions of Definition \ref{def:Grid graph embedding}, in particular, the embeddings are planar (being grid graphs).

\subsection{\textsf{Para-NP}-hardness with Respect to $a_G$ on General Graphs}

We show a reduction from {\sc SAT} to \gridEm\ where if the output is a yes-instance, then the parameter $a_G$ of the output graph is upper bounded by a constant. In order to make the reduction clearer, instead of presenting a direct reduction from {\sc SAT} to \gridEm, we present a two-stage reduction. To this end, we define a simple problem called the {\sc Batteries} problem. We first give a reduction from {\sc SAT} to the {\sc Batteries} problem, and then we give a reduction from the {\sc Batteries} problem to \gridEm. 

We start with the description of the {\sc Batteries} problem. A {\em battery} has two sides, a {\em positive side} and a {\em negative side}, denoted by {\em +} and {\em -}, respectively. Each side of a battery has {\em voltage} of one or zero volts. Formally, a battery is represented by a boolean pair, defined as follows. See Figure \ref{fi:battery} for an illustration.

\begin{figure}[!t]
\centering
\includegraphics[width=0.2\textwidth, page=1]{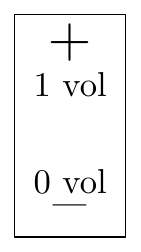}
\caption{A $(1,0)$ battery.}\label{fi:battery}
\end{figure}
             
\begin{definition}[{\bf Battery}] \label{def:battery}
A {\em battery} $B$ is a boolean pair $B=(x_1,x_2), x_1,x_2\in \{0,1\}$, where $x_1$ is the {\em voltage of the positive side} of $B$ (called the {\em positive voltage}) and $x_2$ is the {\em voltage of the negative side} of $B$ (called the {\em negative voltage}).  
\end{definition}

Intuitively, for two positive integers $r,c\in \mathbb{N}$, an {\em $(r,c)$ battery holder} is a ``device'' that holds $r\cdot c$ batteries in ``matrix-like'' cells. Batteries are laid vertically in the battery holder that, if a battery $B$ is laid in cell $(i,j)$ then there are two options: either its $+$ side is on top or its $-$ side is on top. For every $1\leq j \leq c$, the batteries that are laid in the $j$-th column of the battery holder are connected top to bottom. Moreover, for every $1\leq i \leq r$, there is a wire that connects all the top sides of the batteries that are laid in the $i$-th row of the battery holder. That wire transfers only the amounts of voltage that are on top (see Figure \ref{fi:batteriesHolder}). 

\begin{figure}[!t]
\centering
\includegraphics[width=0.3\textwidth, page=2]{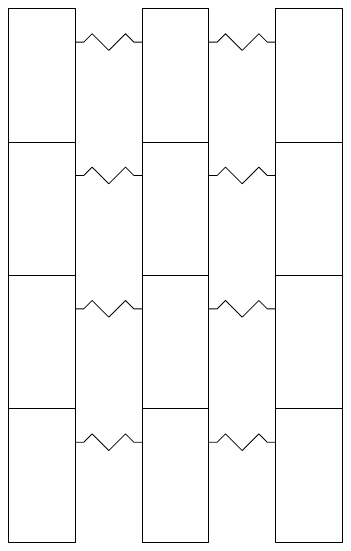}
\caption{A $(4,3)$ battery holder.}\label{fi:batteriesHolder}
\end{figure}

Given a set of batteries $\{ B_{i,j}~|~1\leq i\leq r, 1\leq j\leq c\}$, we would like to place all the batteries in an $(r,c)$ battery holder such that battery $B_{i,j}$ is placed in cell $(i,j)$. Therefore, for each battery $B_{i,j}$, we need to decide if we want to place it with its $+$ side on top or with its $-$ side on top when we place it in cell $(i,j)$. Formally, we describe our choice of the ``direction'' of the batteries, or the ``placement'', by a function $p:[c]\times [r]\to\ \{+,-\}$: if $p(i,j)=+$, then it means that battery $B_{i,j}$ is placed in cell $(i,j)$ with its $+$ side on top, and if $p(i,j)=-$, then it means that battery $B_{i,j}$ is placed in cell $(i,j)$ with its $-$ side on top. We define this function formally in the next definition.

\begin{definition}[{\bf $(r,c)$-Placement}] \label{def:direction}
Let $r,c\in \mathbb{N}$ be two positive integers. An {\em $(r,c)$-placement} is a function $p:[r]\times [c]\to\ \{+,-\}$.
\end{definition}             

There are some restrictions for the placements of the batteries. First, we want them to be placed {\em correctly}, in the sense that $+$ can only be connected to $-$ and vice versa, for every column of batteries. Now, assume that $p$ is correct. Then, for a column $1\leq j\leq c$ and for any two batteries $B_{i,j}$ and $B_{i+1,j}$ that are connected in that column for $1\leq i\leq r-1$, it follows that if the $+$ side of $B_{i+1,j}$ is on top then the $-$ side of $B_{i,j}$ is on bottom, and therefore the $+$ side of $B_{i,j}$ is on top. Similarly, if the $-$ side of $B_{i+1,j}$ is on top, then the $+$ side of $B_{i,j}$ is on bottom, and therefore the $-$ side of $B_{i,j}$ is on top. Thus, we get that $p(i,j)=p(i+1,j)$. Formally, this restriction is defined as follows.

\begin{definition}[{\bf Correct $(r,c)$-Placement}] \label{def:Correctdirection}
An $(r,c)$-placement $p:[r]\times [c]\to\ \{+,-\}$ is {\em correct} if for every $1\leq j\leq c$ and $1\leq i\leq r-1$, $p(i,j)=p(i+1,j)$.     
\end{definition} 

Another restriction is that we disallow having too much voltage transferred in the same row. So, for every row in the battery holder, we would like the sum of voltage transferred by the wire of that row to be at most $c-1$. For a set of batteries $\{ B_{i,j}~|~1\leq i\leq r, 1\leq j\leq c\}$ and an $(r,c)$-placement $p$, we denote the voltage transferred from battery $B_{i,j}$ placed in cell $(i,j)$ by $V_p(i,j)$. Therefore, it follows that if $B_{i,j}=(x_1^{(i,j)},x_2^{(i,j)})$, then if $p(i,j)= +$ then $V_p(i,j)=x_1^{(i,j)}$, and if $p(i,j)= -$ then $V_p(i,j)=x_2^{(i,j)}$. We define this restriction as follows.

\begin{definition}[{\bf Safe $(r,c)$-Placement}] \label{def:Safedirection}
Let $B=\{ B_{i,j}=(x_1^{(i,j)},x_2^{(i,j)})~|~1\leq i\leq r, 1\leq j\leq c\}$ be a set of batteries. An $(r,c)$-placement. $p:[r]\times [c]\to\ \{+,-\}$ is {\em safe with respect to $B$} if for every $1\leq i\leq r$ it follows that $\sum_{j=1}^{c}{V_p(i,j)}\leq c-1$. 
\end{definition} 

When $B$ is clear from the context, we refer to a safe placement with respect to $B$ simply as a safe placement.

We define the {\sc Batteries} problem in the next definition.

\begin{definition}[{\bf The Batteries Problem}] \label{def:BatProb}
Given a set of batteries $B=\{ B_{i,j}=(x_1^{(i,j)},x_2^{(i,j)})~|~1\leq i\leq r, 1\leq j\leq c\}$, does there exist an $(r,c)$-placement $p$ that is correct and safe?
\end{definition} 

\medskip \noindent{\bf Reduction from SAT to Batteries.} We now give a polynomial reduction, called $\mathsf{reduce}_1$, from {\sc SAT} to the {\sc Batteries} problem.

Given an instance $\pi$ of {\sc SAT} with $n$ variables $x_1,\ldots, x_n$ and $m$ clauses $\mu_1,\ldots,\mu_m$, $\mathsf{reduce}_1(\pi)= B_\pi$ where $B_\pi$ is a set of batteries $B_\pi=\{ B_{i,j}=(x_1^{(i,j)},x_2^{(i,j)})~|~1\leq i\leq m, 1\leq j\leq n\}$ defined as follows. For every $1\leq i\leq m, 1\leq j\leq n$, we set $x_1^{(i,j)}=0$ if the literal $x_j$ appears in $\mu_i$, and we set $x_1^{(i,j)}=1$ otherwise. Similarly, for every $1\leq i\leq m, 1\leq j\leq n$, we set $x_2^{(i,j)}=0$ if the literal $\bar{x}_j$ appears in $\mu_i$, and we set $x_2^{(i,j)}=1$ otherwise.

It is clear that this reduction is polynomial, as we state in the next observation.

\begin{observation} \label{obs:polytime}
Let $\pi$ be an instance of {\sc SAT}. Then the function $\mathsf{reduce}_1$ on $\pi$ can be computed in time polynomial in $|\pi|$.
\end{observation}   

Towards the proof of the correctness of this reduction, we have the next simple lemma. We show that if an $(r,c)$-placement $p$ is correct, then all the batteries in every column are placed in the same direction. Note that the converse side, that is, that if all the batteries in every column are placed in the same direction, then $p$ is correct, is trivial.  

\begin{lemma}\label{lem:oneDirection}
An $(r,c)$-placement $p$ is correct if and only if for every $1\leq j\leq c$ there exists $p_j\in \{+,-\}$ such that for every $1\leq i\leq r$ it follows that $p(i,j)=p_j$.
\end{lemma}    

\begin{proof}
Let $p$ be a correct $(r,c)$-placement. Let $1\leq j\leq c$. We set $p_j=p(1,j)$. We show by induction that for every $1\leq i\leq r$ it follows that $p(i,j)=p_j$. The basic case is trivial since $p_j=p(1,j)$. Let $1< i\leq r$. By the inductive hypothesis, it follows that $p(i-1,j)=p_j$. Since $p$ is correct it follows that $p(i-1,j)=p(i,j)$. Therefore, we get that $p(i,j)=p_j$. The other direction of the lemma is trivial.        
\end{proof} 

From Lemma \ref{lem:oneDirection} we get that if $p$ is correct, then for every column of the battery holder, we have that all the batteries are placed with either their $+$ on top or their $-$ on top. For every $1\leq j\leq c$, we refer to $p_j$ from Lemma \ref{lem:oneDirection} as the {\em placement of column $j$}. The placement of every column $1\leq j\leq n$ corresponds to a placement for $x_j$, that is, $p_j=+$ corresponds to $x_j=T$ and $p_j=-$ corresponds to $x_j=F$. Every row $1\leq i\leq m$ of batteries in the battery holder corresponds to clause $\mu_i$ from $\pi$. In order to get a safe placement, we need at least one battery with zero voltage on its top side. This corresponds to having at least one literal that is satisfied in each clause of $\pi$. See Figure \ref{fi:red1}.

\begin{figure}[!t]
\centering
\includegraphics[width=0.35\textwidth, page=3]{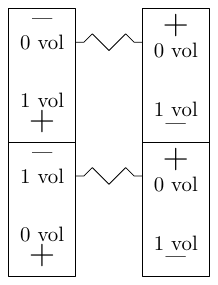}
\caption{Example of a safe and correct $(2,2)$-placement for $\mathsf{reduce}_1(\pi)=B_\pi$ where $\pi = (\bar{x}_1 \vee x_2) \wedge (x_1 \vee x_2)$. This placement corresponds to $x_1=F, x_2=T$.}\label{fi:red1}
\end{figure}

In the next lemma we prove the correctness of the reduction.

\begin{lemma} \label{lem:redproof}
Let $\pi$ be an instance of {\sc SAT} with $n$ variables $x_1,\ldots, x_n$ and $m$ clauses $\mu_1,\ldots,\mu_m$. Then, $\pi$ is a yes-instance of {\sc SAT} if and only if $\mathsf{reduce}_1(\pi)= B_\pi$ is a yes-instance of the {\sc Batteries} problem. 
\end{lemma}             
 
\begin{proof}
In the forward direction, assume that $\pi$ is a yes-instance of {\sc SAT}. Let $s:\{x_1,\ldots,x_n\}\to\ \{T,F\}$ be a satisfying assignment for $\pi$. We define an $(m,n)$-placement $p$ as follows. For every $1\leq i\leq m$ and $1\leq j\leq n$, we set $p(i,j)=+$ if $s(x_j)=T$; otherwise, we set $p(i,j)=-$. We show that $p$ is correct and safe. Notice that for every $1\leq j\leq n$ and $1\leq i,i'\leq m$, it follows that $p(i,j)=p(i',j)$. Therefore, by Lemma \ref{lem:oneDirection} we get that $p$ is correct. Now, let $1\leq i\leq m$. As $s$ is a satisfying assignment for $\pi$, there exists $1\leq j\leq n$ such that $x_j$ is satisfied in $\mu_j$. Therefore, if $s(x_j)=T$, then the literal $x_j$ appears in $\mu_j$, and if $s(x_j)=F$, then the literal $\bar{x}_j$ appears in $\mu_j$. If $s(x_j)=T$, then $p(i,j)=+$ and since the literal $x_j$ appears in $\mu_j$, by the definition of $B_\pi$, we get that $x_1^{(i,j)}=0$. Therefore it follows that $V_p(i,j)=0$, so $\sum_{k=1}^{n}{V_p(i,k)}= \sum_{k=1,k\neq j}^{n}{V_p(i,k)}+V_p(i,j)=\sum_{k=1,k\neq j}^{n}{V_p(i,k)}\leq n-1$. As the choice of $i$ was arbitrary, we get that $p$ is safe. So, we found a correct and safe placement for $B_\pi$, and therefore $B_\pi$ is a yes-instance of the {\sc Batteries} problem.

In the reverse direction, assume that $B_\pi=\{ B_{i,j}=(x_1^{(i,j)},x_2^{(i,j)})~|~1\leq i\leq m, 1\leq j\leq n\}$ is a yes-instance of the {\sc Batteries} problem. Let $p$ be an $(m,n)$-placement such that $p$ is correct and safe. From Lemma \ref{lem:oneDirection}, we get that for every $1\leq j\leq n$ there exists $p_j\in \{+,-\}$ such that for every $1\leq i\leq m$ it follows that $p(i,j)=p_j$. We define an assignment $s:\{x_1,\ldots,x_n\}\to\ \{T,F\}$ as follows. For $1\leq j\leq n$ we set $s(x_j)=T$ if $p_j=+$; otherwise, we set $s(x_j)=F$. We show that $s$ is a satisfying assignment for $\pi$. Let $1\leq i\leq m$. Since $p$ is safe, we get that $\sum_{k=1}^{n}{V_p(i,k)}\leq n-1$. Therefore, there exists $1\leq j\leq n$ such that $V_p(i,j)=0$. If $p(i,j)=p_j=+$, then $x_1^{(i,j)}=0$ and by the definition of $B_\pi$ we get that $x_j$ appears in $\mu_i$. By the definition of $s$, since $p_j=+$, we get that $s(x_j)=T$, therefore $\mu_i$ is satisfied. Similarly, if $p(i,j)=p_j=-$, then $x_2^{(i,j)}=0$ and by the definition of $B_\pi$ we get that $\bar{x}_j$ appears in $\mu_i$. By the definition of $s$, since $p_j=-$, we get that $s(x_j)=F$, therefore $\mu_i$ is satisfied. As $s$ is a satisfying assignment for $\pi$, $\pi$ is a yes-instance of {\sc SAT}.                     
\end{proof}

Combining Lemma \ref{lem:redproof} and Observation \ref{obs:polytime}, we conclude the existence of a polynomial reduction from {\sc SAT} to the {\sc Batteries} problem in the next lemma.

\begin{corollary}\label{BatisHard}
There exists a polynomial reduction from {\sc SAT} to the {\sc Batteries} problem. 
\end{corollary}

\medskip \noindent{\bf Reduction from Batteries to Grid Embedding.} We now give a polynomial reduction from the {\sc Batteries} problem to \gridEm\ where the output graph satisfies that $a_G$ is bounded by a fixed constant (if it is a yes-instance). \footnote{So, if $a_G$ is not bounded by that constant, we can output a trivial no-instance where it is bounded by that constant and hence ensure that $a_G$ is always bounded by that constant.}

For this purpose, we present the {\em battery gadget} (see Figures \ref{fi:batgat} and \ref{fi:batgatexem}). The battery gadget is composed of a $13\times 9$ rectangle, which corresponds to a cell in the battery holder. It has a {\em positive side} and a {\em negative side}, which correspond to battery sides, and two {\em wire vertices} that correspond to the wire that transfers voltage in the battery holder. In addition, there are six {\em synchronization edges} attached to the top and bottom sides of the rectangle. As we will see, the synchronization edges make every two vertically adjacent battery gadget ``synchronized'', similarly to the $+$ and $-$ sides of a battery. 
On the top and bottom sides of the gadget we have the option to add an extra edge, called the {\em positive voltage} and the {\em negative voltage}, respectively; see Figure \ref{fi:batgat}. These edges correspond to the voltage of each side of the battery, that is, we add the positive (negative) voltage if and only if the voltage of the positive (negative) side of the battery is one.  We denote the battery gadget by $H=(x_1,x_2), x_1,x_2\in \{0,1\}$, where $x_1=1$ if and only if we added the positive voltage, and $x_2=1$ if and only if we added the negative voltage. Observe that $H=(x_1,x_2)$ corresponds to the battery $B=(x_1,x_2)$, as exemplified in Figure \ref{fi:batgatexem}. 

\begin{figure}[!t]
\centering
\includegraphics[width=0.5\textwidth, page=4]{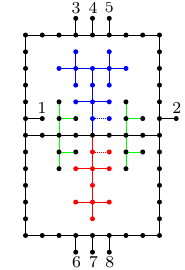}
\caption{The battery gadget. The positive side is in blue and the negative side is in red. The wire is green. The positive voltage is in dashed blue and the negative voltage is in dashed red.}\label{fi:batgat}
\end{figure}  

\begin{figure}[!t]
\centering
\includegraphics[width=0.5\textwidth, page=5]{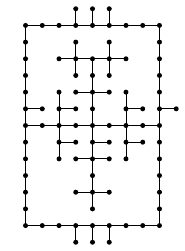}
\caption{A $(1,0)$ battery gadget.}\label{fi:batgatexem}
\end{figure}

Next, we define a graph called {\em $m\times n$-grid frame} (see Figure \ref{fi:gridFrame}).

\begin{definition}[{\bf An $m\times n$-Grid Frame}] \label{def:gridFrame}
Let $m,n\in \mathbb{N}$. An {\em $m\times n$-grid frame} is the graph $G_{m,n}=(V_{m,n},E_{m,n})$ where $V_{m,n}=\{(k,\ell)~|~0\leq k\leq 2, 0\leq \ell \leq 8\cdot n+4\}\cup \{(k,\ell)~|~12\cdot m+2 \leq k\leq 12\cdot m+4, 0\leq \ell \leq 8\cdot n+4\}\cup \{(k,\ell)~|~0 \leq k\leq 12\cdot m+4, 0\leq \ell \leq 2\} \cup \{(k,\ell)~|~0 \leq k\leq 12\cdot m+4, 8\cdot n+2\leq \ell \leq 8\cdot n+4\}$ and $E_{m,n}=\{ \{(r,c),(r',c')\}~|~(r,c),(r',c')\in V_{m,n}$ and $|r-r'|+|c-c'|=1\}$.      
\end{definition}

 \begin{figure}[!t]
\centering
\includegraphics[width=0.5\textwidth, page=6]{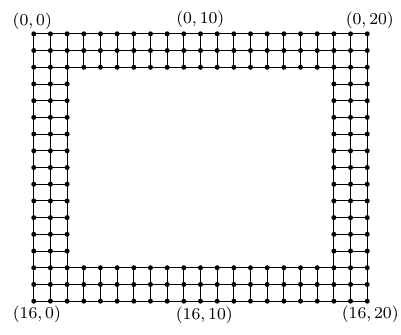}
\caption{A $(1\times 2)$-grid frame.}\label{fi:gridFrame}
\end{figure}

Let $m,n\in \mathbb{N}$ be two positive integers. Given a set of battery gadgets $H=\{H_{i,j}=(x_1^{(i,j)},x_2^{(i,j)})~|~1\leq i\leq m, 1\leq j\leq n\}$, we denote by $G_H$ (when $H$ is clear from the context we refer to $G_H$ simply as $G$) the following graph. The graph $G$ is composed of $m\cdot n$ battery gadgets that are ordered in a ``matrix shape'', i.e. the battery gadget $H_{i,j}$ is located at the $i$-th ``row'' and the $j$-th ``column'' (see Figures \ref{fi:redstru} and \ref{fi:redexa}). For every $1\leq i<m$ the battery gadgets $H_{i,j}$ and $H_{i+1,j}$ share one side of their rectangles, that is, the bottom side of the rectangle of $H_{i,j}$ is the top side of the rectangle of $H_{i+1,j}$. Similarly, for every $1\leq j<n$ the battery gadgets $H_{i,j}$ and $H_{i,j+1}$ share one side of their rectangles, that is, the right side of the rectangle of $H_{i,j}$ is the left side of the rectangle of $H_{i,j+1}$. In addition, the ``matrix'' of battery gadgets is encircled by an $m\times n$-grid frame (see Figures \ref{fi:redstru} and \ref{fi:redexa}). Lastly, we delete some of the top and bottom synchronization edges we do not need, i.e. those which are not attached to an side that is shared by two rectangles. More precisely, for every $1\leq j\leq n$ we delete the three edges attached to the top side of the rectangle of $B_{1,j}$ and the three edges attached to the bottom side of the rectangle of  $B_{m,j}$. Observe that each synchronization vertex is common to two battery gadgets. That is, for every $1\leq i<m$ and $1\leq j\leq n$, the synchronization vertices $3,4,5$ (see Figure \ref{fi:batgat})of the battery gadget $H_{i+1,j}$ are the synchronization vertices $6,7,8$ of $H_{i,j}$, respectively. Therefore, when we consider the graph $G_H$ and not only an arbitrary battery gadget, we denote these synchronization vertices by $S_{i,j}(1),S_{i,j}(2),S_{i,j}(3)$ for every $1\leq i<m$ and $1\leq j\leq n$. Similarly, when we consider the graph $G_H$ and not only an arbitrary battery gadget, we denote the wire vertices $1$ and $2$ by $W(i,j)$ for every $1\leq i\leq m$ and $0\leq j\leq n$ (see Figure \ref{fi:redstrunot}).   

\begin{figure}[!t]
\centering
\includegraphics[width=0.5\textwidth, page=15]{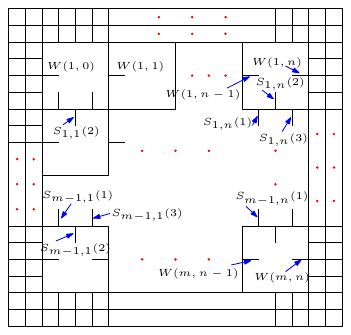}
\caption{Schematic illustration of the graph $G_B$ with notations for the wire vertices and the synchronization vertices.}\label{fi:redstrunot}
\end{figure}

\begin{figure}[!t]
\centering
\includegraphics[width=0.5\textwidth, page=7]{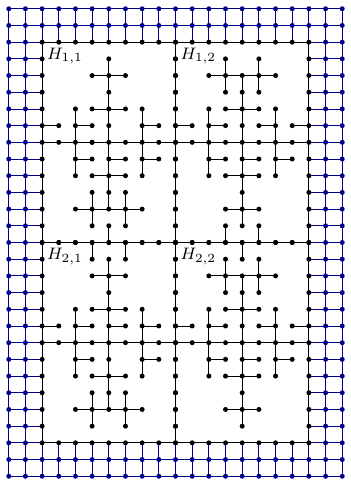}
\caption{Construction of the graph $G_B$ where $B=\{ B_{1,1}=(1,0),B_{1,2}=(0,1),B_{2,1}=(0,1),B_{2,2}=(0,1)\}$.}\label{fi:redexa}
\end{figure}      

Given an instance of the {\sc Batteries} problem $B=\{B_{i,j}=(x_1^{(i,j)},x_2^{(i,j)})~|~1\leq i\leq m, 1\leq j\leq n\}$, we denote by $\mathcal{H}_B$ the {\em corresponding set of battery gadgets of $B$}, that is $\mathcal{H}_B=\{H_{i,j}=(x_1^{(i,j)},x_2^{(i,j)})~|~1\leq i\leq m, 1\leq j\leq n\}$.  We are ready to present our reduction function $\mathsf{reduce}_2$: $\mathsf{reduce}_2(B)=G_{\mathcal{H}_B}$. See Figures \ref{fi:redstru} and \ref{fi:redexa} for illustration. To simplify the notation, we denote the graph $G_{\mathcal{H}_B}$ by $G_B$.   

\begin{figure}[!t]
\centering
\includegraphics[width=0.5\textwidth, page=8]{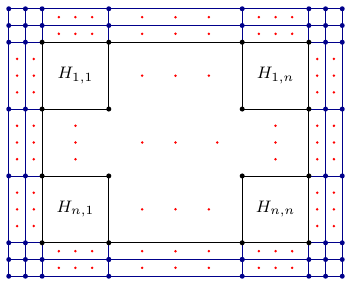}
\caption{Schematic illustration of the graph $G_B$. The grid frame is shown in blue.}\label{fi:redstru}
\end{figure}


It is easy to see that this reduction is polynomial, as we state in the next observation.

\begin{observation} \label{obs:polytimered2}
Let $B$ be an instance of the {\sc Batteries} problem. Then, the function $\mathsf{reduce}_2$ works in time polynomial in $|B|$.
\end{observation} 

Now we turn to prove the correctness of the reduction. First, we distinguish between some ``parts'' of the graph that have a fixed embedding to other parts that might have more than one way to be embedded. First, we show that the embedding of the $m\times n$-grid frame is ``almost fixed''. Formally, in the following lemma we show that the embedding is fixed once we choose an embedding of three vertices. Intuitively, the ``shape'' of the $m\times n$-grid frame is fixed in every embedding, but we may ``rotate'' the frame $90$, $180$ or $270$ degrees, or ``move'' it to another ``location'', but this does not matter for our purposes. 

\begin{lemma}\label{lem:fixedFrame}
Let $m,n\in \mathbb{N}$ be two positive integers. Let $f$ be a grid graph embedding of an $m\times n$-grid frame $G_{m,n}=(V_{m,n},E_{m,n})$. If $f((0,0))=(0,0)$ $f((1,0))=(1,0)$ and $f((1,1))=(1,1)$, then for every $(r,c)\in V_{m,n}$ it follows that $f((r,c))=(r,c)$.  
\end{lemma}

\begin{proof}
First we prove the lemma for the ``top edge'' of the $m\times n$-grid frame, that is $\{(k,\ell)~|~0\leq k\leq 2, 0\leq \ell \leq 8\cdot n+4\}\subseteq V_{m,n}$. We show by induction on $j$ that for every $1\leq j \leq 8\cdot n+4$ it follows that: for every $0\leq \ell \leq j$ and $0\leq k\leq 2$, we have that $f((\ell,k))=(\ell,k)$.

The base case is where $j=1$. Observe that $\{(0,1),(1,1)\},\{(0,1),(0,0)\} \in E_{m,n}$, therefore $d_f((0,1),(1,1))=|\fr(0,1)-\fr(1,1)|+|\fc(0,1)-\fc(1,1)|=|\fr(0,1)-1|+|\fr(0,1)-1|=1$, and also $d_f((0,1),(0,0))=|\fr(0,1)-\fr(0,0)|+|\fc(0,1)-\fc(1,0)|=|\fr(0,1)-0|+|\fr(0,1)-0|=1$. Therefore, we get that $\fr(0,1)=0$ and $\fc(0,1)=1$, so $f((0,1))=(0,1)$. Now, since $\{(2,0),(1,0)\},\{(2,0),(2,1)\},\{(2,1),(1,1)\}\in E_{m,n}$ and since $f((0,0))=(0,0)$ and $f$ is an injection, then it must be that $f((2,0))=(2,0)$ and $f((2,1))=(2,1)$. This proves the base case.

Now, let $1< j \leq 8\cdot n+4$. From the inductive hypothesis it follows that for every $0\leq \ell \leq j-1$ and $0\leq k\leq 2$, we get that $f((\ell,k))=(\ell,k)$. Since $\{(1,j),(1,j-1)\}\in E_{m,n}$ and $f((1,j-1))=(1,j-1)$ then it follows that $f((1,j))=(1,j)$ or $f((1,j))=(0,j-1)$ or $f((1,j))=(2,j-1)$ or $f((1,j))=(1,j-2)$. From the inductive hypothesis we get that $f((0,j-1))=(0,j-1)$, $f((1,j))=(1,j)$ and $f((1,j))=(1,j)$. Therefore, since $f$ is an injection, we get that $f((1,j))=(1,j)$. Now, since $\{(0,j),(1,j)\},\{(0,j),(0,j-1)\} \in E_{m,n}$ and $f((1,j))=(1,j),f((0,j-1))=(0,j-1)$, then we get that  $f((0,j))=(0,j)$. Similarly, since $\{(2,j),(1,j)\},\{(2,j),(2,j-1)\} \in E_{m,n}$ and $f((1,j))=(1,j),f((2,j-1))=(2,j-1)$, then we get that $f((2,j))=(2,j)$, and thus we proved the claim for the ``top edge'' of the $m\times n$-grid frame. Observe that, having completed the proof for the ``top side'', the symetric arguments hold also for the ``left side'' and the ``right side'' of the $m\times n$-grid frame, that is $\{(k,\ell)~|~0 \leq k\leq 12\cdot m+4, 0\leq \ell \leq 2\}\cup \{(k,\ell)~|~0 \leq k\leq 12\cdot m+4, 8\cdot n+2\leq \ell \leq 8\cdot n+4\}$, and, in turn, also for the ``bottom side'' of the $m\times n$-grid frame, that is, $\{(k,\ell)~|~12\cdot m+2 \leq k\leq 12\cdot m+4, 0\leq \ell \leq 8\cdot n+4\}$. This completes the proof of the lemma.                    
\end{proof}

Without loss of generality, we assume that for every grid graph embedding $f$ of $G_B$ it follows that $f((0,0))=(0,0)$ $f((1,0))=(1,0)$ and $f((1,1))=(1,1)$ where $(0,0),(1,0)$ and $(1,1)$ are vertices of the $m\times n$-grid frame embedding. From Lemma \ref{lem:fixedFrame} we get that the $m\times n$-grid frame embedding is fixed. Next, we would like to show that the embedding of the (boundary) rectangle of each battery gadget is also fixed. To this end, consider a grid graph $G$ with grid graph embedding $f$ of $G$, two vertices $u,v$ that are embedded in the same row or column, and a path $P=(u,\ldots,v)$ of size $d_f(u,v)$. Then, it is clear that the embedding of the path must be a straight line between $u$ and $v$. We prove this in the next lemma.

\begin{lemma}\label{lem:pathFix} 
Let $G$ be a grid graph and let $f$ be a grid graph embedding $G$. Let $r,c_1,c_2\in \mathbb{N}, c_2>c_1$ and let $u,v\in V(G)$ such that $f(u)=(r,c_1)$ and $f(v)=(r,c_2)$ and let $P=(a_0=u,a_2,\ldots,a_{\ell}=v)$ where $\ell=c_2-c_1$ be a simple path of length $c_2-c_1$. Then, for every $0\leq i\leq \ell$ it follows that $f(a_i)=(r,c_1+i)$. Similarly, let $c,r_1,r_2\in \mathbb{N}, r_2>r_1$ and let $u,v\in V(G)$ such that $f(u)=(r_1,c)$ and $f(v)=(r_2,c)$ and let $P=(a_0=u,a_2,\ldots,a_{\ell}=v)$ where $\ell=r_2-r_1$ be a simple path of size $r_2-r_1$. Then, for every $0\leq i\leq \ell$ it follows that $f(a_i)=(r_1+i,c)$. 
\end{lemma} 

\begin{proof}
We prove the first part of the lemma; the second part can be proved symmetrically. Assume toward a contradiction that there exists $0\leq j\leq \ell$ such that $f(a_j)\neq(r,c_1+j)$. Let $0\leq i\leq \ell$ be the minimal such $j$. Observe that $0< i< \ell$. By the minimality of $i$, it follows that $f(a_{i-1})=(r,c_1+i-1)$. Notice that $\{a_{i-1},a_i\}\in E(G)$, therefore, there are three options: $f(a_i)=(r+1,c_1+i-1)$, $f(a_i)=(r-1,c_1+i-1)$ or $f(a_i)=(r,c_1+i-2)$. Assume that $f(a_i)=(r-1,c_1+i)$ (the other cases are similar). Then, it follows that $d_f(v,a_i)= |r-(r-1)|+|c_2-(c_1+i-1)|=1+|c_2-c_1-i+1|>c_2-c_1-i$. On the other hand, since $(a_i,a_{i+1},\ldots,a_\ell=v)$ is a path from $a_i$ to $v$ of size $c_2-c_1-i$, it follows that $d(v,a_i)\leq c_2-c_1-i$. So, we get that $d_f(v,a_i)>d(v,a_i)$, a contradiction to Observation \ref{lem:dfLeqd}.         
\end{proof}

 Now, because we saw that the embedding of the  $m\times n$-grid frame is fixed, and the ``sides'' of the rectangles are straight lines between two vertices in the $m\times n$-grid frame, then the embedding of the rectangles is also fixed. In addition, observe that each battery gadget has a straight line that seperates the two sides of the battery gadget. By using similar arguments, we show that the embeddings of these lines are also fixed. We prove these insights in the following lemmas. For that purpose, we denote the vertices of the rectangles and the line that separates the sides of the battery gadgets by $\{(6(i-1)+2,j)~|~1\leq i\leq 2\cdot m, 2\leq j\leq 8\cdot n+2\} \cup \{(i,8(j-1)+2)~|~2\leq i\leq 12\cdot m+2, 1\leq j\leq n\}$ (see Figure \ref{fi:fix embgre}).  

\begin{figure}[!t]
\centering
\includegraphics[width=0.5\textwidth, page=9]{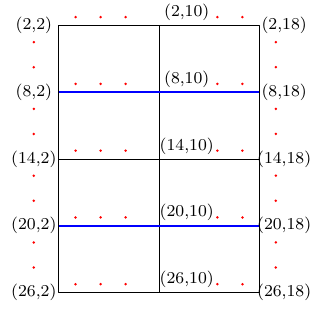}
\caption{Notation for $2\times 2$ rectangles corresponding to rectangles of battery gadgets. The lines that separates the sides of the gadgets are in blue.}\label{fi:fix embgre}
\end{figure}

\begin{lemma} \label{obs:fixemb}
Let $B=\{B_{i,j}=(x_1^{(i,j)},x_2^{(i,j)})~|~1\leq i\leq m, 1\leq j\leq n\}$ be a set of batteries. Let $f$ be a grid graph embedding of $G$. Assume that $f((0,0))=(0,0)$ $f((1,0))=(1,0)$ and $f((1,1))=(1,1)$. Then, for every $1\leq i\leq 2\cdot m, 2\leq j\leq 8\cdot n+2$ it follows that $f((6(i-1)+2,j))=(6(i-1)+2,j)$, and for every $2\leq i\leq 12\cdot m+2, 1\leq j\leq n$ it follows that $f((i,8(j-1)+2))=(i,8(j-1)+2)$.  
\end{lemma}

\begin{proof}
Let $1\leq i\leq 2\cdot m$. From Lemma \ref{lem:fixedFrame} we get that $f((6(i-1)+2,2))=(6(i-1)+2,2)$ and $f((6(i-1)+2,8\cdot n+2))=(6(i-1)+2,8\cdot n+2)$. In addition, we have that $((6(i-1)+2,2),(6(i-1)+2,3),\ldots,(6(i-1)+2,8\cdot n+1),(6(i-1)+2,8\cdot n+2)$ is a path from $(6(i-1)+2,2)$ to $(6(i-1)+2,8\cdot n+2)$ of size $8\cdot n$. Therefore, by Lemma \ref{lem:pathFix}, we get that for every $2\leq j\leq 8\cdot n+2$ it follows that $f((6(i-1)+2,j))=(6(i-1)+2,j)$. 
Similarly, let $1\leq j\leq n$. From Lemma \ref{lem:fixedFrame} we get that $f(2,8(j-1)+2))=(2,8(j-1)+2)$ and $f(12\cdot m+2,8(j-1)+2))=(12\cdot m+2,8(j-1)+2)$. In addition, we have that $((2,8(j-1)+2),(3,8(j-1)+2),\ldots,(12\cdot m+1,8(j-1)+2),(12\cdot m+2,8(j-1)+2))$ is a path from $(2,8(j-1)+2)$ to $(12\cdot m+2,8(j-1)+2)$ of size $12\cdot m$. Therefore, by Lemma \ref{lem:pathFix}, we get that for every $2\leq i\leq 12\cdot m+2$ it follows that $f((i,8(j-1)+2))=(i,8(j-1)+2)$.         
\end{proof}  


Now we want to focus on the positive and the negative sides. For this purpose, we denote the vertex set of the positive side of the battery gadget $H_{i,j}$ by $P_{i,j}=\{P_{i,j}(\ell)~|~1\leq \ell \leq 14\}$, and vertex set of the negative side of the battery gadget $H_{i,j}$ by $N_{i,j}=\{N_{i,j}(\ell)~|~1\leq \ell \leq 9\}$ (see Figure \ref{fi:posNegNot}). For a battery gadget $H$ with a grid graph embedding $f$, we set $p_f(H)=+$ if the positive side of the gadget is embedded to the top of the gadget, that is, $f(P_{i,j}(1))=(12(i-1)+9,8(j-1)+6)$; otherwise the negative side is embedded to the top of the gadget ($f(N_{i,j}(1))=(12(i-1)+9,8(j-1)+6)$), and we set $p_f(H)=-$. Observe that these are the only options in any embedding of $H$. In the following lemmas we show that once we choose which side to embedded in the top of the gadget, the embedding of the vertices of the positive and negative sides are ``almost fixed''.

Notice that the six synchronization vertices in the battery gadget might be embedable inside or outside the rectangle. If we look at two adjacent battery gadget in the same column, say $H_{i,j}$ and $H_{i+1,j}$. If $p_f(H_{i+1,j})=+$, then vertices $1$ and $3$ (see Figure \ref{fi:batgat}) must be embedded outside of $H_{i+1,j}$, therefore they are embedded inside $H_{i,j}$. Thus in $H_{i,j}$ the positive side cannot be embedded at the bottom, and hence, we get that $p_f(H_{i,j})=+$. Similarly, if $p_f(H_{i+1,j})=-$, then vertex $4$ (see Figure \ref{fi:batgat}) must be embedded outside of $H_{i+1,j}$, so it is embedded inside $H_{i,j}$. Thus in $H_{i,j}$ the negative side cannot be embedded at the bottom, so we get that $p_f(H_{i,j})=-$. We see that the battery gadgets are ``synchronized'' like the batteries, i.e. $+$ can only be connected to $-$ and vice versa. In order to prove this claim, first we prove that if $p_f(H)=+$ (resp.~$p_f(H)=-$), then vertices $1$ and $3$ (resp.~vertex $4$) must be embedded outside of $H$. For that purpose, we also denote some vertices of the gadget by $a,b,c,d,e,f,g,h$ (see Figure \ref{fi:posNegNot}). 

\begin{figure}[!t]
\centering
\includegraphics[width=0.5\textwidth, page=10]{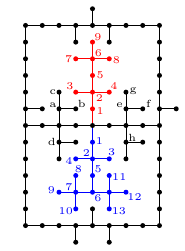}
\caption{A battery gadget $H=(0,0)$ with notations for some of the vertices.}\label{fi:posNegNot}
\end{figure}    

\begin{lemma}\label{lem:synclem}
Let $B=\{B_{i,j}=(x_1^{(i,j)},x_2^{(i,j)})~|~1\leq i\leq m, 1\leq j\leq n\}$ be a set of batteries. Let $f$ be a grid graph embedding of $G$. Let $1\leq i\leq m, 1\leq j\leq n$. Assume that $f((0,0))=(0,0)$ $f((1,0))=(1,0)$ and $f((1,1))=(1,1)$.
\begin{itemize}
\item If $f(P_{i,j}(1))=(12(i-1)+9,8(j-1)+6)$ then there exist $u,v\in P_{i,j}$ such that $f(u)=(12(i-1)+5,8(j-1)+5)$ and $f(u)=(12(i-1)+5,8(j-1)+7)$.
\item If $f(P_{i,j}(1))=(12(i-1)+11,8(j-1)+6)$ then there exist $u,v\in P_{i,j}$ such that $f(u)=(12(i-1)+15,8(j-1)+5)$ and $f(u)=(12(i-1)+15,8(j-1)+7)$.
\item If $f(N_{i,j}(1))=(12(i-1)+9,8(j-1)+6)$ then there exists $u\in N_{i,j}$ such that $f(u)=(12(i-1)+5,8(j-1)+6)$.
\item If $f(N_{i,j}(1))=(12(i-1)+11,8(j-1)+6)$ then there exists $u\in P_{i,j}$ such that $f(u)=(12(i-1)+15,8(j-1)+6)$. 
\end{itemize}
\end{lemma}

\begin{proof}
We prove that if $f(N_{i,j}(1))=(12(i-1)+9,8(j-1)+6)$ then there exists $u\in N_{i,j}$ such that $f(u)=(12(i-1)+5,8(j-1)+6)$. The other cases can be proved similarly.
Since $\{N_{i,j}(1),N_{i,j}(2)\}\in E(G)$ and by Lemma \ref{obs:fixemb} it follows that $f((12(i-1)+10,8(j-1)+6))= (12(i-1)+10,8(j-1)+6)$ then $f(N_{i,j}(2))=(12(i-1)+8,8(j-1)+6)$ or $f(N_{i,j}(2))=(12(i-1)+9,8(j-1)+5)$ or $f(N_{i,j}(2))=(12(i-1)+9,8(j-1)+7)$. Assume that $f(N_{i,j}(2))=(12(i-1)+9,8(j-1)+5)$. Observe that $N_{i,j}(2)$ has four neighbors, $N_{i,j}(1)$, $N_{i,j}(3)$, $N_{i,j}(4)$ and $N_{i,j}(5)$, therefore one of them is embedded at $(12(i-1)+10,8(j-1)+5)$. This is a contradiction, since from Lemma \ref{obs:fixemb} we get that $f((12(i-1)+10,8(j-1)+5))=(12(i-1)+10,8(j-1)+5)$ and $f$ is an injection. From similar arguments, we get a contradiction also if $f(N_{i,j}(2))=(12(i-1)+9,8(j-1)+7)$, so we get that $f(N_{i,j}(2))=(12(i-1)+8,8(j-1)+6)$.

Now, since $\{N_{i,j}(2),N_{i,j}(5)\}\in E(G)$ it follows that $f(N_{i,j}(5))=(12(i-1)+7,8(j-1)+6)$ or $f(N_{i,j}(5))=(12(i-1)+8,8(j-1)+5)$ or $f(N_{i,j}(5))=(12(i-1)+8,8(j-1)+7)$. Assume that $f(N_{i,j}(5))=(12(i-1)+8,8(j-1)+5)$. Since $\{N_{i,j}(5),N_{i,j}(6)\}\in E(G)$ we get that $f(N_{i,j}(6))=(12(i-1)+7,8(j-1)+5)$ or $f(N_{i,j}(6))=(12(i-1)+9,8(j-1)+5)$ or $f(N_{i,j}(6))=(12(i-1)+8,8(j-1)+4)$. If $f(N_{i,j}(6))=(12(i-1)+7,8(j-1)+5)$ we get that one of $N_{i,j}(7),N_{i,j}(8),N_{i,j}(9)$ is embedded to $(12(i-1)+7,8(j-1)+6)$, but one of $N_{i,j}(3),N_{i,j}(4)$ must be embedded there, so we get a contradiction. Similarly, we get a contradiction also if $f(N_{i,j}(6))=(12(i-1)+9,8(j-1)+5)$.

Now, if $f(N_{i,j}(6))=(12(i-1)+8,8(j-1)+4)$, then one of $N_{i,j}(7),N_{i,j}(8),N_{i,j}(9)$ is embedded at $(12(i-1)+9,8(j-1)+4)$. From Lemma \ref{obs:fixemb} we get that $f((12(i-1)+10,8(j-1)+4))=(12(i-1)+10,8(j-1)+4)$ and $(12(i-1)+10,8(j-1)+4)$ has four neighbors that $N_{i,j}(7),N_{i,j}(8),N_{i,j}(9)$ are not among them. Therefore one of them must be embedded at $(12(i-1)+9,8(j-1)+4)$, a contradiction. So we get that $f(N_{i,j}(5))=(12(i-1)+7,8(j-1)+6)$. By similar arguments, we get that $f(N_{i,j}(6))=(12(i-1)+6,8(j-1)+6)$. Therefore one of $N_{i,j}(7),N_{i,j}(8),N_{i,j}(9)$ is embedded at $(12(i-1)+5,8(j-1)+6)$. This completes the proof.                  
\end{proof}

Now we show that the battery gadgets are vertically ``synchronized'', that is, $p_f(i,j)=p_f(i+1,j)$.  

\begin{lemma}\label{lem:sync}
Let $B=\{B_{i,j}=(x_1^{(i,j)},x_2^{(i,j)})~|~1\leq i\leq m, 1\leq j\leq n\}$ be a set of batteries. Let $f$ be a grid graph embedding of $G$. Let $1\leq i< m, 1\leq j\leq n$. Assume that $f((0,0))=(0,0)$ $f((1,0))=(1,0)$ and $f((1,1))=(1,1)$. Then $p_f(i,j)=p_f(i+1,j)$ for every $1\leq i\leq m$ and $1\leq j\leq n$.  
\end{lemma}

\begin{proof}
Assume without loss of generality that $p_f(i+1,j)=-$ (the other case can be shown similarly), so $f(N_{i+1,j}(1))=(12(i)+9,8(j-1)+6)$ . Therefore, by Lemma \ref{lem:synclem}, we get that there exists $u\in N_{i,j}$ such that $f(u)=(12(i)+5,8(j-1)+6)$. Now, $\{S_{i,j}(2),(12(i)+4,8(j-1)+6)\}\in E(G)$, and by Lemma we get that $f((12(i)+4,8(j-1)+6))=(12(i)+4,8(j-1)+6)$, $f((12(i)+3,8(j-1)+6))=(12(i)+3,8(j-1)+6)$ and $f((12(i)+5,8(j-1)+6))=(12(i)+5,8(j-1)+6)$. Therefore, we get that $f(S_{i,j}(2))=(12(i)+3,8(j-1)+6)$. Assume towards a contradiction that $p_f(i,j)=+$. Then, we get that $f(N_{i,j}(1))=(12(i-1)+11,8(j-1)+6)$, so, by Lemma \ref{lem:synclem} we get that there exists $u\in P_{i,j}$ such that $f(u)=(12(i-1)+15,8(j-1)+6)=(12(i)+3,8(j-1)+6)=f(S_{i,j}(2))$. This is a contradiction, since $f$ is an injection.       
\end{proof}


Let $B=\{B_{i,j}=(x_1^{(i,j)},x_2^{(i,j)})~|~1\leq i\leq m, 1\leq j\leq n\}$ be an instance of the {\sc Batteries} problem, and let $f$ be a grid graph embedding of $G$. Let $1\leq j\leq n, 1\leq i\leq m$. We denote by $V_f(H_{i,j})$ the voltage of the side of the battery gadget that is embedded by $f$ at the top of $H_{i,j}$.  Notice that in Lemma \ref{lem:sync}, we stated a property that is, in some sense, analogous to the correctness of placement of a set of batteries. Now we show that if $G$ is a grid graph, then the property analogous to safeness is also preserved. That is, for every $1\leq i\leq m$, there exists $1\leq j\leq n$ such that $V_f(H_{i,j})=0$. Notice that if, in $H_{i,j}$, vertex $1$ (see Figure \ref{fi:batgat}) is embedded inside $G_{i,j}$ and $V_f(G_{i,j})=1$, then it must be that vertex $2$ is embedded outside $G_{i,j}$. We prove this in the next lemma.

\begin{lemma}
Let $B=\{B_{i,j}=(x_1^{(i,j)},x_2^{(i,j)})~|~1\leq i\leq m, 1\leq j\leq n\}$ be a set of batteries. Let $f$ be a grid graph embedding of $G$ and let $1\leq i\leq m, 1\leq j\leq n$. Assume that $f((0,0))=(0,0)$ $f((1,0))=(1,0)$ and $f((1,1))=(1,1)$. If $V_f(i,j)=1$ and $f(W(i,j-1))=(12(i-1)+7,8(j-1)+3)$, then $f(W(i,j))=(12(i-1),8j+3)$.
\end{lemma}

\begin{proof}
Assume without loss of generality that $p_f(H_{i,j})=+$; then it follows that $f(P_{i,j}(1))=(12(i-1)+9,8(j-1)+6)$. Observe that from Lemma \ref{obs:fixemb} we get that $f((12(i-1)+8,8(j-1)+4))=(12(i-1)+8,8(j-1)+4)$. Since we have that $\{(12(i-1)+8,8(j-1)+4),a_{i,j}\},\{(12(i-1)+8,8(j-1)+4),d_{i,j}\}\in E(G)$, and again from Lemma \ref{obs:fixemb} we also have that $f((12(i-1)+8,8(j-1)+5))=(12(i-1)+8,8(j-1)+5)$ and $f((12(i-1)+8,8(j-1)+3))=(12(i-1)+8,8(j-1)+3)$. Therefore, we get that $f(a_{i,j})=(12(i-1)+7,8(j-1)+4))$ or $f(d_{i,j})=(12(i-1)+7,8(j-1)+4))$. Assume without loss of generality that $f(a_{i,j})=(12(i-1)+7,8(j-1)+4))$. Since $\{a_{i,j},b_{i,j}\},\{a_{i,j},c_{i,j}\}\in E(G)$ and $f(W(i,j-1))=(12(i-1)+7,8(j-1)+3)$, we get that one of $b_{i,j},c_{i,j}$ is embedded to $(12(i-1)+7,8(j-1)+5)$. Now, since $f(P_{i,j}(1))=(12(i-1)+9,8(j-1)+6)$ and $V_f(i,j)=1$ it follows that one of the three neighbors of $f(P_{i,j}(1)$ is embedded to $(12(i-1)+9,8(j-1)+7)$. Moving on forward, again from Lemma \ref{obs:fixemb} we get that $f((12(i-1)+10,8(j-1)+8))=(12(i-1)+10,8(j-1)+8)$. Since $\{e_{i,j},(12(i-1)+10,8(j-1)+8)\},\{h_{i,j},(12(i-1)+10,8(j-1)+8)\}\in E(G)$ and from Lemma \ref{obs:fixemb} we have that $f((12(i-1)+10,8(j-1)+9))=(12(i-1)+10,8(j-1)+9)$ and $f((12(i-1)+10,8(j-1)+7))=(12(i-1)+10,8(j-1)+7)$, we get that $f(e_{i,j})=(12(i-1)+9,8(j-1)+8)$ or $f(h_{i,j})=(12(i-1)+9,8(j-1)+8)$. Assume without loss of generality that $f(e_{i,j})=(12(i-1)+9,8(j-1)+8)$. Since $\{e_{i,j},f_{i,j}\},\{e_{i,j},g_{i,j}\}\in E(G)$, we get that one of $f_{i,j},g_{i,j}$ is embedded to $(12(i-1)+9,8(j-1)+9)$. Now, again from Lemma \ref{obs:fixemb}, we get that $f((12(i-1)+9,8(j-1)+10))=(12(i-1)+9,8(j-1)+10)$. Since $\{W(i,j),(12(i-1)+9,8(j-1)+10)\}\in E(G)$ and $f((12(i-1)+10,8(j-1)+10))=(12(i-1)+10,8(j-1)+10)$ and $f((12(i-1)+8,8(j-1)+10))=(12(i-1)+8,8(j-1)+10)$ (by Lemma \ref{obs:fixemb}), then we get that $f(W(i,j))=(12(i-1),8(j-1)+11)$ or $f(W(i,j))=(12(i-1),8(j-1)+9)$. Since we saw that one of $f_{i,j},g_{i,j}$ is embedded to $(12(i-1)+9,8(j-1)+9)$, then it follows that $f(W(i,j))=(12(i-1),8(j-1)+11)=(12(i-1),8j+3)$. This ends the proof.                  
\end{proof}

Moreover, observe that for every $1\leq i\leq m$ it follows that vertex $1$ of $H_{i,1}$ is embedded inside $H_{i,1}$ and vertex $2$ of $H_{i,n}$ is embedded inside $H_{i,n}$, because these two battery gadgets are adjacent to the $m\times n$-grid frame. We prove this in the next lemma.

\begin{lemma}\label{lem:wire FirstLast}
Let $B=\{B_{i,j}=(x_1^{(i,j)},x_2^{(i,j)})~|~1\leq i\leq m, 1\leq j\leq n\}$ be a set of batteries. Let $f$ be a grid graph embedding of $G$ and let $1\leq i\leq m$. Assume that $f((0,0))=(0,0)$ $f((1,0))=(1,0)$ and $f((1,1))=(1,1)$. Then, for every $1\leq i\leq m$ it follows that $f(W(i,0))=(12(i-1)+7,3)$ and $f(W(i,n))=(12(i-1)+7,8n+1)$. 
\end{lemma}

\begin{proof}
We prove that $f(W(i,0))=(12(i-1)+7,3)$, the other case is similar. From Lemma \ref{obs:fixemb} we get that $f((12(i-1)+7,8(j-1)+2))=(12(i-1)+7,8(j-1)+2)$. Since $\{W(i,0),(12(i-1)+7,8(j-1)+2)\}\in E(G)$ it follows that $f(W(i,0))=(12(i-1)+7,3)$ or $f(W(i,0))=(12(i-1)+7,1)$ or $f(W(i,0))=(12(i-1)+6,2)$ or $f(W(i,0))=(12(i-1)+8,2)$. Again, from Lemma \ref{obs:fixemb}, we get that $f((12(i-1)+7,1))=(12(i-1)+7,1)$, $f((12(i-1)+6,2))=(12(i-1)+6,2)$ and $f((12(i-1)+8,2))=(12(i-1)+8,2)$. Therefore, since $f$ is an injection, we get that $f(W(i,0))=(12(i-1)+7,3)$.     
\end{proof}


In the next lemma we show that if $G$ is a grid graph with grid graph embedding $f$ of $G$, then the ``safeness'' is also preserved.               

\begin{lemma} \label{lem:safenessEmb}
Let $G=\{G_{i,j}=(x_1^{(i,j)},x_2^{(i,j)})~|~1\leq i\leq m, 1\leq j\leq n\}$ be a set of battery gadgets. Let $f$ be a grid graph embedding of $G'$ and let $1\leq i\leq m$. Then there exists $1\leq j\leq n$ such that $V_f(G_{i,j})=0$. 
\end{lemma}

\begin{proof}
Assume towards a contradiction that there is $1\leq i\leq m$ such that for every $1\leq j\leq n$ it follows that $V_f(G_{i,j})=1$. We show by induction that for every $1\leq j\leq n$ it follows that vertex $2$ of $G_{i,j}$ is embedded outside $G_{i,j}$. For $j=1$ from Observation \ref{lem:safenessEmb} it follows that vertex $1$ of $G_{i,1}$ is embedded inside $G_{i,1}$. As we assume that $V_f(G_{i,1})=1$ we get that vertex $2$ is embedded outside $G_{i,1}$. Now, let $1<j\leq n$. By the inductive hypothesis, we get that vertex $2$ is embedded outside $G_{i,j-1}$. Therefore, it follows that vertex $1$ of $G_{i,j}$ is embedded inside $G_{i,j}$. Moreover, as we assume that $V_f(G_{i,j})=1$ it follows that vertex $2$ is embedded outside $G_{i,j}$. So we get from the induction that vertex $2$ is embedded outside $G_{i,n}$. This contradicts Observation \ref{lem:safenessEmb}.            
\end{proof}

Now we are ready to prove the reverse direction of the correctness of the reduction. In the next lemma we show that if $G$ is a grid graph, then $B$ is a a yes-instance of the {\sc Batteries} problem.

\begin{lemma} \label{lem:EmbtoBat}
Let $B=\{B_{i,j}=(x_1^{(i,j)},x_2^{(i,j)})~|~1\leq i\leq m, 1\leq j\leq n\}$ be an instance of the {\sc Batteries} problem. If $\mathsf{reduce}_2(B)=G_B$ is a yes instance of \gridEm, then $B$ is is a a yes-instance of the {\sc Batteries} problem.  
\end{lemma}

\begin{proof}
Let $f$ be a grid graph embedding of $G_B$. For every $1\leq i\leq m$ and $1\leq j\leq n$, we set $p(i,j)=p_f(H_{i,j})$. We show that $p$ is correct and safe. Let $1\leq j\leq n, 1\leq i<m$. Then by Lemma \ref{lem:sync}, we get that $p_f(H_{i,j})=p_f(H_{i+1,j})$. Therefore we get that $p(i,j)=p(i+1,j)$. So, as the choice of $i$ and $j$ was arbitrary, $p$ is correct. Now, let $1\leq i'\leq m$. By Lemma \ref{lem:safenessEmb} there exists $1\leq j'\leq n$ such that $V_f(H_{i',j'})=0$. Therefore, we have that $V_p(i',j')=0$ so, as the choice of $i$ and $j$ was arbitrary, $p$ is safe.
We found a correct and safe placement for $B$, so $B$ is a yes-instance of the {\sc Batteries} problem.     
\end{proof}

Now, we turn to prove the other direction of the correctness of the reduction. Given a yes-instance $B=\{B_{i,j}=(x_1^{(i,j)},x_2^{(i,j)})~|~1\leq i\leq m, 1\leq j\leq n\}$ of the {\sc Batteries} problem with an $(m,n)$-placement $p$ that is correct and safe, we construct a grid graph embedding $f_p$ of $G_B$. For every $1\leq i\leq m$ and $1\leq j\leq n$, we will define grid graph embedding $f_p^{i,j}$ of the gadget $H_{i,j}=(x_1^{(i,j)},x_2^{(i,j)})$ such that all the embeddings of the gadgets ``agree''. Intuitively, we can think of our goal as having a puzzle where we show that every piece of the puzzle is placed ``correctly'' and all the pieces together are connected. In our case, we need to make sure that the locations of the embedding of the six synchronization vertices and the two wire vertices are synchronized. That is, for every $1< i\leq m, 1\leq j\leq n$ and $\ell \in \{ 3,4,5\}$ it follows that vertex $\ell$ (see Figure \ref{fi:batgat}) of $H_{i+1,j}$ is embedded inside $H_{i+1,j}$ in $f_p^{i+1,j}$ if and only if the vertex $\ell+3$ is embedded outside $H_{i,j}$ in $f_p^{i,j}$. Notice that in the embedding we give in this section for the battery gadget (see Figures \ref{fi:obs1}, \ref{fi:obs2}, \ref{fi:obs3}), if $p_f(H)=+$, then vertices $3,5$ and $7$ are embedded outside the gadget and vertices $4,6$ and $8$ are embedded inside the gadget. Otherwise, $p_f(H)=-$, and then vertices $3,5$ and $7$ are embedded inside the gadget and vertices $4,6$ and $8$ are embedded outside the gadget. Therefore, for our purpose, we get that $H_{i+1,j}$ and $H_{i,j}$ are ``connected'' or synchronized if and only if $p_f(H_{i+1,j})=p_f(H_{i+1,j})$. In addition, for every $1\leq i\leq m, 1\leq j< n$ it follows that vertex $2$ of $H_{i,j}$ is embedded inside $H_{i,j}$ in $p_d^{i,j}$ if and only if vertex $1$ is embedded outside $H_{i,j+1}$ in $f_p^{i,j+1}$. Moreover, as we saw in Lemma \ref{lem:wire FirstLast}, we also need to make sure that vertex $1$ of $H_{i,1}$ is embedded inside $H_{i,1}$ in $f_p^{i,1}$ and vertex $2$ of $H_{i,n}$ is embedded inside $H_{i,n}$ in $f_p^{i,n}$ for every $1\leq i\leq n$. If we can find such embeddings, then we can find an embedding of $G_B$. We prove this insight in the next lemma. For this purpose, for every $1\leq i\leq m$ and $1\leq j\leq n$ we denote the vertices of $H^{i,j}$ that we must embed inside the rectangle of the gadget by $\mathsf{InH_{i,j}}$ (see Figure \ref{fi:GreenIn}).    

\begin{figure}[!t]
\centering
\includegraphics[width=0.5\textwidth, page=11]{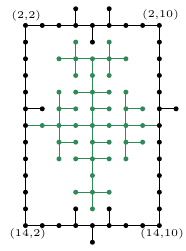}
\caption{A battery gadget $H=(1,1)$. The set of vertices $\mathsf{InH}$ is in green.}\label{fi:GreenIn}
\end{figure}

\begin{lemma}\label{lem:condForGrid}
Let $B=\{B_{i,j}=(x_1^{(i,j)},x_2^{(i,j)})~|~1\leq i\leq m, 1\leq j\leq n\}$ be a set of batteries. Assume that for every $1\leq i\leq m, 1\leq j\leq n$ there exists a grid graph embedding $f^{i,j}$ of the battery gadget $H^{i,j}=(x_1^{(i,j)},x_2^{(i,j)})$ such that the following conditions are satisfied.
\begin{itemize}  
\item For every $2\leq j\leq 10$ it follows that $f^{i,j}((2,j))=(2,j), f^{i,j}((14,j))=(14,j)$, and for every $2\leq i\leq 14$ it follows that $f^{i,j}((i,2))=(i,2), f^{i,j}((i,10))=(i,10)$.
\item For every $1\leq i< m, 1\leq j\leq n$ it follows that $f^{i+1,j}(S_1)=(1,5)$ if and only if $f^{i,j}(S_4)=(13,5)$ and $f^{i+1,j}(S_1)=(3,5)$ if and only if $f^{i,j}(S_4)=(15,5)$ if. Similarly, $f^{i+1,j}(S_2)=(1,6)$ if and only if $f^{i,j}(S_5)=(13,6)$ and $f^{i+1,j}(S_2)=(3,6)$ if and only if $f^{i,j}(S_5)=(15,6)$. In addition, $f^{i+1,j}(S_3)=(1,7)$ if and only if $f^{i,j}(S_6)=(13,7)$ and $f^{i+1,j}(S_3)=(3,7)$ if and only if $f^{i,j}(S_6)=(15,7)$.  
\item For every $1\leq i\leq m, 1\leq j< n$ it follows that $f^{i,j}(W_2)=(7,9)$ if and only if $f^{i,j+1}(W_1)=(7,1)$, and $f^{i,j}(W_2)=(7,11)$ if and only if $f^{i,j+1}(W_1)=(7,3)$.
\item For every $1\leq i\leq m$ it follows that $f^{i,1}(W_1)=(7,3)$ and $f^{i,n}(W_2)=(7,9)$.    
\end{itemize}
Then, $G$ is a grid graph.
\end{lemma}

\begin{proof}
We define a grid graph embedding $f$ of $G$. First, for every vertex $(i,j)$ of the $m\times n$-grid frame we set $f((i,j))=(i,j)$. Now, for every $1\leq i\leq m, 2\leq j\leq 8\cdot n+2$ we set $f((12(i-1)+2,j))=(12(i-1)+2,j)$, and for every $2\leq i\leq 12\cdot m+2, 1\leq j\leq n$ we set $f((i,8(j-1)+2))=(i,8(j-1)+2)$. For every $1\leq i\leq m$ and $1\leq j\leq n$, for every $u\in \mathsf{InH_{i,j}}$ we set $f(u)=f^{i,j}(u)+(12(i-1),8(j-1))$. For every $1\leq i\leq m$ and $1\leq j\leq n$ we set $f(W(i,j))=f_{i,j}(W_2)+(12(i-1),8(j-1))$ and $f(W(i,0))=f_{i,1}(W_1)+(12(i-1),0)$. For every $1\leq i< m-1$ and $1\leq j\leq$ we set $f(S_{i,j}(1))=f_{i,j}(S_4)+(12(i-1),8(j-1))$,$f(S_{i,j}(2))=f_{i,j}(S_5)+(12(i-1),8(j-1))$ and $f(S_{i,j}(3))=f_{i,j}(S_6)+(12(i-1),8(j-1))$. In addition, for every $1\leq j\leq$, we set $f(S_{m-1,j}(1))=f_{m-1,j}(S_1)+(12(m-2),8(j-1))$,$f(S_{m-1,j}(2))=f_{m-1,j}(S_2)+(12(m-2),8(j-1))$ and $f(S_{m-1,j}(3))=f_{i,j}(S_3)+(12(m-2),8(j-1))$. We show that $f$ is a grid graph embedding of $G$.

First, observe that $f$ is a function from $V(G)$ to $\mathbb{N}\times\mathbb{N}$. We show that $f$ is an injection. Let $u,v\in V(G)$. We show that $f(u)\neq f(v)$ where $u=S_{i-1,j}(1)$ and $v\in InH_{i,j}$ for some $1< i\leq m$ and $1\leq j\leq n$, the other cases are simple, or can be proved similarly. We have that $f(S_{i-1,j}(1))=f_{i-1,j}(S_4)+(12(i-2),8(j-1))$. Assume that $f^{i-1,j}(S_4)=(13,5)$ (the other case is similar). Then, we have that $f^{i,j}(S_1)=(1,5)$. Since $f^{i,j}$ is an injection, we get that $f^{i,j}(v)\neq f^{i,j}(S_1)$.

Therefore we have that $f(u)=f(S_{i-1,j}(1))=(13,5)+(12(i-2),8(j-1))=(1,5)+(12(i-1),8(j-1))=f^{i,j}(S_1)+(12(i-1),8(j-1))\neq f^{i,j}(v)+(12(i-1),8(j-1))=f(v)$. So we get that $f$ is an injection. Now, observe that for every $\{u,v\}\in E(G)$ it follows that $d_f(u,v)=1$ since every such edge is in $H^{i,j}$ for some $i,j$ or from $m\times n$-grid frame. From Definition \ref{def:Grid graph embedding} we get that $f$ is a grid graph embedding of $G$.         
\end{proof}

   
In the next observations, we consider some embeddings of the battery gadget.

\begin{observation}\label{obs:ex1}
There exists a grid graph embedding $f$ of the battery gadget $H=(0,x_2)$ where $x_2\in \{0,1\}$, such that $P_f(H)=+$ and vertices $1$ and $2$ are embedded inside $H$ by $f$. Similarly, there exists a grid graph embedding $f$ of the battery gadget $H=(x_1,0)$ where $x_1\in \{0,1\}$, such that $p_f(G)=-$ and vertices $1$ and $2$ are embedded inside $H$ in $f$.
\end{observation}

\begin{proof}
We present an embedding $f$ of $H=(0,1)$ where $p_f(G)=+$ and vertices $1$ and $2$ are embedded inside $H$ in Figure \ref{fi:obs1}.
The other cases are similar.
\end{proof}

\begin{figure}[!t]
\centering
\includegraphics[width=0.5\textwidth, page=12]{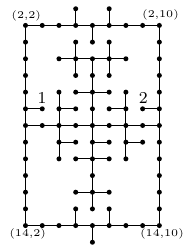}
\caption{Example of an embedding of $H=(0,1)$ where $p_f(G)=+$ and vertices $1$ and $2$ are embedded inside $H$.}\label{fi:obs1}
\end{figure}

\begin{observation}\label{obs:ex2}
There exists a grid graph embedding $f$ of the battery gadget $H=(x_1,x_2)$ where $x_1,x_2\in \{ 0,1\}$ such that either $p_f(G)=+$ or $p_f(H)=-$, and vertex $1$ is embedded inside $H$ in $f$ and vertex $2$ is embedded outside $H$ in $f$. 
\end{observation}

\begin{proof}
We present an embedding $f$ of $H=(1,1)$ where $p_f(G)=+$, vertex $1$ is embedded inside $H$ and vertex $2$ is embedded outside $H$ in Figure \ref{fi:obs2}.
The other cases are similar.
\end{proof}

\begin{figure}[!t]
\centering
\includegraphics[width=0.5\textwidth, page=13]{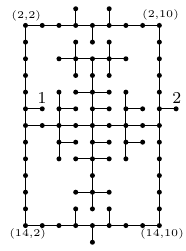}
\caption{Example of an embedding of $H=(1,1)$ where $p_f(G)=+$, vertex $1$ is embedded inside $H$ and vertex $2$ is embedded outside $H$.}\label{fi:obs2}
\end{figure}

\begin{observation}\label{obs:ex3}
There exists a grid graph embedding $f$ of the battery gadget $H=(x_1,x_2)$ where $x_1,x_2\in \{0,1\}$ such that either $p_f(G)=+$ or $p_f(G)=-$ and vertex $1$ is embedded outside $H$ in $f$ and vertex $2$ is embedded inside $H$ in $f$.
\end{observation}

\begin{proof}
We present an embedding $f$ of $H=(1,1)$ where $p_f(G)=+$, vertex $1$ is embedded outside $H$ and vertex $2$ is embedded inside $H$ in Figure \ref{fi:obs3}.
The other cases are similar.
\end{proof}

\begin{figure}[!t]
\centering
\includegraphics[width=0.5\textwidth, page=14]{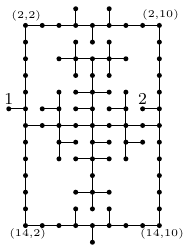}
\caption{Example of an embedding of $H=(1,1)$ where $p_f(G)=+$, vertex $1$ is embedded outside $H$ and vertex $2$ is embedded inside $H$.}\label{fi:obs3}
\end{figure}

Next, we prove the forward direction of the correctness of the reduction. That is, we show that if $B$ is a yes-instance of the {\sc Batteries} problem, then $G_B$ is a yes-instance of \gridEm.

\begin{lemma}\label{lem:BatisGrid}
Let $B=\{B_{i,j}=(x_1^{(i,j)},x_2^{(i,j)})~|~1\leq i\leq m, 1\leq j\leq n\}$ be a yes-instance of the {\sc Batteries} problem. Then $G_B$ is a yes-instance of \gridEm.  

\end{lemma}

\begin{proof}
Let $p$ be an $(m,n)$-placement that is correct and safe. We construct, for every $1\leq i\leq m$ and $1\leq j\leq n$, a grid graph embedding $f_p^{i,j}$ of the gadget $H_{i,j}=(x_1^{(i,j)},x_2^{(i,j)})$ as follows. First, as $p$ is safe, for every $1\leq i\leq m$ there exists $1\leq k_i\leq n$ such that $V_p(i,k_i)=0$. Observe that if $p(i,k_i)=+$, then $x_1^{(i,k_i)}=0$, and if $p(i,k_i)=-$, then $x_2^{(i,k_i)}=0$. Therefore, from Observation \ref{obs:ex1}, for every $1\leq i\leq m$ there exists an embedding $f_p^{i,k_i}$ of $H_{i,k_i}$ with $f_p(H_{i,j})=p(i,j)$, such that vertices $1$ and $2$ are embedded inside $H_{i,k_i}$ in $f_p^{i,k_i}$. Now, from Observation \ref{obs:ex2} we get that for every $1\leq i\leq m$ and for every $1\leq j<k_i$ there exists an embedding $f_p^{i,j}$ of $H_{i,j}$ with $f_p(H_{i,j})=p(i,j)$ such that vertex $1$ is embedded inside $H_{i,j}$ and vertex $2$ is embedded outside $H_{i,j}$ in $f_p^{i,j}$. Similarly, from Observation \ref{obs:ex3} we get that for every $1\leq i\leq m$ and for every $k_i< j\leq n$ there exists an embedding $f_p^{i,j}$ of $H_{i,j}$ with $f_p(H_{i,j})=p(i,j)$ such that vertex $1$ is embedded outside $H_{i,j}$ and vertex $2$ is embedded inside $H_{i,j}$ in $f_p^{i,j}$.

We now show that the embeddings $\{ f_p^{i,k_i}~|~1\leq i\leq m, 1\leq j\leq n \}$ satisfy the conditions of Lemma \ref{lem:condForGrid}. Observe that the first condition is satisfied by the construction of each embedding. In addition, notice that by the construction of the embeddings, the second condition is satisfied if and only if $p_f(H_{i,j})=p_f(H_{i+1,j})$ for every $1\leq i< m, 1\leq j\leq n$. Let $1\leq i< m, 1\leq j\leq n$. Since $p$ is a correct placement, we get that $p(i,j)=p(i+1,j)$, and therefore $p_f(H_{i,j})=p_f(H_{i+1,j})$, so the second condition is satisfied. 

If $j< k_i$ then vertex $2$ of $H_{i,j}$ is embedded outside $H_{i,j}$ in $f_p^{i,j}$ and vertex $1$ of $H_{i,j+1}$ is embedded inside $H_{i,j+1}$ in $f_p^{i,j+1}$. If $j\geq k_i$ then vertex $2$ of $H_{i,j}$ is embedded inside $H_{i,j}$ in $f_p^{i,j}$ and vertex $1$ of $H_{i,j+1}$ is embedded outside $H_{i,j+1}$ in $f_p^{i,j+1}$. Therefore, we get that the third condition is satisfied.
 
Observe that for every $1\leq i\leq k_i$ vertex $1$ of $H_{i,j}$ is embedded inside $H_{i,j}$ in $f_p^{i,j}$. Since $k_i\geq 1$ we get that vertex $1$ of $H_{i,1}$ is embedded inside $H_{i,1}$ in $f_p^{i,1}$. Similarly, observe that for every $k_i\leq i\leq n$ vertex $2$ of $H_{i,j}$ is embedded inside $H_{i,j}$ in $f_p^{i,j}$. Since $k_i\leq n$ we get that vertex $2$ of $H_{i,n}$ is embedded inside $H_{i,n}$ in $f_p^{i,n}$. Therefore we get that the last condition is satisfied. Every condition of Lemma \ref{lem:condForGrid} is satisfied, therefore $G_B$ is a grid graph, and we get that $G_B$ is a yes-instance of \gridEm.                     
\end{proof}

So far, we have proved the correctness of our construction. It only remains to show is that, if $G_B$ is a grid graph, then the distance approximation of $G_B$ is bounded by a constant. We use the fact that in every grid graph embedding of $G_B$ the embedding of the rectangles of the battery gadgets are fixed, as we saw in Observation \ref{obs:fixemb}. We show in the next lemma that if $G_B$ is a grid graph with grid graph embedding $f$, then necessarily $a_f\leq 234$. 

\begin{lemma} \label{lem:boundedPara}
Let $B=\{B_{i,j}=(x_1^{(i,j)},x_2^{(i,j)})~|~1\leq i\leq m, 1\leq j\leq n\}$ be an instance of the {\sc Batteries} problem.  If $G_B$ is a grid graph, then for every grid graph embedding $f$ of $G_B$ it follows that $a_f\leq 234$. 
\end{lemma} 

\begin{proof}
Let $f$ be a grid graph embedding of $G_B$. Assume that $f((0,0))=(0,0)$ $f((1,0))=(1,0)$ and $f((1,1))=(1,1)$. Let $u,v\in V(G_B)$. For $1\leq i,i'\leq m, 1\leq j,j'\leq n$, let $H_{i,j}$ be a closest battery gadget to $u$ and let $H_{i',j'}$ be a closest battery gadget to $v$ in graph distance. Let $u'$ be a vertex from the rectangle of battery gadget $H_{i,j}$ closest to $u$ and let $v'$ be a vertex from the rectangle of battery gadget $H_{i',j'}$ closest to $v$. Consider that the dimension of each rectangle is $13\times 9$ and each one contains a total of $91$ vertices. Observe that $d(u',v') \leq 9(|i-i'|+1)+13(|j-j'|+1)$. Moreover, from Lemma \ref{obs:fixemb} it follows that the embedding of $u'$ and $v'$ is fixed. Therefore $d_f(u',v')\geq 9(|i-i'|-1)+13(|j-j'|-1)$, since that is the grid graph distance of the closest vertices from $H_{i,j}$ and $H_{i',j'}$. If $u$ and $v$ are from the gadgets $H_{i,j}$ and $H_{i',j'}$ then $d_f(u,v)\geq 9(|i-i'|-1)+13(|j-j'|-1)$. Otherwise, one or both are from the $m\times n$-grid frame embedding. Assume that $u$, without loss of generality, is from the $m\times n$-grid frame embedding. Since $u'$ is a closest vertex to $u$ in $H_{i,j}$ and $H_{i,j}$ is a closest battery gadget to $u$, it follows that $d(u,u')\leq 4$. Therefore, since the grid graph distance is bounded by the graph distance, it follows that in that case $d_f(u,u')\leq 4$. So $u$ might be closer in $f$ to the closest vertex in $H_{i',j'}$ by at most $4$ more than from the closest vertex in $H_{i,j}$ to $H_{i',j'}$. Therefore, in the worst case, both $u$ and $v$ are from the $m\times n$-grid frame embedding, and we get that $d_f(u,v)\geq 9(|i-i'|-1)+13(|j-j'|-1)-4-4$. Now, if $u$ is a vertex from the battery gadget $H_{i,j}$ then $d(u,u')\leq 91$ since there are only $91$ vertices in $H_{i,j}$. If not, then $u$ is a vertex from the $m\times n$ grid frame embedding. In that case, since $u'$ is a closest vertex to $u$ in $H_{i,j}$ and $H_{i,j}$ is a closest battery gadget to $u$, it follows that $d(u,u')\leq 4$. In any case we get that $d(u,u')\leq 91$. Similarly, we get that $d(v,v')\leq 91$. We get from the triangle inequality that $d(u,v)\leq d(u,u')+d(u',v')+d(u',v')$. In conclusion, we get that $d(u,v)-d_f(u,v)\leq d(u,u')+d(u',v')+d(u',v')-d_f(u,v)\leq 91+9(|i-i'|+1)+13(|j-j'|+1)+91-(9(|i-i'|-1)+13(|j-j'|-1)-3-3)=91+9+13+91+9+13+4+4=234$. Since $u$ and $v$ are arbitrary vertices, we get that $a_f\leq 234$.  
\end{proof}

In the proof of the next lemma, we invoke Observation \ref{obs:polytimered2}, and Lemmas \ref{lem:BatisGrid}, \ref{lem:EmbtoBat} and \ref{lem:boundedPara} in order to assert the existence of a polynomial reduction from the {\sc Batteries} problem to \gridEm\ where $a_G$ is bounded by a constant. 
  
\begin{lemma}\label{gridisHard}
There exists a polynomial reduction from the {\sc Batteries} problem to \gridEm\ where $a_G$ is bounded by a constant. 
\end{lemma}

\begin{proof}
We show that $\mathsf{reduce}_2$ is such a reduction. From Observation \ref{obs:polytimered2} we get that $\mathsf{reduce}_2$ is computable in polynomial time. Now, let $B=\{B_{i,j}=(x_1^{(i,j)},x_2^{(i,j)})~|~1\leq i\leq m, 1\leq j\leq n\}$ be an instance of the {\sc Batteries} problem. If $B$ is a yes-instance of the {\sc Batteries} problem, then by Lemma \ref{lem:BatisGrid}, we get that $\mathsf{reduce}_2(B)=G_B$ is a yes-instance of \gridEm. If $\mathsf{reduce}_2(B)=G_B$ is a yes-instance of \gridEm, then by Lemma \ref{lem:EmbtoBat} $B$ is a yes-instance of the {\sc Batteries} problem. In addition, by Lemma \ref{lem:boundedPara} we get that $a_{G_B}$ is bounded by $232$ (if $G_B$ is a grid graph). This completes the proof.     
\end{proof}

In conclusion, in Lemma \ref{BatisHard} we proved the existence of a polynomial reduction from {\sc SAT} to the {\sc Batteries} problem. In Lemma \ref{gridisHard} we proved the existence of a polynomial reduction from the {\sc Batteries} problem to \gridEm\ with $a_G$ that is bounded by a constant. Combining these two results, we conclude the correctness of Theorem $1.4$.   

\distanceHard*

\subsection{\bGridEm\ is \FPT\ with Respect to $k+a_G$ on General Graphs}

We present an \FPT\ algorithm with respect to $a_G$ and $k$ for the \bGridEm\ problem. We remark that we do not need to know the value of $a_G$ in advance, in order to use our algorithm, as we iterate over all the potential values for $a_G$. The idea of the algorithm is as follows. We iterate over every possible value for $a_f$, from $1$ to $|V(G)|-2$. For every such $a_f$, we do the following. We guess one of the leftmost vertices $v$ in the $k\times r$ grid, i.e.~such that $\fc(v)=0$. We divide the $k\times r$ grid into $\left \lceil{\frac{r}{k+a_G}}\right \rceil$ ``small rectangles'' of size $k\times (a_G+k)$ from left to right, except the last rectangle that might be smaller. We find the graph distance of each vertex from $v$, i.e. we compute $d(v,u)$ for every $u\in V$. Then, we sort the vertices into the small rectangles as follows. For each vertex $u$, we put $u$ in the $\left \lceil{\frac{d(v,u)}{k+a_G}}\right \rceil$-th rectangle. If there is no such rectangle (i.e. $\left \lceil{\frac{d(v,u)}{k+a_G}}\right \rceil$ is too large), then we put $u$ in the last rectangle, or we can just conclude that we have a no-instance.

Afterwards, we show that in every $k\times r$ grid graph embedding $f$ of $G$ with $a_f=a_G$ where $v$ is one of the leftmost vertices, every $u$ is embedded either in its sorted rectangle or the previous one. An intuition for this is similar to our former argument about the structure of a grid graph embedding with distance approximation $a_f$: On the one hand, the vertex $u$ cannot be embedded into a farther rectangle as its shortest path (or paths) from $v$ cannot be embedded in that case; on other hand, if $u$ is embedded into a closer rectangle, then the embedding does not respect the distance approximation. We will prove this claim formally; here we just give an overview. Then, after we know the approximate ``location'' of each vertex, we try to find a $k\times r$ grid graph embedding of $G$. 

We denote the set of vertices located between the $s$-th column and the $t$-th column in a $k\times r$ grid graph embedding $f$ of $G$ by $C_f(s,t)$, i.e. $C_f(s,t)=\{u\in V~:~s\leq \fc(u)\leq t \}$. In addition, for any two integers $0\leq s\leq t$, we denote the set of vertices with graph distance at least $s$ and at most $t$ from $v$ by $D_v(s,t)$, i.e. $D_v(s,t)=\{u\in V~|~s\leq d(v,u)\leq t \}$. See Figure~\ref{fi:distanceFPTK}.

\begin{figure}[!ht]
\centering
\includegraphics[width=0.5\textwidth, page=6]{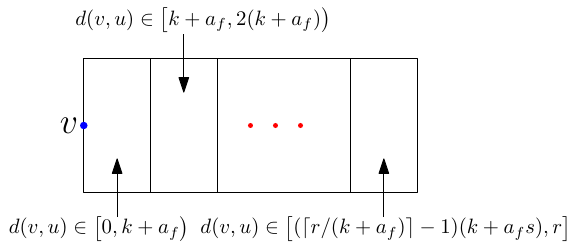}
\caption{The partition of the $k\times r$ grid into small rectangles and the approximate position of any vertex $u$ with respect to $v$.}\label{fi:distanceFPTK}
\end{figure}
 
\begin{lemma}\label{lem:distanceLem1}
Let $G=(V,E)$ be a $k\times r$ grid graph, and let $f$ be a $k\times r$ grid graph embedding of $G$. Let $v\in V$ such that $\fc(v)=0$. Let $u\in V$. Then, $u\in C_f(d(u,v)-a_f-k, d(u,v))$ and $u\in D_v(\fc(u),\fc(u)+a_f+k)$.  
\end{lemma}  
\begin{proof}
First, we show that $u\in C_f(d(u,v)-a_f-k, d(u,v))$. Observe that for every $u\in V$, we have that $d(u,v)\geq d_f(u,v)\geq |\fc(u)-\fc(v)|=|\fc(u)-0|=\fc(u)$. In addition, notice that $|\fr(u)-\fr(v)|\leq k$. Therefore, we get that $d_f(u,v)=|\fc(u)-\fc(v)|+|\fr(u)-\fr(v)|\leq |\fc(u)-\fc(v)|+k=\fc(u)+k$. Now, we have that $d(u,v)-d_f(u,v)\leq a_f$. Therefore, $d(u,v)-\fc(u)-k\leq d(u,v)-d_f(u,v)\leq a_f$. Combining these two inequalities, we get that $d(u,v)-a_f-k\leq \fc(u)\leq d(u,v)$. So we have that $u\in C_f(d(u,v)-a_f-k, d(u,v))$.

Second, we show that $u\in D_v(\fc(u),\fc(u)+a_f+k)$. We have that $a_f$ is the distance approximation of $f$, therefore we get that $d(u,v)\leq d_f(u,v)+a_f$. We saw that $d_f(u,v)\leq \fc(u)+k$. Therefore we get that $d(u,v)\leq \fc(u)+a_f+k$. In addition, we also saw that $d(u,v)\geq \fc(u)$. Combining these two inequalities, we get that $\fc(u)\leq d(u,v)\leq \fc(u)+a_f+k$. So we have that $u\in D_v(\fc(u),\fc(u)+a_f+k)$.  
\end{proof}

In Lemma \ref{lem:distanceLem1} we assume that there exists a vertex $v$ with $\fc(v)=0$. Notice that for every $k\times r$ grid graph embedding $f$ of $G$, we can construct a $k\times r$ grid graph embedding $f'$ of $G$ with $a_{f'}=a_f$ and with such a vertex $v$: indeed, we can simply define $f'(u)=(\fr(u),\fc(u)-m)$ where $m=\mathsf{Min}_{v\in V} \{\fc(v)\}$; then, observe that $f'$ is a $k\times r$ grid graph embedding of $G$ with $a_{f'}=a_f$ and with $\mathsf{Min}_{v\in V} \{f'_\mathsf{col}(v)\}=0$. We state this in the next observation.
\begin{observation} \label{obs:VerInfirstColumn}
Let $G=(V,E)$ be a $k\times r$ grid graph, and let $f$ be a $k\times r$ grid graph embedding of $G$. Then, there exists a $k\times r$ grid graph embedding $f'$ of $G$ with $a_{f'}=a_f$ and $\mathsf{Min}_{v\in V} \{f'_\mathsf{col}(v)\}=0$. 
\end{observation} 

Now we further discuss the idea of the algorithm.  We use dynamic programming in order to find a $k\times r$ grid graph embedding of $G$ in the following manner. First, we get the approximate location for each vertex as explained earlier. Then, starting from left to right, we try to find an embedding for each of the ``small rectangles'' that agrees with one of the embeddings for the previous one. In order to do this, we only need to know which vertices from the current rectangle we have already embedded in the previous one, and what is the embedding in the last column. Notice that, each time, we need to store information whose size depends only on $r$ and $a_G$. In this way, after the last iteration corresponding to the last rectangle, we can get the embedding of $G$ (if such exists), or conclude that there is no such an embedding.   

In what follows, we will need the following notations:
\begin{itemize}
\item For a $U\subseteq V$, we denote the set of vertices of $U$ that have neighbors from $V\setminus U$ by $N_{\mathsf{O}}(U)$, i.e. $N_{\mathsf{O}}(U)=\{ u\in U$: there exists $ \{u,v\}\in E$ where $ v\in V\setminus U \}$.   
\item For a $k\times r$ grid graph embedding $f$ of $G$, we refer to the set $\{ (v,(i,j))\in f~|~j=0\}$ as the {\em left column} of $f$. Similarly, we refer to the set $\{ (v,(i,j))\in f~|~j=r\}$ as the {\em right column} of $f$. Sometimes we refer to these sets simply as columns. In addition, we refer to the column that is immediately ``left'' to the right column (i.e. $\{ (v,(i,j))\in f~|~j=r-1\}$), as the {\em $-1$-right column}.
\item For a column $\mathsf{column}$, we denote the set of vertices within it by $V(\mathsf{column})$, i.e. $V(\mathsf{column})=\{ v\in V~|~(v,(i,j))\in \mathsf{column}$ for some $i,j\}$. 
\item For $U\subseteq V$, we denote the set $\{ f~|~f$ is a $k\times 1$ grid graph embedding of $G[U']$ for some $U'\subseteq U\}$ by $\mathsf{AllColEmbd}(U,k)$.
\end{itemize}

In the next lemma we show that $|\mathsf{AllColEmbd}(U,k)|$ is bounded by $(|U|k)^{\OO(k)}$.

\begin{lemma}\label{lem: AllEmbAna}
Let $G$ be a graph, $k\in \mathbb{N}$ and $U\subseteq V(G)$. Then, $|\mathsf{AllColEmbd}(U,k)|=\OO((|U|k)^{\OO(k)})$.
\end{lemma}                

\begin{proof}
Observe that, for every $0\leq k'\leq k$, the number of possibilities to choose $k'$ elements from $U$ is $\binom{|U|}{k'}\leq |U|^{k'}\leq |U|^{k}$. In addition, the number of possibilities to order $k'$ elements in $k$ places is bounded by $k!\leq k^k$. Therefore the total number of possibilities to choose at most $k$ elements and then order them in $k$ places is bounded by $\sum_{k'=0}^{k}|U|^{k}k^k \leq (k+1)|U|^{k}k^k=(|U|k)^{\OO(k)}$.   
\end{proof}

\begin{algorithm} [t!]
    \SetKwInOut{Input}{Input}
    \SetKwInOut{Output}{Output}
	\medskip
    {\textbf{function} $k\times r \mathsf{Grid Recognition Iteration}$}$(\langle \mathsf{PreviousInf}, \mathsf{CurrentRect}, \mathsf{NextRect},a_f, k, r\rangle)$\;
		 $\mathsf{CurrentInfo}\gets \emptyset$ \;
				\For{every $(\mathsf{CurrentUsed},\mathsf{LeftC},\mathsf{-1LeftC})\in \mathsf{PreviousInf}$}
				{
				\For{every subset $\mathsf{NextRectUse}$ of $\mathsf{NextRect}$ }
				{
				\For{every $(\mathsf{RightC}\in \mathsf{AllColEmbd}((\mathsf{CurrentRect}\setminus \mathsf{CurrentUsed})\cup \mathsf{NextRectUse},k)$}
				{
				Use brute force to find (if  such exists) a $k\times (a_f+k)$ grid graph embedding $f'$ of $G(\mathsf{CurrentRect}\setminus \mathsf{CurrentUsed}\cup \mathsf{NextRectUse}\cup \mathsf{LeftC})$ with left column $\mathsf{LeftC}$ and right column $\mathsf{RightC}$ and $N_{\mathsf{O}}(\mathsf{CurrentRect}\setminus \mathsf{CurrentUsed}\cup \mathsf{NextRectUse}\cup \mathsf{LeftC}\cup \mathsf{-1LeftC})\cap(\mathsf{CurrentRect}\setminus \mathsf{CurrentUsed}\cup \mathsf{NextRectUse}\cup \mathsf{LeftC}\cup \mathsf{-1LeftC})\subseteq V(\mathsf{RightC})$\;
				\If{found such an embedding}
				{
						$\mathsf{-1RightC}\gets$ $-1$-right column of $f'$\; 
					 $\mathsf{CurrentInfo}\gets \mathsf{CurrentInfo}\cup \{(\mathsf{NextRectUse},\mathsf{RightC},\mathsf{-1RightC})\}$\;
				}
					
				}
				}
				}
				\Return $\mathsf{CurrentInfo}$\; 
				\caption{$\mathsf{k\times r Grid Recognition Iteration}$}
    \label{alg:Grid Recognition Iteration}
\end{algorithm}

Algorithm \ref{alg:Grid Recognition Iteration} is where we find embeddings for a current small rectangle, using the information of its previous small rectangle. Each time we use the algorithm, we seek $k\times(a_f+k)$ embeddings for the vertices in the current rectangle and for a subset $U$ of the vertices of the next rectangles. For each such embedding, we only need to store the following information to proceed to the next rectangle: the subset $U$ and the ``right column'' of the embedding (so as to ``glue'' embeddings of adjacent rectangles properly). In the next rectangle, we try to embed (using brute-force) the vertices we did not embedd in the previous rectangle (those outside $U$) and some of the vertices of its next rectangle, such that the left column of the current embedding ``agrees'' with the right column of the embedding of the previous rectangle. Notice that at each step we only store an ``\FPT\ amount'' of information, and so we will eventually achieve a fixed-parameter algorithm. We summarize the information the algorithm returns and analyze its runtime in the next observation.

\begin{observation} \label{obs: runtime iteration}
Let $G=(V,E)$ be a connected graph, $v\in V$, $k,r,a\in \mathbb{N}$ where $a<|V|$, and $U, U',U''\subseteq V$ with $|U|,|U'|,|U''|\leq 2k(k+a)$. Let $\mathsf{PreviousInf}\subseteq \{(A',\mathsf{LeftColumn},\mathsf{-1LeftC})~|~ A'\subseteq U', \mathsf{LeftColumn}\in \mathsf{AllColEmbd}(A'\cup U,k)\}$. Then, Algorithm \ref{alg:Grid Recognition Iteration} on the input $(\langle \mathsf{PreviousInf}, U', U'',a, k, r\rangle)$ runs in time $(ka)^{\OO(ka+k^2)}$, and returns the set $\{(A'',\mathsf{RightColumn},\mathsf{-1RightC})~|~A''\subseteq U'',$ there exist $(A',\mathsf{LeftColumn})\in \mathsf{PreviousInf}$ and a $k\times (a+k)$ grid graph embedding of $G[(U'\setminus A')\cup A'']$ with left column $\mathsf{LeftColumn}$, right column $\mathsf{RightColumn}$ and $-1$ right column $\mathsf{-1RightC}$ and $N_{\mathsf{O}}(\mathsf{NextRectUse})\subseteq \mathsf{RightColumn}\cup \mathsf{LeftColumn}\}$.   
\end{observation}

\begin{proof}
The proof for the output of the algorithm trivially follows from the pseudocode. As for the runtime of Algorithm \ref{alg:Grid Recognition Iteration}, observe that, since $|U|,|U'|,|U''|\leq 2k(k+a)$, there are at most $2^{2k(k+a)}$ different subsets of $U'$. For every such subset $A'$, $|A'\cup U|\leq 4k(k+a)$. Notice that for a set of vertices $B$ with $|B|\leq 4k(k+a)$, by Lemma \ref{lem: AllEmbAna}, we get that $|\mathsf{AllColEmbd}(B,k)|=(2k(k+a)k)^{\OO(k)}=(k+a)^{\OO(k)}$. Therefore, we get that $|\mathsf{PreviousInf}|=\OO((k+a)^{\OO(k)}a)=(k+a)^{\OO(k)}$. So, we get that the algorithm performs at most $(k+a)^{\OO(k)}\cdot 2^{2k(k+a)}\cdot (k+a)^{\OO(k)}= (k+a)^{\OO(k+a)}$ iterations. At each iteration, we look at all possible $k\times (a+k)$ grid graph embeddings. Notice that the number of these grid graph embeddings is bounded by $(ka+k^2)^{ka+k^2}$ (see a similar analyze in the proof of Lemma \ref{lem: AllEmbAna}). For a given grid graph embedding, the algorithm works with runtime $\OO(k^2+ka)$. Therefore, the algorithm runs in time $(ka)^{\OO(ka+k^2)}$. 
\end{proof}

\begin{algorithm}[t!]
    \SetKwInOut{Input}{Input}
    \SetKwInOut{Output}{Output}
	\medskip
    {\textbf{function} $k\times r \mathsf{ Grid Recognition with }a_f$}$(\langle G=(V,E), k, r, a_f\rangle)$\;
		\For{every $v\in V$  \label{algo3:every v}} 
		{\label{alg3:distance}
			Find $d(u,v)$ for every $u\in V$ \; 
			\If{there exists $u\in V$ with $d(u,v)-a_f>r$ \label{alg3:smallDs}} 
			{\Return $\langle G,k,r,a_f \rangle $ is a no-instance\;} 
			\For{every $i\in [\left \lceil{\frac{r}{k+a_G}}\right \rceil]+1$}
			{
			Create an empty set $D_i$\;
			}
			\For {every $u\in V$}
			{$i\gets \left \lfloor{\frac{d(u,v)}{k+a_f}}\right \rfloor$\;
			$D_i\gets D_i\cup \{ u\}$\;}
			\If {there exists $D_i$ with $|D_i|>2k(k+a_f)$}
			{Break\;}
			Let $j$ be the maximum such that $D_j\neq \emptyset$\;

			$\mathsf{CurrentInfo}\gets (\emptyset,\mathsf{emptyColumn})$\;
			$i\gets 0$
			\While{$\mathsf{CurrentInfo}\neq \emptyset$ and $i\leq j+1$} 
			{\label{algo3: while}
				$\mathsf{CurrentInfo}\gets k\times r \mathsf{Grid Recognition Iteration}(\langle \mathsf{CurrentInfo}, D_i, D_{i+1},a_f, k, r\rangle)$\;
				$i\gets i+1$\;
			}
			\If{$\mathsf{CurrentInfo}= \emptyset$ \label{alg: not empty}} 
			{Break\;}
			\Return $\langle G,k,r,a_f \rangle $ is a yes-instance\; \label{alg: yes}
		}
		\Return $\langle G,k,r,a_f \rangle $ is a no-instance\;

    \caption{$\mathsf{k\times r Grid Recognition with }a_f$}
    \label{alg:Grid Recognition a_f}
\end{algorithm}

Algorithm \ref{alg:Grid Recognition a_f} (presented ahead) is where we guess a leftmost vertex, sort the vertices into the ``small rectangles'' and then use Algorithm \ref{alg:Grid Recognition Iteration} in order to find a $k\times r$ grid graph embedding of $G$. 
        
We denote the set $\mathsf{CurrentInfo}$ at the $i$-th iteration of Algorithm \ref{alg:Grid Recognition Iteration} by $\mathsf{CurrentInfo}_i$.
Similarly, we denote the set $\mathsf{PreviousInfo}$ at the $i$-th iteration of Algorithm \ref{alg:Grid Recognition Iteration} by $\mathsf{PreviousInfo}_i$.
Observe that, for $i>1$, it follows that $\mathsf{PreviousInfo}_{i}=\mathsf{CurrentInfo}_{i-1}$.
At each iteration, we store the information we need from the last iteration in $\mathsf{CurrentInfo}_{i-1}$. Every element in $\mathsf{CurrentInfo}$ represents a way to embed the last rectangle, in a way that agrees with some of the embeddings of the last rectangle. In particular, we show in the next lemma that if $\mathsf{CurrentInfo}_i\neq \emptyset$ then there exists a grid graph embedding of the set of vertices that are in the rectangles we analyzed so far.  

\begin{lemma} \label{lem: notEmptyIsGrid}
Let $G=(V,E)$ be a connected graph, $k,r,a\in \mathbb{N}_0$. When Algorithm \ref{alg:Grid Recognition a_f} is called on the input $<G,k,r,a>$, for every $i$ it follows that for every $(A',\mathsf{rightColumn},\mathsf{-1RightC})\in \mathsf{PreviousInf}_i$ there exists a $k\times i(a+k)$ grid graph embedding of $G[D_1\cup\ldots\cup D_i\cup A']$ with right column $\mathsf{rightColumn}$ and $N_{\mathsf{O}}(D_1\cup\ldots\cup D_i\cup A')\subseteq V(\mathsf{rightColumn})$.
\end{lemma}

\begin{proof}
We prove this lemma by induction on $i$. For $i=1$ the claim is trivial. Let $i>1$ and let $(A',\mathsf{rightColumn},\mathsf{-1RightC}')\in \mathsf{PreviousInf}_i$. From Observation \ref{obs: runtime iteration}, we know that there exists $(A,\mathsf{Column},\mathsf{-1RightC})\in \mathsf{PreviousInf}_{i-1}$ such that $G[(D_i\setminus A)\cup A']$ is a $k\times (a+k)$ grid graph with grid graph embbeding $f''$ with left column $\mathsf{Column}$ and right column $\mathsf{rightColumn}$ and $N_{\mathsf{O}}((D_i\setminus A)\cup A'\cup \mathsf{-1RightC})\cap((D_i\setminus A)\cup A'\cup)\subseteq \mathsf{rightColumn}$. From the induction hypothesis, we get that there exists a $k\times (i-1)(a+k)$ grid graph embedding $f'$ of $G[D_1\cup\ldots\cup D_{i-1}\cup A]$ with right column $\mathsf{Column}$ and $N_{\mathsf{O}}(D_1\cup\ldots\cup D_{i-1}\cup A)\subseteq V(\mathsf{Column})$. We can assume that $f''_\mathsf{col}(\mathsf{Column})=0$, and $f'_\mathsf{col}(v)=0$ for some vertex $v$. Observe that $(D_i-A\cup A')\cap (D_1\cup\ldots\cup D_{i-1}\cup A)=V(\mathsf{Column})$. We define a $k\times i(a+k)$ grid graph embedding $f$ of $G[D_1\cup\ldots\cup D_i\cup A']$. If $v\in (D_1\cup\ldots\cup D_{i-1}\cup A)$, then $f(v)=f'(v)$; otherwise $f(v)=(f''_\mathsf{row}(v),f''_\mathsf{col}(v)+f''_\mathsf{col}(\mathsf{Column}))$. Observe that $f$ is an injection. For $\{u,v\}\in E(G[D_1\cup\ldots\cup D_i\cup A'])$, if $u,v\in D_1\cup\ldots\cup D_{i-1}\cup A$ then $d_f(u,v)=1$ since $f'$ is a $k\times i(a+k)$ grid graph embedding of $G[D_1\cup\ldots\cup D_{i-1}\cup A]$. Similarly, if $u,v\in D_i-A\cup A'$ then $d_f(u,v)=1$ since $f''$ is a $k\times (a+k)$ grid graph embedding of $G[D_i-A\cup A']$. Assume that $u \in D_1\cup\ldots\cup D_{i-1}\cup A$ and $v\in D_i-A\cup A'$. Then, $u \in V(\mathsf{Column})\subseteq D_1\cup\ldots\cup D_{i-1}\cup A$, so $d_f(u,v)=1$ since $f'$ is a $k\times i(a+k)$ grid graph embedding of $G[D_1\cup\ldots\cup D_{i-1}\cup A]$. Now, let $v\in N_{\mathsf{O}}(D_1\cup\ldots\cup D_i\cup A')$. If $v\in D_1\cup\ldots\cup D_{i-1}\cup A$ then $v\in N_{\mathsf{O}}(D_1\cup\ldots\cup D_{i-1}\cup A)$. Therefore, we have that $v\in V(\mathsf{Column})$. Since $v\notin V(\mathsf{rightColumn})$, we have that $v\notin N_{\mathsf{O}}((D_i\setminus A)\cup A'\cup \mathsf{-1RightC})\cap((D_i\setminus A)\cup A')$, therefore $v\notin N_{\mathsf{O}}(D_1\cup\ldots\cup D_i\cup A')$, a contradiction. Thus, $v\in (D_i\setminus A)\cup A'$ and by Observation \ref{obs: runtime iteration} $v\in V(\mathsf{rightColumn})$. This completes the proof.
\end{proof}

In the next lemma, we show the opposite direction of Lemma \ref{lem: notEmptyIsGrid}. That is, we show that if there exists a grid graph embedding of $G$, then for every iteration $i$ we have that $(C_f(i)\setminus D_v(i),\mathsf{rightColumn}(C_f(i))$ $,\mathsf{-1rightColumn}(C_f(i))\in \mathsf{CurrentInfo}_i$, and therefore $\mathsf{CurrentInfo}_i\neq \emptyset$.

\begin{lemma}\label{lem:if grid then not empty}
Let $G=(V,E)$ be a $k\times r$ grid graph with a $k\times r$ grid graph embedding $f$, and let $a\geq a_f$. Assume that $f_j(v)=0$ for some $v\in V$. Then, when Algorithm \ref{alg:Grid Recognition Iteration} is called on the input $(\langle G=(V,E), k, r, a\rangle)$, in the iteration corresponding to $v$, it follows that for every $i$ $(C_f(i)\setminus D_v(i),\mathsf{rightColumn}(C_f(i)),\mathsf{-1rightColumn}(C_f(i))\in \mathsf{CurrentInfo}_i$.
\end{lemma}

\begin{proof}
We prove this lemma by induction on $i$. For $i=1$ the claim is trivial. Now, let $i>1$. By the inductive hypothesis, $(C_f(i-1)\setminus D_v(i-1),\mathsf{rightColumn}(C_f(i-1))\in \mathsf{CurrentInfo}_{i-1}$. From Lemma \ref{lem:distanceLem1}, we conclude that $C_f(i)\subseteq D_v(i)\cap D_v(i+1)$. Therefore, $C_f(i)\setminus D_v(i)\subseteq D_v(i+1)$. Observe that $(D_i\setminus C_f(i-1))\setminus D_v(i-1)\cap (C_f(i)\setminus D_v(i))=C_f(i)$. Since $f$ is a $k\times r$ grid graph embedding of $G$, there exists a $k\times (a+k)$ grid graph embedding of $G(C_f(i))$. Notice that $\mathsf{rightColumn}(C_f(i-1)=\mathsf{leftColumn}(C_f(i))$, and $O(C_f(i)\cup \mathsf{-1rightColumn}(C_f(i-1)))\cap C_f(i)\subseteq \mathsf{rightColumn}(C_f(i))$. Therefore, by Observation \ref{obs: runtime iteration}, we get that $(C_f(i)\setminus D_v(i),\mathsf{rightColumn}(C_f(i)), \mathsf{-1rightColumn}(C_f(i))) \in \mathsf{CurrentInfo}_i$.     
\end{proof}

\begin{algorithm}[t!]
    \SetKwInOut{Input}{Input}
    \SetKwInOut{Output}{Output}
	\medskip
    {\textbf{function} $k\times r \mathsf{Grid Recognition Iteration}$}$(\langle G=(V,E), k, r\rangle)$\;
		 $a_G\gets 0$\;
		\While{$a_G<|V|$\label{alg4:while}}
		{
		\If{$k\times r \mathsf{ Grid Recognition with }a_f(\langle G=(V,E), k, r, a_G\rangle)$ is a yes-instance \label{alg4:checkAg}}
		{
			\Return $\langle G,k,r \rangle $ is a yes-instance with $a_G$\;
		}
		$a_G\gets a_G+1$\;
		}
		\Return $\langle G,k,r \rangle $ is a no-instance\;
    	 
    \caption{$\mathsf{k\times r Grid Recognition}$}
    \label{alg:Grid Recognition}
\end{algorithm}

Now we use Lemmas \ref{lem: notEmptyIsGrid} and \ref{lem:if grid then not empty} to show that for a given $a\in \mathbb{N}_0$, Algorithm \ref{alg:Grid Recognition a_f} reports that $G$ is a $k\times r$ grid graph if and only if there exists a $k\times r$ grid graph embedding $f$ of $G$ with $a_f\leq a$.  

\begin{lemma}\label{lem45}
Let $G=(V,E)$ be a connected graph, let $k,r\in \mathbb{N}$ and let $a\in \mathbb{N}, a<|V|$. Then, Algorithm \ref{alg:Grid Recognition a_f} with input $(\langle G=(V,E), k, r, a\rangle)$ works in runtime $\OO (|V|^2(ka)^{\OO(ka+k^2)})$, and returns ``yes-instance'' if and only if there exists a grid graph embedding $f$ of $G$ with $a_f\leq a$. 
\end{lemma}

\begin{proof}
First, assume that there exists a grid graph embedding $f$ of $G$ with $a_f\leq a$. From Observation \ref{obs:VerInfirstColumn} we can assume that there exists $v\in V$ with $\fc(v)=0$. Because Algorithm \ref{algo3:every v} iterates over every $u\in V$, there is an iteration with $v$ such that $\fc(v)=0$. In that iteration, from Lemma \ref{lem:if grid then not empty} we get that $\mathsf{CurrentInfo}_i\neq \emptyset$ for every $i$. Therefore, in line \ref{alg: not empty} of Algorithm \ref{alg:Grid Recognition a_f}, the condition is false, and the algorithm skips to line \ref{alg: yes} and returns ``yes-instance''. 
Now, assume that the algorithm returns ``yes-instance''. Thus, in line \ref{alg: not empty}, the condition is false, so we get that $\mathsf{CurrentInfo}_{j+1}\neq \emptyset$. From Lemma \ref{lem: notEmptyIsGrid}, we get that there exists a $k\times r$ grid graph embedding $f$ of $G$ with $a_f\leq a$.

As for the runtime, for a chosen vertex $v\in V$ in line \ref{algo3:every v} of Algorithm \ref{alg:Grid Recognition a_f}, computing $d(u,v)$ in line \ref{alg3:distance} for every $u\in V$ can be done with runtime $\OO (|V|^2)$. From line \ref{alg3:smallDs}, we get that unless the iteration breaks, $|D_i|\leq 2k(k+a_f)$. Therefore, by Observation \ref{obs: runtime iteration}, each iteration in the while loop in line \ref{algo3: while} takes runtime $\OO(ka)^{\OO(ka+k^2)}$, and there are at most $r\leq |V|$ such iterations. The other operations in the algorithm can be done with runtime $\OO(|V|)$. In conclusion, we get that the algorithm works with runtime $\OO (|V|^2(ka)^{\OO(ka+k^2)})$ and returns ``yes-instance'' if and only if there exists a grid graph embedding $f$ of $G$ with $a_f\leq a$.                
\end{proof}

Finally, we are ready to show that Algorithm \ref{alg:Grid Recognition} solves the \bGridEm\ problem with FPT runtime with respect to $r+a_G$.
The algorithm, in every iteration, activates Algorithm \ref{alg:Grid Recognition a_f} with a larger $a_f$ than before. If $G$ is a $k\times r$ grid graph, then by Lemma \ref{lem45}, Algorithm \ref{alg:Grid Recognition a_f} will report this in the $a_G$-th iteration. Otherwise, in every iteration the algorithms fails to prove that $G$ is a $k\times r$ grid graph, and it will report that $(G,k,r)$ is a no-instance. 

\begin{lemma}\label{lem:kralgoproof}
Let $G=(V,E)$ be a connected graph, and $k,r\in \mathbb{N}$. Then, Algorithm \ref{alg:Grid Recognition} with input $\langle G=(V,E), k, r\rangle$ runs in time $\OO (|V|^2(ka_G)^{\OO(ka_G+k^2)})$ and returns ``yes-instance'' if and only if $G$ is a $k\times r$ grid graph.  
\end{lemma}

\begin{proof}
First, notice that every iteration in the while loop in line \ref{alg4:while} of Algorithm \ref{alg:Grid Recognition} runs Algorithm \ref{alg:Grid Recognition a_f} on the input $(\langle G=(V,E), k, r, a_G\rangle)$. Therefore, by Lemma \ref{lem45}, the $a_G$-th iteration runs in time $\OO (|V|^2(ka_G)^{\OO(ka_G+k^2)})$. Assume that $G$ is a $k\times r$ grid graph. Then, there exists a $k\times r$ grid graph embedding $f$ of $G$ with $a_f=a_G$. Therefore, in the $a_f$-th iteration the algorithm activates Algorithm \ref{alg:Grid Recognition a_f} on the input $(\langle G=(V,E), k, r, a_f\rangle)$, and, by Lemma \ref{lem45}, returns yes-instance. Therefore, in that case, the algorithm runs $a_f$ iterations, where each iteration runs in time at most $\OO (|V|^2(ka_f)^{\OO(ka_f+k^2)})$. Therefore the total runtime is $a_f\cdot \OO (|V|^2(ka_f)^{\OO(ka_f+k^2)})=\OO (|V|^2(ka_f)^{\OO(ka_f+k^2)})=\OO (|V|^2(ka_G)^{\OO(ka_G+k^2)})$. Now, assume that Algorithm \ref{alg:Grid Recognition} returns ``no-instance''. Then, for every $0\leq a_G\leq |V|$, it follows that Algorithm \ref{alg:Grid Recognition a_f} returns ``no-instance''. Thus, by Lemma \ref{alg4:while}, there does not exist a $k\times r$ grid graph embedding $f$ of $G$ with $0\leq a_f\leq |V|$. By Observation \ref{obs:disapp bounded}, we get that $G$ is not a $k\times r$ grid graph. Notice that, in this case, we get that $a_G=|V|$. Furthermore, in that case, the algorithm runs $|V|$ iterations, where each iteration runs in time at most $\OO (|V|^2(k|V|)^{\OO(k|V|+k^2)})$. Therefore, the total runtime is $|V|\cdot \OO (|V|^2(k|V|)^{\OO(k|V|+k^2)})=a_G\cdot \OO(|V|^2(ka_G)^{\OO(ka_G+k^2)})=\OO (|V|^2(ka_G)^{\OO(ka_G+k^2)})$. In conclusion, we get that Algorithm \ref{alg:Grid Recognition} on the input $\langle G=(V,E), k, r\rangle$ runs in time $\OO (|V|^2(ka_G)^{\OO(ka_G+k^2)})$ and returns ``yes-instance'' if and only if $G$ is a $k\times r$ grid graph problem.           
\end{proof}

Having proved that Algorithm \ref{alg:Grid Recognition} solves the \bGridEm problem with \FPT\ runtime, we can conclude the correctness of Theorem $1.5$.

\distanceFPT*

\subsection{\bGridEm\ Is \FPT\ with Respect to $a_G$ on Trees}
Let $T$ be a tree that is a grid graph, with a grid graph embedding $f'$ of $T$. First, we show that $f'$ has a special structure, that depends on $a_{f'}$. We begin with a simple example. For this purpose, we present a way to define ``directions'' in a given grid graph embedding. Let $G$ be a connected grid graph, with grid graph embedding $f$ of $G$, and let $\{u,v\}\in E(G)$ be an edge. We say that $\{u,v\}$ is {\em directed up} in $f$ if $\fr(v)=\fr(u)+1$. Similarly, we say that $\{u,v\}$ is {\em directed down} in $f$, {\em directed right} in $f$, or {\em directed left} in $f$, if $\fr(v)=\fr(u)-1$, $\fc(v)=\fc(u)+1$, or $\fc(v)=\fc(u)-1$, correspondingly. Notice that, since $f$ is a grid graph embedding of $G$, exactly one of these four options is valid.

It is easy to see that every path is a grid graph. Consider a path $P=(v_1,v_2,\ldots v_\ell)$, where $\ell\geq 2$, and a grid graph embedding $f$ of $P$ with $a_f$=0. Without loss of generality, assume that $\{v_1,v_2\}$ is directed up in $f$. Assume that there exists $2\leq t\leq \ell$ such that $\{v_{t-1},v_{t}\}$ is not directed up in $f$, and let $s$ be the minimum integer that satisfies this condition. Because $\{v_{s-2},v_{s-1}\}$ is directed up, notice that $\{v_{s-1},v_{s}\}$ cannot be directed down. Assume that $\{v_{s-1},v_{s}\}$ is directed right in $f$. We show that, for every $1\leq t\leq \ell-1$, $\{v_t,v_{t+1}\}$ is directed up or right in $f$. An intuition for that is, every time we ``go forward'' in the path, we increase the graph distance between the starting vertex and the current vertex by one, and therefore we need to increase the grid graph embedding by one, otherwise we get that $a_f>0$. Thus, we cannot ``go back''. For a general grid graph embedding $f$ of $P$, we show that there are at most two ``main directions'' in which we move from each $v_i$ to $v_{i+1}$, and we cannot go in the other directions, more then $a_f$ times. Observe that given a general connected grid graph $G$, and a path $P=(v_1,v_2,\ldots v_\ell), v_i\in V(G)$, such that $P$ is a shortest path between $v_1$ to $v_\ell$, we can argue the same claim. Therefore, we conclude that, given a tree $T$ that is a grid graph, with grid graph embedding $f$, for every path $P\subseteq T$ it follows that there are two ``main directions'' in which the path ``moves''. We start by proving this idea. Towards that, we define the terms of ``directions'':

\begin{definition}[{\bf Grid Direction for Edges}]\label{def:Grid direction}
Let $G$ be a connected grid graph, with grid graph embedding $f$ of $G$. Let $\{u,v\}\in E(G)$. We say that the {\em grid direction of $\{u,v\}$ in $f$} is {\em up}, {\em down}, {\em right}, or {\em left}, if $\fr(v)=\fr(u)+1$, $\fr(v)=\fr(u)-1$, $\fc(v)=\fc(u)+1$, or $\fc(v)=\fc(u)-1$, correspondingly. We refer to $\mathsf{up}, \mathsf{down}, \mathsf{right}$ and $\mathsf{left}$ as {\em directions}. We denote the set of directions by $\mathsf{dir}$, i.e. $\mathsf{dir}=\{ \mathsf{up}, \mathsf{down}, \mathsf{right}, \mathsf{left}\}$.       
\end{definition}
  
We refer to the up and down directions as {\em vertical directions}, and to the right and left directions as {\em horizontal directions}.



Now we prove the claim in the example we saw before. Let $T$ be a tree, and let $f$ be a grid graph embedding of $T$. For a simple path $P=(v_1,\ldots,v_\ell)$, we consider the direction of the edges $\{\{v_i,v_{i+1}\}~|~1\leq i\leq \ell-1\}$. For a set of one or more directions $\Delta \subseteq \mathsf{dir}$, we say that {\em $P$ moves mainly in directions $\Delta$ in $f$} or {\em $P$ moves in directions $\Delta$ in $f$} if the following condition is satisfied: The number of edges with direction $d'\notin \Delta$ in $f$ is equal or less than $a_f+1$. Sometimes, we only refer to the number of direction in $\Delta$. In the next lemma, we show that $P$ moves in two directions, one vertical and one horizontal.   

\begin{lemma} \label{lem: directionsINPath}
Let $T$ be a tree that is a grid graph, with a grid graph embedding $f$ of $T$. Then, every simple path moves in $f$ in at most one vertical direction and at most one horizontal directions.
\end{lemma}

\begin{proof}
Let $P=(u,w_1\ldots,w_\ell, v)$ be a simple path. $T$ is a tree, therefore, $P$ is a shortest path, from $u$ to $v$, so $|P|-1=d(u,v)$. Assume towards a contradiction, without loss of generality, that $P$ moves in two vertical directions more than $a_f+1$ steps. Observe that $d_f(u,v)=|\fc(u)-\fc(v)|+|\fr(u)-\fr(v)|$. Denote the number of edges in $P$ directed up in $f$, and the number of edges in $P$ directed down in $f$, by $D_U(f,P)$ and $D_D(f,P)$, correspondingly. Notice that $\fc(v)=\fc(u)+D_U(f,P)-D_D(f,P)$. Since there at least $a_f+1$ edges in $P$ directed up, and at least $a_f+1$ edges in $P$ directed down, it follows that $d_f(u,v)\leq |P|-a_f-1=d(u,v)-a_f-1$. Therefore, we get that $d(u,v)-d_f(u,v)\geq a_f+1$. This is a contradiction to the definition of $a_f$.    
\end{proof}

Assume that we have a path $P=(v_1,\ldots,v_n)$, and a positive integer $a$, and we aim to construct a grid graph embedding $f$ of $P$ with $a_f\leq a$ in \FPT\ runtime with respect to $a$. In Lemma \ref{lem: directionsINPath}, we saw that, in every such embedding, $P$ moves in at most one vertical and at most one horizontal directions in $f$. In addition, assume that we construct this grid graph embedding $f$, by induction on $1\leq i\leq n$. Assume that we found an embedding for the first $i$ vertices of $P$, where $a<i<n$, and we wish to find an embedding for the next vertex, $v_{i+1}$. As the path $P$, in every such embedding $f$, ``moves'' most of the time except at most $a$ ``steps'' it ``moves forward'' and does not ``move backward''. Therefore, the vertices that might be ``close'' to $v_{i}$, that is, vertices with graph distance $1$ from $v_i$, in $f$, are only from the set $\{v_{i-a},\ldots,v_{i-1}\}$. Therefore, like in the problem of the previous subsection, as long as we keep track of the numbers of ``wrong directions'', we have so far in our embedding, we only need to remember the location of ``small'' amount of vertices in order to find an embedding for the next vertex. This is the basic intuition for the algorithm we present later on. As a tree is not as simple as a path, we continue to develop the idea of the algorithm. 

Now, we define the term {\em $(P,t)$-path}. A $(P,t)$-path is a tree that has a path $P$ that is its ``skeleton'', and all the other vertices are ``close'' to the skeleton (see Figure~\ref{fi:PtPath}). 

\begin{figure}[!ht]
\centering
\includegraphics[width=0.5\textwidth, page=8]{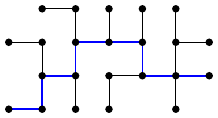}
\caption{Example of a $(P,2)$-path where the path $P$ is marked in blue.}\label{fi:PtPath}
\end{figure}

\begin{definition}[{\bf $(P,t)$-Path}]
Let $T$ be a tree, let $P\subseteq T$ be a path, and let $t\in\mathbb{N}_0$. Then, $T$ is {\em $(P,t)$-path}, if $\{v\in V(T)~|~$there exists $u\in V(P)$ such that $d(u,v)\leq t\}=V(T)$.  
\end{definition}

If $T$ is a $(P,a)$-path where $P=(v_1,\ldots,v_k)$, then we can construct a grid graph embedding $f$ of $T$ where $P$ moves at most one vertical and at most one horizontal directions, by induction on the number of vertices in $P$, in a similar way to our describtion in a previous paragraph. For a vertex $v\in V(T)$, we denote by $P_c(v)$ the closest vertex in $P$ to $v$. For a vertex $v\in V(P)$, we denote by $P(v)$ the set of vertices in $T$, that $v$ is the closest vertex to them in $P$, , that is, $P(v)=\{u\in V~|~P_c(u)=v\}$. In our construction of a grid graph embedding of $T$, at the $i$-th step, we find an embedding for the set of vertices $P(v_i)$. Assume that we are in the $(i+1)$-th step, where we have found an embedding for the set of vertices $P(v_1)\cup P(v_2)\cup \ldots \cup P(v_i)$, where $3a<i<k$. As we saw in the previous example, the vertices from $P$ to whom $v_{i+1}$ might be ``close'', are only from the set $\{v_{i-a},\ldots,v_{i}\}$. Therefore, as long as we keep track of the numbers of ``wrong directions'' of the path $P$, we only need to remember the location of the vertices of $P(v_{i-3a})\cup \ldots \cup P(v_{i})$ we have so far in our embedding, in order to find an embedding for the vertices of $P(v_{i+1})$. Observe that this number of vertices is polynomial in $a$.

In the next lemma, we prove this idea. More precisely, we prove that if $d_f(v,u)\leq a$, then $d(u,v)\leq 6a$.

\begin{lemma} \label{lem: gridDisAppGraphDis}
Let $T$ be a $(P,t)$-path that is a grid graph, with grid graph embedding $f$ of $T$, such that $P$ moves in two directions in $f$ except for at most $t$ edges. Then, for every $v,u\in V$ such that $d_f(v,u)\leq t$, it follows that $d(u,v)\leq 6t$.  
\end{lemma}

\begin{proof}
Let $u,v\in V(T)$ such that $d_f(v,u)\leq t$. Let $u'=P_c(u)$ and $v'=P_c(v)$. Notice that, because $T$ is a $(P,t)$-path, it follows that $d(u,u')\leq t$ and $d(u,u')\leq t$. Therefore, it follows that $d_f(u,u')\leq t$, and $d_f(v,v')\leq t$. Now, because $P$ moves in at most two directions in $f$, it follows that $d_f(v',u')\geq d(u',v')-t$. Therefore, we get that $d(v,u)\leq d(v,v')+d(v',u')+d(u',u)\leq t+ d_f(u',v')+t+t\leq d_f(v,v')+d_f(v,u)+d_f(u',u)+3t\leq d(v,v')+d_f(v,u)+d(u',u)+3t\leq t+t+t+3t=6t$. 
\end{proof} 

Obviously, for a given tree $T$ and a positive integer $a$, it is not sure that there exists a simple path $P$, such that $T$ is a $(P,a)$-path. In order to find such a $P$, we need $P$ to ``get inside'' every ``big'' subtree. Assume that $T$ has a vertex with three or four neighbors that are ``attached'' to some ``big'' subtrees of $T$. Since every simple path can ``get inside'' at most two of them, we cannot find such a $P$. We call these vertices {\em $a$-split vertices}. An {\em $a$-one-split vertex} is a vertex that is ``connected'' to exactly three ``big'' subtrees, that is, each subtree with at least $a$ vertices. An {\em $a$-double-split vertex} is a vertex that is ``connected'' to four ``big'' subtrees, that is, each subtree with at least $a$ vertices. An {\em $a$-split vertiex} is an $a$-one-split vertex or an $a$-double-split vertex. See Figure~\ref{fi:splitVertices}. We define these terms in the next definitions:     

\begin{figure}[!t]
\centering
\includegraphics[width=0.5\textwidth, page=7]{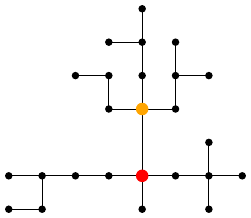}
\caption{Example of a $5$-one-split vertex (in red) and a $3$-double-split vertices (in orange).}\label{fi:splitVertices}
\end{figure}

\begin{definition}[{\bf One-Split Vertex}]\label{def:One Split Vertex}
Let $T$ be a tree, and let $k\in \mathbb{N}$. A vertex $v\in V$ is a {\em $k$-one-split vertex} if $T\setminus \{v\}$ contains at least three connected components, and exactly three of them contain at least $k$ vertices.
\end{definition}

\begin{definition}[{\bf Double-Split Vertex}]\label{def:DoubleSplit Vertex}
Let $T$ be a tree, and let $k\in \mathbb{N}$. A vertex $v\in V$ is a {\em $k$-double-split vertex} if $T\setminus \{v\}$ contains four connected components, each one with at least $k$ vertices.
\end{definition}

\begin{definition}[{\bf Split Vertex}]\label{def:Split Vertex}
Let $T$ be a tree, and let $k\in \mathbb{N}$. A vertex $v\in V$ is a {\em $k$-split vertex} if $v$ is a $k$-one-split vertex or a $k$-double-split vertex.
\end{definition}

Now, we continue with the discussion of the structure of an embedding $f$ of a tree with $a_f\leq a$. In the next lemma, we prove that, if $T$ is a grid graph with a grid graph embedding $f$ of $T$, then $T$ has at most two $a$-split vertices. An intuition for this statement, note that if we have too many $a$-split vertices, then we can find a path $P$ that moves in two horizontal directions or two vertical directions, a contradiction to Lemma \ref{lem: directionsINPath}. For a given path $P=(v_1,\ldots,v_\ell)$, we denote the ``opposite'' path by $P^*$, i.e. $P^*=(v_\ell,\ldots,v_1)$. Observe that, if $P$ moves in $f$ in $\Delta$ directions, then $P^*$ moves in $f$ in ``opposite'' directions.

First, we give a simple observation for a general grid graph. Suppose we have a connected grid graph $G$ with $|V(G)|\geq k^2$. Then, $G$ cannot be embedded into a ``small'' grid graph, that is, a grid graph smaller than a $k\times k$ grid graph. So, we have some vertices with a grid graph distance that is at least $k$, and therefore, we have a simple path with size at least $k$ in $G$.

\begin{observation} \label{numOfVerinGrid}
Let $G$ be a connected grid graph. For any $k\in \mathbb{N}$, if $|V(G)|\geq k^2$, then there exists a simple path $P$ in $G$ with size $k$.  
\end{observation}

Let $P_1=(v_1,\ldots, v_k)$ and $P_2=(v_k,\ldots,u_\ell)$ be two simple paths in a connected graph graph $G$. We denote by $(P_1,P_2)$ the path $(v_1,\ldots, v_k,\ldots,v_\ell)$.

\begin{lemma}\label{lem:splitVer}
Let $T$ be a tree that is a grid graph, with a grid graph embedding $f$ of $T$. Then, exactly one of the following holds:
\begin{enumerate}
\item $T$ has no $81a_f^2$-split vertices.
\item $T$ has exactly one $81a_f^2$-split vertex that is an $81a_f^2$-one-split vertex.
\item $T$ has exactly two $81a_f^2$-split vertices that are $81a_f^2$-one-split vertices.
\item $T$ has exactly one $81a_f^2$-split vertex that is an $81a_f^2$-double-split vertex.  
\end{enumerate}
\end{lemma}

\begin{proof}
We show that, if $T$ has one $81a_f^2$-double-split vertex, then $T$ has no other $81a_f^2$-split vertices. The other cases can be shown similarly. Assume that $v\in V(T)$ is an $81a_f^2$-double-split vertex. Let $u_1,u_2,u_3,u_4\in V(T)$ be the four neighbors of $v$. Now, $v$ is an $81a_f^2$-double-split vertex, so, for every $1\leq i\leq 4$, the connected component in $T\setminus \{v\}$ to which $u_i$ belongs, has at least $81a_f^2$ vertices. From Observation \ref{numOfVerinGrid}, we get that there exists a simple path, $P'_i=(v^i_1,v^i_2,\ldots,v^i_{m_i})$, in that connected component with at least $9a_f$ vertices. We show that there exists a path $P_i$, in that connected component where $|V(P_i)|>4a_f$, that starts from $u_i$. We show this for $i=1$. The proof for $i=2,3,4$ is identical. If $u_1\in V(P'_1)$, then, it is trivial. Otherwise, $u_1\notin V(P_1)$. Let $v^1_j$ be the closest vertex to $u_1$ in $P'_1$. Let $P''=(u_1,w_1\ldots,v^1_j)$ be the simple path from $u_1$ to $v^1_j$. Observe that $V(P'')\cap V(P'_1)=\{v^1_j\}$. If $j>4a_f$, then $P_1=(u_1,w_1\ldots,v^1_j,v^1_{j-1}\ldots,v^1_1)$ is a simple path with $|V(P_1)|>4a_f$ that starts from $u_1$. Otherwise, $P_1=(u_1,w_1\ldots,v^1_j,v^1_{j+1}\ldots,v^1_{m_i})$ is such a path. 

Now, assume that there exist $i\neq j\in[4]$ such that $P_i$ and $P_j$ move to the same direction in $f$; without loss of generality, they move to the right in $f$. Then, the path $(P^*_i,v,P_j)$ moves in two horizontal directions in $f$, and by Lemma \ref{lem: directionsINPath}, this is a contradiction. Therefore, each $P_i$ moves in one different direction. Assume that there exists another $81a_f^2$-split vertex $s\in V(T)$. Without loss of generallity, assume that $s$ is in the connected component of $u_1$, in $T\setminus \{v\}$. By using arguments similar to those in the previous paragraph, $s$ has at least three neighbors $t_1,t_2,t_3$, and there exist three simple paths, $Q_1,Q_2,Q_3$ in $T\setminus \{s\}$, such that $Q_i$ starts from $t_i$ and $|V(Q_i)|>4a_f$. Moreover, for every $i,j\in [3], i\neq j$, $Q_i$ and $Q_j$ do not move to the same direction in $f$. Let $B$ be the simple path from $s$ to $u_1$. Observe that for at least two of $Q_1,Q_2,Q_3$, it follows that $V(Q_i)\cap V(B)=\emptyset$. Assume, without loss of generality, that $V(Q_1)\cap V(B)=\emptyset$ and $V(Q_2)\cap V(B)=\emptyset$. So, we get that at least one of $Q_1,Q_2$ moves to the same direction as one of $P_2, P_3, P_4$. Assume, without loss of generality, that $P_2$ and $Q_2$ move to the same direction. Then, we get that, $(P^*_2,B,Q_2)$ is a simple path that moves in two vertical directions or two horizontal directions in $f$, a contradiction due to Lemma \ref{numOfVerinGrid}.                          
\end{proof}


In the next lemma, we prove that, if $T$ has no $81a_f^2$-split vertices, then there exists a path $P$ such that $T$ is a $(P,81a_f^2)$-path. In addition, we prove that such a path $P$ can be found in $\OO(n^2)$ time.

\begin{lemma} \label{lem:noSplitIsPTPath}
Let $T$ be a tree that is a grid graph, with grid graph embedding $f$ of $T$. If $T$ has no $81a_f^2$-split vertices, then there exists a path $P$ such that $T$ is a $(P,81a_f^2)$-path. Moreover, such a $P$ can be found in $\OO(n^2)$ time.
\end{lemma}      

\begin{proof}
First, we mark every vertex $v\in V(T)$ such that $T\setminus \{v\}$ has two connected component that have more than $81a_f^2$ vertices, each. Observe that this process takes $\OO(n)$ time for each vertex, so it takes $\OO(n^2)$ time overall.

Now, if we did not mark any vertex, then we look for a vertex $v$ such that the number of vertices of a connected component with the maximal number of vertices in $T\setminus \{v\}$, is minimal. In this case, notice that each of the connected components in $T\setminus \{v\}$ cannot have more than $81a_f^2$ vertices; otherwise, we could take the neighbor of $v$, $u$, in that connected component, and get less vertices in a connected component with the maximal number of vertices in $T\setminus \{u\}$. Therefore, we get that $T$ is a $(P,81a_f^2)$-path where $P=(v)$.

Otherwise, we marked at least one vertex. We claim that the subgraph of $T$ induced by the set of vertices we marked, denoted by $T_M$, is a path. First, we show that $T_M$ is connected. Let $u,v\in V(T_M)$. Let $P'=(u,v_1,\ldots,v_k,v)$ be the simple path from $u$ to $v$. Let $1\leq i\leq k$. Observe that in $T\setminus \{v_i\}$, there is a connected component that contains all the connected components of $T\setminus \{u\}$ except one of them, and there is another connected component that contains all the connected components of $T\setminus \{v\}$ except one of them. Thus, there are at least two connected components in $T\setminus \{v_i\}$ that have more than $81a_f^2$ vertices each, so $v_i\in V(T_M)$. Now, we show that each vertex in $V(T_M)$ has at most two neighbors in $T_M$. Assume, toward a contradiction, that there exists $v\in V(T_M)$ with three neighbors $u_1,u_2,u_3\in V(T_M)$. Then, observe that, for $1\leq i\leq 3$, the connected component to which $u_i$ belongs in $T\setminus \{v\}$ contains at least one of the two connected components in $T\setminus \{u_i\}$ that have at least $81a_f^2$ vertices, and so, it contains at least $81a_f^2$ vertices. Therefore, $T\setminus \{v\}$ has at least three connected components that have at least $81a_f^2$ vertices each, so $v$ is an $81a_f^2$-split vertex, a contradiction. Thus, we get that $T_M$ is a path.

Now, we show that, for every vertex $v$ in $T\setminus V(T_M)$, there exists a vertex $u\in V(T_M)$ such that $d(u,c)\leq 81a_f^2$. Let $v\in V(T)\setminus V(T_M)$. Let $u$ be the vertex in the same connected components as $v$ in $T\setminus V(T_M)$, that has a neighbor, $s$, in $T_M$. If $d(u,v)\geq 81a_f^2$, then observe that in the connected components to which $v$ belongs in $T\setminus \{u\}$, there are at least $81a_f^2$ vertices. In addition, the connected components to which $s$ belongs in $T\setminus \{u\}$, contains at least one of the two connected components that have at least $81a_f^2$ vertices in $T\setminus \{s\}$, so, also has a at least $81a_f^2$ vertices. Therefore, there exist two connected components in $T\setminus \{u\}$, that have at least $81a_f^2$ vertices, so $u\in V(T_M)$, a contradiction. So, $d(u,v)< 81a_f^2$, and we get that $d(s,v)\leq 81a_f^2$. 

In conclusion, we got that $P=T_M$ is a simple path such that for every $v\in V(T)$, there exists $u\in V(P)$ such that $d(u,v)\leq 81a_f^2$. Thus, $T$ is a $(P,81a_f^2)$-path.       
\end{proof}


In Lemma \ref{lem:splitVer}, we proved that any tree $T$, with grid graph embedding $f$ of $T$, has at most two $81a_f^2$-split vertices. Moreover, in Lemma \ref{lem:noSplitIsPTPath}, we proved that, if a tree $T$ has no $81a_f^2$-split vertices, then, $T$ is a $(P,81a_f^2)$-path, for a path $P$. Therefore, the structure of a tree $T$, with a grid graph embedding $f$ of $T$, is simple (see Figures \ref{fi:onesplitVertices}, \ref{fi:twosplitVertices} and \ref{fi:oneDsplitVertices}). The tree $T$ is composed of at most four $(P,81a_f^2)$-paths, each moves in different directions in $f$. Therefore, every such $(P,81a_f^2)$-path is almost ``independent'', since the vertices of a $(P,81a_f^2)$-path might be close to vertices of another $(P,81a_f^2)$-path, only around the split vertices. Therefore, the idea of our algorithm is as follows. First, we identify the split vertices, and find the $(P,81a_f^2)$-paths. For every split vertex, we guess a ``small'' grid graph embedding $f'$ ``around'' it. Then, for every $(P,81a_f^2)$-path we found, we look for a grid graph embedding that ``continues'' our guess $f'$. We define the term {\em subgrid} in the following definition.   


\begin{figure}[!t]
\centering
\includegraphics[width=0.5\textwidth, page=9]{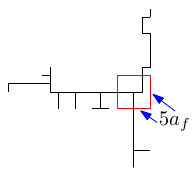}
\caption{An illustration for the construction of a grid graph embedding of a tree with one split-vertex.}\label{fi:onesplitVertices}
\end{figure}

\begin{figure}[!t]
\centering
\includegraphics[width=0.5\textwidth, page=10]{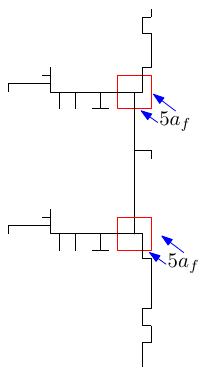}
\caption{An illustration for the construction of a grid graph embedding of a tree with two split-vertices.}\label{fi:twosplitVertices}
\end{figure}

\begin{figure}[!t]
\centering
\includegraphics[width=0.5\textwidth, page=11]{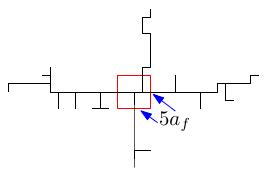}
\caption{An illustration for the construction of a grid graph embedding of a tree with one double-split vertex.}\label{fi:oneDsplitVertices}
\end{figure}

\begin{definition}[{\bf Subgrid}]\label{def:subgrid}
Let $G$ be a graph, let $U\subseteq V(G)$, let $U'\subseteq U$, let $k,r,k',r'\in \mathbb{N}$ such that $k'\leq k$ and $r'\leq r$, let $f$ be a $k\times r$ grid graph embedding of $G[U]$, and let $f'$ be a $k'\times r'$ grid graph embeddings of $G[U']$. We say that {\em $f'$ is a subgrid of $f$} if there exist $a,b\in \mathbb{Z}$, such that, the following conditions are satisfied:
\begin{enumerate}
\item For every $u\in U'$, $f'(u)=(\fr(u)+a,\fc(u)+b)$. \label{cond1:agee}
\item For every $1\leq i\leq k'$ and $1\leq j\leq r'$, $f^{-1}(i+a,j+b)=f'^{-1}(i,j)$. \label{cond2:agee}
\end{enumerate}   
\end{definition}

Observe that Condition \ref{cond1:agee} of Definition \ref{def:subgrid}, simply requests that $f'\subseteq f$ up to ``shifting'' $f'$ in parallel to the axes. Condition \ref{cond2:agee} requests that there are no vertices of $U\setminus U'$ that are embedded in the $k'\times r'$ grid graph embedding of $U'$ in $f$. When we say that $f'$ is a subgrid of $f'$, we can look at this definition as if there exists a $k'\times r'$ grid graph ``inside $f$'' that ``looks exactly'' like $f'$. 

For a $k\times r$ grid graph embedding $f$ of a graph $G$, we have $\mathsf{length}(f)=\max\{\fr(u)-\fr(v)+1~|~u,v\in V(G)\}$ and $\mathsf{width}(f)=\max\{\fc(u)-\fc(v)+1~|~u,v\in V(G)\}$. Observe that $\mathsf{length}(f)\leq k$ and $\mathsf{width}(f)\leq r$. In addition, observe that $G$ is a $k\times r$ grid graph, if and only if, there exists a $k'\times r'$ grid graph $f$ of $G$ such that $\mathsf{length}(f)\leq k$ and $\mathsf{width}(f)\leq r$  

We proceed with the description of our algorithm. Assume that we want to construct, if possible, a $k\times r$ grid graph of $T$ with $a_f\leq a$, for $k,r,a\in \mathbb{N}$. Now, we focus on finding a $k\times r$ grid graph embedding $f$ of a $(P,81a^2)$-path such that $P$ moves in two directions in $f$, except for at most $a$ edges. In addition, we want $f$ to have $f'$ as subgrid where $f'$ is a $(366a^2+1)\times (366a^2+1)$ grid graph embedding of a ``small environment'', $U\subseteq V(T)$, of a split vertex $u$ where $f'(u)=(183a^2+1,183a^2+1)$, i.e., $u$ is located ``in the middle'' of $f'$. As we have explained, $f'$ is the ``problematic area'', where the grid graph embeddings of some $(P,81a^2)$-paths might have ``conflicts''. Therefore, we look for a grid graph embedding, of each $(P,81a^2)$-path separately, while the ``problematic area'', that is, $f'$, is fixed. 

The algorithm we suggest is similar to Algorithm \ref{alg:Grid Recognition a_f} (given in the previous subsection). Here, we only state the changes in the algorithm for the purposes of this part. First, we remind the idea of Algorithm \ref{alg:Grid Recognition a_f}, where we aim to find a $k\times r$ grid graph embedding of a graph $G$ in \FPT\ time with respect to the parameter $a_G+k$. We guessed an initial vertex, to embedded in the first column of the grid, and then we had some ``order'' on the ``appearance'' of the vertices in the grid. Then, by iterations using Algorithm \ref{alg:Grid Recognition Iteration}, we tried to construct the next ``small piece'' of the grid that ``fits'' to one of the previous ``small piece'' possibilities.

Now, in our case, where we want to find a $k\times r$ grid graph embedding $f$ of a $(P,81a^2)$-path such that $P$ moves in two directions in $f$ except for at most $4a$ edges, we can use a similar approach. Assume that $P=(v_1,\ldots,v_\ell)$. Here, in the $i$-th iteration, we want to find an embedding for $P(v_i)$, where $P(v_i)=\{v\in V((P,81a^2))~|~P_c(v)=v_i\}$, and $P_c(v)$ is the closest vertex to $v$ in $P$. Due to Lemma \ref{lem: directionsINPath}, where we show that if $P$ moves in two directions in $f$ except for at most $4a$ edges, then, for every two vertices, $u$ and $v$, such that $d_f(u,v)\leq a$, we get that $d(u,v)\leq 6a$, it is enough to remember the embedding of the vertices of $P(v_{i-6a})\cup \ldots \cup P(v_{i-1})$. Now, observe that, by arguing similar arguments as we did for Observation \ref{numOfVerinGrid}, the number of vertices, in each $P(v_{j})$, is bounded by $(81a^2)^2=6561a^4$. Therefore, we get that $|P(v_{i-6a})\cup \ldots \cup P(v_{i-1})|\leq 6561a^4\cdot 6a=39366a^5= \OO(a^5)$. Denote $39366a^5$ by $q$. So, at each step, we have from the previous step, the $(i-1)$-th step, some $2q+1\times 2q+1$ grid graphs of the graph $T$, induced by the set of vertices $P(v_{i-6a})\cup \ldots \cup P(v_{i-1})$, where $f(v_{i-1})=(q,q)$. For every such embedding, we also have the number of ``wrong directions'', that is, the number of edges that were not directed up or right. In addition, we have $\mathsf{MaxColumn}$, $\mathsf{MinColumn}$, $\mathsf{MaxRow}$ and $\mathsf{MinRow}$, which mark the maximal column, the minimal column, the maximal row and the minimal row, respectively, that we have embedded any vertex. Moreover, we have $\mathsf{startingLocation}$, which marks the ``real'' location in the ``entire'' grid graph embedding we have so far of the point $(1,1)$. At the first iteration, $\mathsf{startingLocation}=(1,1)$, and, using brute force, we construct every possible $2q\times 2q$ grid graph of $P(v_{1})\cup \ldots \cup P(v_{6a})$ such that $P$ moves up and right, except for at most $4a$ edges. We count the number of ``wrong directions''. For each such an embedding we found, we also update $\mathsf{MaxColumn}$, $\mathsf{MinColumn}$, $\mathsf{MaxRow}$ and $\mathsf{MinRow}$. At the $i$-th step, we take each embedding, $f$, from the last step, and we use brute force, in order to find all the possibilities to ``add'' the embedding of $P(v_i)$ to $f$, as long as the number of ``wrong directions'' is still bounded by $4a$. Also we update $\mathsf{MaxColumn}$, $\mathsf{MinColumn}$, $\mathsf{MaxRow}$ and $\mathsf{MinRow}$, according to the vertices we added to the embedding, considering the ``real'' location of the embedding, according to $\mathsf{startingLocation}$. We make sure that $\mathsf{MaxColumn}-\mathsf{MinColumn}+1\leq r$ and $\mathsf{MaxRow}-\mathsf{MinRow}+1\leq k$. Then, we delete the embeddings of $P(v_{i-6a})$, as they are ``irrelevant'' for the next step, as we have explained, and ``shift'' the grid graph embedding, so we get that $f(v_{i})=(q,q)$. We update $\mathsf{startingLocation}$ accordingly. If there exists an iteration where we cannot find any such embedding, we can conclude that we have a no-instance. Observe that the number of possible $2q\times 2q$ grid graph embeddings of the graph $T$ induced by the set of vertices $P(v_{i-6a+1})\cup \ldots \cup P(v_{i})$ is bounded by $(4q^2)!=\OO((4q^2)^{4q^2})=\OO((a^{10})^{\OO(a^{10})})$. In addition, the number of possibilities for $\mathsf{MaxColumn}$, $\mathsf{MinColumn}$, $\mathsf{MaxRow}$ and $\mathsf{MinRow}$ and $\mathsf{startingLocation}$ is bounded by $r\cdot r\cdot k\cdot k\cdot k\cdot r= k^3\cdot r^3=\OO(|V(G)|^6)$. Therefore, the number of grid graph embeddings we have at the end of each iteration, is bounded by $\OO(|V(G)|^6(a^{10})^{\OO(a^{10})})$. At each iteration, with a similar counting arguments, we have at most $\OO(|V(G)|^6(a^{10})^{\OO(a^{10})})$ ``small embeddings'', and we need to check each of them, if there exists some ``small embeddings'' from the last iteration that ``fit''. Therefore, each iteration takes at most $\OO(|V(G)|^{12}(a^{10})^{\OO(a^{10})})$. Observe that the number of iteration is bounded by $|V(G)|$, so the total running time of the algorithm is bounded by $\OO(|V(G)|^{13}(a^{10})^{\OO(a^{10})})$. Now, in order to make sure that $f'$ is a subgraph of the grid graph embedding, $f$, that we construct, we can start with the first iteration, with $f'$, with the constraint that we cannot embed any vertices in the ``area'' of $f'$. 

Therefore, we have the next lemma:

 \begin{lemma}\label{algoExPath}
There exists an algorithm that gets as input, a tree $T$, $a,k,r\in \mathbb{N}$, a $(P,t)$-path in $T$, $U\subseteq V(T)$, a $(366a^2+1)\times (366a^2+1)$ grid graph $f'$ of $U$ and at most two directions $\Delta \subset \mathsf{dir}$, at most one horizontal direction, and at most one vertical direction. The algorithm returns ``yes'' if and only if there exists a $k\times r$ grid graph embedding $f'$ of $T[V(P,t)\cup U]$ such that $P$ moves in $\Delta$ directions in $f'$, except for at most $4a$ edges, and $f'$ agrees with $f$. In addition, the algorithm works in time $\OO(|V(G)|^{13}(a^{10})^{\OO(a^{10})})$.   
\end{lemma}       

In Lemma \ref{algoExPath}, we show the existence of an algorithm, that given a $(P,t)$-path, and a ``small grid environment'' around a split vertex, finds a $k'\times r'$ grid graph embeddings, of the $(P,t)$-paths, with the given ``environment'' as a subgrid, if such an embedding exists, in \FPT\ time. Now, we show how to use such an algorithm in order to determine if there exists a $k\times r$ grid graph embedding of $T$. 

Assume that we are given a tree $T$, positive integers $k,r,a\in \mathbb{N}$, and we want to determine if there exists a $k\times r$ grid graph embedding of $T$, $f$, such that $a_f\leq a$. In Lemma \ref{lem:splitVer}, we proved that, if such a grid graph embedding exists, then we have few cases to consider. $T$ has no $81a^2$-split vertices, or it has exactly one split vertex, that is an $81a^2$-one-split vertex, or it has exactly two split vertices, that are some $81a^2$-one-split vertices, or it has exactly one split vertex, that is an $81a^2$-double-split vertex. For the case where $T$ has no $a$-split vertices, we can simply use an algorithm that corresponds to Lemma \ref{algoExPath}. Now, we will prove the existence of an algorithm that solves the case, where has one $81a^2$-double-split vertex, in \FPT\ time. The other cases can be solved similarly, and we will state the differences. 

In this case, we get that, by removing the only split vertex, we have four connected components, each is a $(P,t)$-path, where $t=81a^2$. We show that each such $(P,t)$-path ``moves'' to a different direction in $f$. Then, as we explained and as we will prove, each embedding is ``almost independent'', and there might be ``overlapping'' only ``around'' the split vertex. Therefore, we guess a ``small environment'' around the split vertex, we guess the directions that the $(P,t)$-path moves, and we look for grid graph embeddings of the $(P,t)$-paths, such that, the ``small environment'' is a subgrid of them. We show that there exists a grid graph embedding of $T$, if and only if, we managed to find such embeddings that correspond to a possible guess. We start with the easy direction, that is, if there exists a grid graph embedding $f$ of $T$, then there exist a ``small environment'' around the split vertex, a different ``main direction'' for every $(P,t)$-path, and grid graph embeddings of the $(P,t)$-paths, such that the ``small environment'' is a subgrid of them. We simply can take a ``small environment'' around the split vertex in $f$, and consider the ``parts'' of $f$ where we have the embedding of each $(P,t)$-path. We also consider the dimensions of the grid graph embedding. Since every $(P,t)$-path ``moves'' to a different direction, we have two $(P,t)$-paths that move to vertical directions, and two $(P,t)$-paths that move to horizontal directions. In a sense, as we will prove later, the $(P,t)$-paths that move to vertical directions, determine the number of rows in the embedding, and the $(P,t)$-paths that move to horizontal directions determine the number of columns in the embedding. We prove this direction in the next lemma below. We first present several notations:
 
\begin{itemize}
\item For a given graph $G$, a vertex $v\in V$, and a positive integer $t\in \mathbb{N}$, we denote by $r(v,t)$, the set of vertices, which in a graph distance from $v$, less or equal to $t$, that is, $r(v,t)=\{u\in V~|~d(u,v)\leq t\}$.
\item For a given graph $G$, a grid graph embedding $f$ of $G$, a vertex $v\in V$, and a positive integer $t\in \mathbb{N}$, we denote by $(f,u,t)$, the $t\times t$ subgrid of $G$, where $(f,u,t)(u)=(\lceil \frac{t}{2} \rceil,\lceil \frac{t}{2} \rceil)$.                   
\item For a given graph $G$, a subset of vertices $U\subseteq V$, and a grid graph embedding $f$ of $G[U]$, we denote $U$ by $V(f)$.
\item For a given graph $G$, a subset of vertices $U\subseteq V$, and a grid graph embedding $f$ of $G$, we denote by $f[U]$, the grid graph embedding of $G[U]$, defined as, $f[U](u)=f(u)$, for every $u\in U$. 
\end{itemize}

\begin{lemma} \label{lem:algo1Dir}
Let $T=(V,E)$ be a tree, let $k,r,a\in \mathbb{N}$, let $f$ be a $k\times r$ grid graph embedding of $T$. Assume that $a_f\leq a$, and that $T$ has exactly one $81a^2$-split vertex $u\in V$, that is an $81a^2$-double-split vertex. Let $U_1,U_2,U_3$ and $U_4$ be subsets of $V$ such that $\bigcup_{1\leq i\leq 4}U_i=V\setminus \{u\}$, $U_i$ is a set of vertices of a connected component in $T\setminus \{u\}$ such that $|U_i|\geq 81a^2$, for $i\in [4]$. Then, there exist $i_1,i_2,i_3,i_4\in [4]$, such that $i_j\neq i_\ell$, for $1\leq j< \ell \leq 4$, $r_1,r_2,r_3,r_4,k_1,k_2,k_3,k_4\in \mathbb{N}$, $U\subseteq r(u,367a^2+1)$ such that $u\in U$, and a $(366a^2+1)\times (366a^2+1)$ grid graph embedding $f'$ of $T[U]$ such that $f'(u)=(183a^2+1,183a^2+1)$, and the following conditions are satisfied: 
\begin{enumerate} 
\item For every $1\leq i\leq 3$, $T[U_i]$ is a $(P,81a^2)$-path. \label{lem:algo1DirCon1}
\item There exists a $k_1\times r_1$ grid graph embedding, $f^1$, of $T[U_{i_1}\cup U]$, such that, for every $P=(v_0,\ldots,v_\ell)$ such that $T[U_{i_1}]$ is a $(P,81a^2)$-path, and $P_c(u)=v_0$, $P$ moves up in $f^1$, except for at most $4a$ edges, and $f'$ is a subgrid of $f^1$. \label{lem:algo1DirCon2}
\item There exists a $k_2\times r_2$ grid graph embedding, $f^2$, of $T[U_{i_2}\cup U]$, such that, for every $P=(v_0,\ldots,v_\ell)$ such that $T[U_{i_2}]$ is a $(P,81a^2)$-path, and $P_c(u)=v_0$, $P$ moves left in $f^2$, except for at most $4a$ edges, and $f'$ is a subgrid of $f^2$. \label{lem:algo1DirCon3}
\item There exists a $k_3\times r_3$ grid graph embedding, $f^3$, of $T[U_{i_3}\cup U]$, such that, for every $P=(v_0,\ldots,v_\ell)$ such that $T[U_{i_3}]$ is a $(P,81a^2)$-path, and $P_c(u)=v_0$, $P$ moves down in $f^3$, except for at most $4a$ edges, and $f'$ is a subgrid of $f^3$. \label{lem:algo1DirCon4}
\item There exists a $k_4\times r_4$ grid graph embedding, $f^4$, of $T[U_{i_4}\cup U]$, such that, for every $P=(v_0,\ldots,v_\ell)$ such that $T[U_{i_4}]$ is a $(P,81a^2)$-path, and $P_c(u)=v_0$, $P$ moves right in $f^4$, except for at most $4a$ edges, and $f'$ is a subgrid of $f^4$. \label{lem:algo1DirCon45}
\item $k_1+k_3-\mathsf{length}(f')\leq k$ and $r_2+r_4-\mathsf{width}(f')\leq r$. \label{lem:algo1DirCon5} 
\end{enumerate}  
\end{lemma}

\begin{proof}
Since $T$ has exactly one $81a^2$-split vertex $u\in V$, we get that $T[U_i]$ has no split-vertices, and therefore, by Lemma \ref{lem:noSplitIsPTPath}, $T[U_i]$ is a $(P,81a^2)$-path, for every $1\leq i\leq 4$. Thus, Condition \ref{lem:algo1DirCon1} of the lemma is satisfied. Now, it can be proved, similarly to the proof of Lemma \ref{lem:splitVer}, that for every $1\leq i\leq 4$, there exists a direction $d_i\in \mathsf{dir}$, such that, $d_i\neq d_j$, for every $1\leq i<j\leq 4$, and the following condition is satisfied. For all simple paths, $P^1=(v_0^1,\ldots,v_{t_1}^1)$, $P^2=(v_0^2,\ldots,v_{t_2}^2)$, $P^3=(v_0^3,\ldots,v_{t_3}^3)$ and $P^4=(v_0^4,\ldots,v_{t_4}^4)$, such that $T[U_i]$ is a $(P^i,81a^2)$-path, and $P_c^i(u)=v^i_0$, we get that $P^i$ moves to direction $d_i$, except for at most $4a$ edges, in $f^i$, for every $1\leq i\leq 4$. Let $i_1,i_2,i_3,i_4\in \{1,2,3,4\}$ such that $d_{i_1}=\{\mathsf{up}\}$, $d_{i_2}=\{\mathsf{left}\}$, $d_{i_3}=\{\mathsf{down}\}$ and $d_{i_4}=\{\mathsf{right}\}$. For the sake of simplicity, we assume, without loss of generality, that $i_1=1, i_2=2$, $i_3=3$ and $i_4=4$. 

Now, let $f'=(f,u,366a^2+1)$. Observe that $f'(u)=(183a^2+1,183a^2+1)$, and also, for every $v\in V(f')$, $d_f(v,u)\leq 81a^2$. Therefore, since $a_f\leq a$, for every $v\in V(f')$, we get that $d(v,u)\leq 81a^2+a<82a^2$. So, we have that $U=V(f')\subseteq r(u,82a^2)$. Now, for every $1\leq i\leq 4$, we set $f^i=f[U_{i}\cup U]$. In addition, for every $1\leq i\leq 4$, we set $k_{i}=\mathsf{length}(f^i)$ and $r_{i}=\mathsf{width}(f^i)$. Observe that this completes the proof of Conditions \ref{lem:algo1DirCon2}, \ref{lem:algo1DirCon3}, \ref{lem:algo1DirCon4} and \ref{lem:algo1DirCon45} of the lemma.

Now, we show that Condition \ref{lem:algo1DirCon5} is satisfied. We show that $\mathsf{argmax}\{\fr(v)~|~v\in V\}\subseteq U_3$. Let $P^3=(v_0^3,\ldots,v_{t_3}^3)$ such that $T[U_3]$ is a $(P^3,81a^2)$-path, and $P_c^3(u)=v^3_0$. Observe that $d(u,v^3_0)\leq 81a^2+1$, so $d_f(u,v^3_0)\leq 81a^2+1$, and we get that $\fr(v^3_0)\geq \fr(u)-81a^2-1$. Since $P^3$ moves down in $f$, there exist at most $a4$ edges that are not directed down. We assume that $t_3>164a^2$ (otherwise, observe that $|U_3|=\OO(a^{\OO(1)})$ and we can use brute force). Therefore, we get that $\fr(v_{t_3}^3)\geq \fr(v^3_0)+164a^2-4a\geq \fr(u)-81a^2-1+164a^2-4a>\fr(u)+81a^2$. Now, let $v\in V\setminus U_3$; we assume that $v\in U_2$ (the other cases are similar). Let $P^2=(v_0^2,\ldots,v_{t_2}^2)$ such that $T[U_2]$ is a $(P^2,81a^2)$-path, and $P_c^2(u)=v^2_0$. Observe that, $d(u,v^2_0)\leq 81a^2+1$, so $d_f(u,v^2_0)\leq 81a^2+1$, and we get that $\fr(v^2_0)\leq \fr(u)+81a^2+1$. Since $P^2$ moves left in $f$, then, there exist at most $a$ edges that directed down. Therefore, for every $0\leq i\leq t_2$, $\fr(v_i^2)\leq \fr(v^2_0)+a\leq \fr(u)+82a^2$. Since $T[U_2]$ is a $(P^2,81a^2)$-path, there exists $i$ such that $d(v,v_i^2)\leq 81a^2$. Therefore, $d_f(v,v_i^2)\leq 81a^2$, so $\fr(v)\leq \fr(v_i^2)+81a^2\leq \fr(u)+82a^2+81a^2=\fr(u)+163a^2<164a^2<\fr(v_{t_3}^3)$. Thus, we proved that $\mathsf{argmax}\{\fr(v)~|~v\in V\}\subseteq U_3$. Similarly, it can be proved that $\mathsf{argmin}\{\fr(v)~|~v\in V\}\subseteq U_1$, $\mathsf{argmax}\{\fc(v)~|~v\in V\}\subseteq U_4$ and $\mathsf{argmin}\{\fr(v)~|~v\in V\}\subseteq U_2$.

Now, we show that $\mathsf{argmin}\{f^3_\mathsf{row}(v)~|~v\in V(f^3)\}\subseteq U$. Notice that $V(f^3)=U\cup U_3$, and since $f^3=f[U\cup U_3]$, it is enough to prove that $\mathsf{argmin}\{f_\mathsf{row}(v)~|~v\in U\cup U_3\}\subseteq U$. Let $P^3=(v_0^3,\ldots,v_{t_3}^3)$ such that $T[U_3]$ is a $(P^3,81a^2)$-path, and $P_c^3(u)=v^3_0$. Observe that $d(u,v^3_0)\leq 81a^2+1$, so $d_f(u,v^3_0)\leq 81a^2+1$, and we get that $\fr(v^3_0)\geq \fr(u)-81a^2-1$. Since $P^3$ moves down in $f$, there are at most $a$ edges in $P^3$ that are directed up. Therefore, for every $0\leq i\leq t_3$, $\fr(v^3_i)\geq fr(v^3_0)-a\geq \fr(u)-81a^2-1-a>\fr(u)-82a^2$. Let $v\in U_3$. Since $T[U_3]$ is a $(P^3,81a^2)$-path, there exists $i$ such that $d(v,v_i^3)\leq 81a^2$. So, $d_f(v,v_i^3)\leq 81a^2$, and we get that $\fr(v)\geq \fr(v_i^3)-81a^2>\fr(u)-82a^2-81a^2=\fr(u)-163a^2$. Now, we show that there exists $u_1\in U_1$ such that $\fr(u_1<\fr(u)-163a^2)$, and $d(u_1,u)<325a^2$. Let $P^1=(v_0^1,\ldots,v_{t_1}^1)$ such that $T[U_1]$ is a $(P^1,81a^2)$-path, and $P_c^1(u)=v^1_0$. Observe that $d(u,v^1_0)\leq 81a^2+1$, so $d_f(u,v^1_0)\leq 81a^2+1$, and we get that $\fr(v^1_0)\leq \fr(u)+81a^2-1$. Since $P^1$ moves up in $f$, except for at most $a$ edges, $\fr(v^1_{244a^2})\leq \fr(v^1_0)-(244a^2-a)+a\leq \fr(u)+81a^2-1-(244a^2-a)+a=\fr(u)-163a^2-2a-1<\fr(u)-163a^2<\fr(u_3)$. We got that for $u_1=v^1_0$, $\fr(u_1)<\fr(u)-163a^2$, and $d(u_1,u)<325a^2$. Moreover, $d(u_1,u)<325a^2$ implies that $d_f(u_1,u)<325a^2$, so $u_1\in U$. Thus, we proved that $\mathsf{argmin}\{f^3_\mathsf{row}(v)~|~v\in V(f^3)\}\subseteq U$. Similarly, it can be proved that $\mathsf{argmax}\{f^1_\mathsf{row}(v)~|~v\in V(f^1)\}\subseteq U$ and $\mathsf{argmin}\{f^1_\mathsf{col}(v)~|~v\in V(f^1)\}\subseteq U$ and $\mathsf{argmin}\{f^4_\mathsf{col}(v)~|~v\in V(f^4)\}\subseteq U$.

Now, let $v^3_{\mathsf{Rmax}}\in \mathsf{argmax}\{f^3_\mathsf{row}(v)~|~v\in V(f^3)\}$, so, from what we have proved, $v^3_{\mathsf{Rmax}}\in \mathsf{argmax}\{f_\mathsf{row}(v)~|~v\in V\}$. Similarly, let $v^1_{\mathsf{Rmin}}\in \mathsf{argmin}\{f^1_\mathsf{row}(v)~|~v\in V(f^1)\}$, so, from what we have proved, $v^1_{\mathsf{Rmin}}\in \mathsf{argmin}\{f_\mathsf{row}(v)~|~v\in V\}$. Therefore, observe that, $\mathsf{length}(f)=\fr(v^3_{\mathsf{Rmax}})-\fr(v^1_{\mathsf{Rmin}})+1$. Now, let $v^3_{\mathsf{Rmin}}\in \mathsf{argmin}\{f^3_\mathsf{row}(v)~|~v\in V(f^3)\}$, so, from what we have proved, $v^3_{\mathsf{Rmin}}\in \mathsf{argmin}\{f'_\mathsf{row}(v)~|~v\in U\}$. Similarly, let $v^1_{\mathsf{Rmax}}\in \mathsf{argmax}\{f^1_\mathsf{row}(v)~|~v\in V(f^1)\}$, so, from what we have proved, $v^1_{\mathsf{Rmax}}\in \mathsf{argmax}\{f'_\mathsf{row}(v)~|~v\in U\}$. Therefore, we get that $\mathsf{length}(f')=\fr(v^1_{\mathsf{Rmax}})-\fr(v^3_{\mathsf{Rmin}})+1$, $\mathsf{length}(f^1)=\fr(v^1_{\mathsf{Rmax}})-\fr(v^1_{\mathsf{Rmin}})+1$, and $\mathsf{length}(f^3)=\fr(v^3_{\mathsf{Rmax}})-\fr(v^3_{\mathsf{Rmin}})+1$. Thus, we got that $\mathsf{length}(f)=\fr(v^3_{\mathsf{Rmax}})-\fr(v^1_{\mathsf{Rmin}})+1=\fr(v^3_{\mathsf{Rmax}})-\fr(v^1_{\mathsf{Rmin}})+1+\fr(v^3_{\mathsf{Rmin}})-\fr(v^3_{\mathsf{Rmin}})+\fr(v^1_{\mathsf{Rmax}})-\fr(v^1_{\mathsf{Rmax}})+1-1=\fr(v^3_{\mathsf{Rmax}})-\fr(v^3_{\mathsf{Rmin}})+1+\fr(v^1_{\mathsf{Rmax}})-\fr(v^1_{\mathsf{Rmin}})+1-(\fr(v^3_{\mathsf{Rmax}})-\fr(v^1_{\mathsf{Rmin}})+1)=\mathsf{length}(f^3)+\mathsf{length}(f^1)-\mathsf{length}(f')=k_3+k_1-\mathsf{length}(f')$. In conclusion, we got that $\mathsf{length}(f)=k_3+k_1-\mathsf{length}(f')$, and therefore, $k_1+k_3-\mathsf{length}(f')\leq k$. Similarly, it can be proved, that $r_2+r_4-\mathsf{width}(f')\leq r$, so Condition \ref{lem:algo1DirCon5} is satisfied.            
\end{proof}

Now, we turn to prove the opposite direction of Lemma \ref{lem:algo1Dir}. We are given four grid graphs embeddings $f^1$, $f^2$, $f^3$ and $f^4$, of $T[U_1], T[U_2], T[U_3]$ and $T[U_4]$, correspondingly, such that each $T[U_i]$ is a $(P,81a^2)$-path that ``moves'' to a different direction in $f^i$, for $i\in[4]$. In addition, we have ``small environment'', around the split vertex $u$, which is a grid graph embedding $f'$ of $T[U]$, for a subset $u\in U\subseteq V$, such that, $f'$ is subgrid of $f^i$. We show, that there exists a grid graph embedding of $T$. For this, we ``glue'' $f^1$,$f^2$,$f^3$ and $f^4$ together on the subgrid $f'$ to get a grid graph embedding $f$ of $T$. Since each $f^i$ ``moves'' to another direction, we will prove that there is no ``overlap'' in $f$, that is, there are no two vertices that are embedded to the same point, and therefore $f$ is a grid graph embedding. 

First, we want to define a relation between two grid graph embeddings, $f^1$ and $f^2$, such that it is possible that $V(f^1)\cap V(f^2)\neq \emptyset$. We say, that $f_2$ {\em agrees} with $f_1$, if we can ``stick'' $f_2$ on $f_1$, and get no ``overlaps''. We define this term in the next definition:     

\begin{definition}\label{def:agree}
Let $G$ be a graph, let $U_1,U_2\subseteq V(G)$, and let $f^1,f^2$ be grid graph embeddings of $T[U_1]$ and $T[U_2]$, correspondingly. Then, $f^2$ {\em agrees} with $f^1$ if there exist $a,b\in \mathbb{Z}$ such that the following conditions are satisfied:
\begin{enumerate}
\item For every $u\in U_1\cap U_2$, $f^1(u)=(\fr^2(u)+a,\fc^2(u)+b)$. \label{cond1:agree}
\item For every $u_1\in U_1\setminus U_2$, $(f^2)^{-1}(\fr^1(u_1)-a,\fc^1(u_1)-b)=\emptyset$. \label{cond2:agree}
\item For every $u_2\in U_2\setminus U_1$, $(f^1)^{-1}(\fr^2(u_1)+a,\fc^1(u_1)+b)=\emptyset$. \label{cond3:agree}
\end{enumerate}   
\end{definition}

Assume that we have $f^1$ and $f^2$, that are grid graph embeddings of $T[U_1]$ and $T[U_2]$, correspondingly, such that $f^2$ agrees with $f^1$, and $U_1\cap U_2\neq \emptyset$. Then, observe that the integers $a,b\in \mathbb{Z}$ that satisfy the conditions of Definition \ref{def:agree} are unique. Notice that this means that there is only one way to ``glue'' $f^2$ with $f^1$ such that the vertices of $U_1$ do not ``overlap'' any vertex of $U_2$, and vice versa.

Now, assume that we have $U_1,U_2,U_3,U_4,U\subseteq V(G)$, such that, $|U|\geq 1$, $(U_i\cup U)\cap (U_j\cup U)=U$, for every $1\leq i<j\leq 4$, and $f^1,f^2,f^3$ and $f^4$, grid graph embeddings, of $T[U_1\cup U],T[U_2\cup U],T[U_3\cup U]$ and $T[U_4\cup U]$, correspondingly. Moreover, assume that, for every $\{u,v\} \in E(G)$, there exists $i\in \{1,2,3,4\}$, such that, $u,v\in U_i\cup U$. In addition, assume that $f^j$ agrees with $f^i$, for every $1\leq i<j\leq 4$. Then, since there is only one way to ``stick'' $f^j$ on $f^i$, if we ``stick'' them all together, we get that none of the grid graph embeddings ``overlaps'' another. Therefore, this is define a ``natural'' way to define a grid graph embedding of $T[U_1\cup U_2\cup U_3\cup U_4\cup U]$, as we define in the next observation.

\begin{observation}\label{obs:stikgrid}
Let $G$ be a graph, $U_1,U_2,U_3,U_4,U\subseteq V(G)$, such that, $|U|\geq 1$, $V(G)=U\cup U_1\cup U_2\cup U_3\cup U_4$, $(U_i\cup U)\cap (U_j\cup U)=U$, for every $1\leq i<j\leq 4$, and $f^1,f^2,f^3$ and $f^4$, grid graph embeddings, of $G[U_1\cup U],G[U_2\cup U],G[U_3\cup U]$ and $G[U_4\cup U]$, correspondingly. Moreover, assume that, for every $\{u,v\} \in E(G)$, there exists $i\in \{1,2,3,4\}$, such that, $u,v\in U_i\cup U$. In addition, assume that $f^j$ agrees with $f^i$, for every $1\leq i<j\leq 4$. Let $a,b,a',b',a'',b''\in \mathbb{Z}$ be the unique integers, such that $a,b$ are the guaranteed integers, from Definition \ref{def:agree}, correspond to $f^1$ and $f^2$, $a',b'$ are the guaranteed integers, from Definition \ref{def:agree}, corresponding to $f^1$ and $f^3$, and $a'',b''$ are the guaranteed integers, correspond to $f^1$ and $f^4$. Let $f$, denoted also by $f^1*f^2*f^3*f^4$, be the function $f:V(G)\rightarrow \mathbb{Z}\times \mathbb{Z}$, defined as follows. For every $u_1\in U_1\cup U$, $f(u_1)=f^1(u_1)$, for every $u_2\in U_2$, $f(u_2)=(\fr^2(u_2)+a,\fc^2(u_2)+b)$, for every $u_3\in U_3$, $f(u_3)=(\fr^3(u_3)+a',\fc^3(u_3)+b')$, and for every $u_4\in U_4$, $f(u_4)=(\fr^4(u_4)+a'',\fc^4(u_2)+b'')$. Then, $f$ is a grid graph embedding of $G$.       

\end{observation}

Now, we are ready to prove the opposite direction of Lemma \ref{lem:algo1Dir}. We show that $f^j$ agrees with $f^i$, for every $1\leq i<j\leq 4$. Then, we evoke Observation \ref{obs:stikgrid}, in order to prove that there exists a grid graph embedding of $T$. 

\begin{lemma} \label{lem:algo1Dir2}
Let $T=(V,E)$ be a tree and let $a\in \mathbb{N}$. Assume that $T$ has exactly one $81a^2$-split vertex $u\in V$ which is an $81a^2$-double-split vertex. Let $U_1,U_2,U_3$ and $U_4$ be subsets of $V$ such that $\bigcup_{1\leq i\leq 4}U_i=V\setminus \{u\}$, $U_i$ is a set of vertices of a connected component in $T\setminus \{u\}$and $|U_i|\geq 81a^2$, for $i\in [4]$. Let $U\subseteq r(u,367a^2+1)$ such that $u\in U$, and let $f'$ be a $(366a^2+1)\times (366a^2+1)$ grid graph embedding of $T[U]$ such that $f'(u)=(183a^2+1,183a^2+1)$. Assume that there exist $i_1,i_2,i_3,i_4\in [4]$, such that $i_j\neq i_\ell$ for $1\leq j< \ell \leq 4 $, $r_1,r_2,r_3,r_4,k_1,k_2,k_3,k_4\in \mathbb{N}$ and the following conditions are satisfied: 
\begin{enumerate} 
\item There exists a path $P=(v_0,\ldots,v_\ell)$ such that $T[U_{i_1}]$ is a $(P,81a^2)$-path, and $P_c(u)=v_0$, and a $k_1\times r_1$ grid graph embedding $f^1$ of $T[U_{i_1}\cup U]$ such that $P$ moves up in $f^1$, except for at most $4a$ edges, and $f'$ is a subgrid of $f^1$. 
\item There exists a path $P=(v_0,\ldots,v_\ell)$ such that $T[U_{i_2}]$ is a $(P,81a^2)$-path, and $P_c(u)=v_0$, and a $k_2\times r_2$ grid graph embedding $f^2$ of $T[U_{i_2}\cup U]$ such that $P$ moves left in $f^2$, except for at most $4a$ edges, and $f'$ is a subgrid of $f^2$. 
\item There exists a path $P=(v_0,\ldots,v_\ell)$ such that $T[U_{i_3}]$ is a $(P,81a^2)$-path, and $P_c(u)=v_0$, and a $k_3\times r_3$ grid graph embedding $f^3$ of $T[U_{i_3}\cup U]$ such that $P$ moves down in $f^3$, except for at most $4a$ edges, and $f'$ is a subgrid of $f^3$.
\item There exists a path $P=(v_0,\ldots,v_\ell)$ such that $T[U_{i_4}]$ is a $(P,81a^2)$-path, and $P_c(u)=v_0$, and a $k_4\times r_4$ grid graph embedding $f^4$ of $T[U_{i_4}\cup U]$ such that $P$ moves right in $f^4$, except for at most $4a$ edges, and $f'$ is a subgrid of $f^4$.
\item $k_1+k_3-\mathsf{length}(f')\leq k$ and $r_2+r_4-\mathsf{width}(f')\leq r$. 
\end{enumerate} 
Then, $T$ is a $k\times r$ grid graph. 

\end{lemma}

\begin{proof}
First, we show that $f^j$ agrees with $f^i$, for every $1\leq i<j\leq 4$. We prove that $f^3$ agrees with $f^1$, the other cases can be proved similarly. Since $f'$ is subgrid of $f^1$ and $f^3$,$V(f^1)\cap V(f^3)=V(f')$, and $|V(f')|\geq 1$, it is enough to show that for every $u_1\in U_1\setminus U$, $\fr^1(u_1)\leq\fr^1(u)-202a^2$, and for every $u_3\in U_3\setminus U$, $\fr^3(u_3)\geq\fr^3(u)-163a^2$. Let $u_1\in U_1\setminus U$ and let $P=(v_0,\ldots,v_\ell)$ such that $T[U_{i_1}]$ is a $(P,81a^2)$-path, and $P_c(u)=v_0$. Observe that $\fr^1(v_0)\leq \fr^1(u)+81a^2+1$. Let $0\leq i\leq \ell$ be such that $P_c(u_1)=v_i$. Observe that since $u_1\notin V(f')$, then, $d_{f^1}(u_3,u)>365a^2$, so $d(u_3,u)>284a^2$, and so $d(v_i, u)\geq 284a^2$, thus, $i\geq 284a^2$. Since $P$ moves up in $f^1$ except for at most $a$ edges, $\fr^1(v_i)\leq \fr^1(v_0)-284a^2+a\leq \fr^1(u)-283a^2$. Since $d(u_1,v_i)\leq 81a^2$, so $d_{f^1}(u_1,v_i)\leq 81a^2$, and we get that $\fr^1(u_1)\leq \fr^1(v_i)+81a^2$. Therefore, we get that $\fr^1(u_1)\leq \fr^1(v_i)+81a^2\leq \fr^1(u)-283a^2+81a^2=\fr^1(u)-202a^2$.

Now, let $u_3\in U_3\setminus U$ and let $P=(v_0,\ldots,v_\ell)$ such that $T[U_{i_3}]$ is a $(P,81a^2)$-path, and $P_c(u)=v_0$. Observe that $\fr^3(v_0)\geq \fr^3(u)-81a^2-1$. Since $P$ moves down in $f^3$, accept at most $4a$ edges, we get that, for every $0\leq i\leq \ell$, $\fr^3(v_i)\geq \fr^3(v_0)-4a\geq \fr^3(u)-81a^2-1-4a\geq \fr^3(u)-82a^2$. Let $0\leq i\leq \ell$, such that, $P_c(u_3)=v_i$. Now, since $d(v_i,u_3)\leq 81a^2$, $\fr^3(u_3)\geq \fr^3(v_i)-81a^2\geq \fr^3(u)-82a^2-81a^2=\fr^3(u)-163a^2$.

We prove that $f^2$ agrees with $f^1$. We conclude that, $f^j$ agrees with $f^i$, for every $1\leq i<j\leq 4$. Therefore, we get that, the function, defined by $f^1*f^2*f^3*f^4$, by Observation \ref{obs:stikgrid}, is a grid graph embedding of $T$. The proof that $\mathsf{length}(f^1*f^2*f^3*f^4)\leq k_1+k_3-\mathsf{length}(f')$ and that $\mathsf{width}(f^1*f^2*f^3*f^4)\leq r_2+r_4-\mathsf{width}(f')$, is similar to the proof of Condition \ref{lem:algo1DirCon5} of Lemma \ref{lem:algo1Dir}. This ends the proof.    
\end{proof}

Now, we present an algorithm that implements the terms we described in Lemmas \ref{lem:algo1Dir} and \ref{lem:algo1Dir2}.

\begin{algorithm}[t!]
    \SetKwInOut{Input}{Input}
    \SetKwInOut{Output}{Output}
	\medskip
    {\textbf{function} $k\times r \mathsf{Grid Recognition on Trees with Double-Split Vertex}$}$(\langle T,a,k,r,u,U_1,U_2,U_3,U_4\rangle)$\;
		 
		\For{every $1\leq i\leq 4$} 	
	{\label{algo22:Line1}
	Find a path $P^i$ such that $T[U_i]$ is a $(P^i,81a^2)$-path\;}
		
		\For{every $U\subseteq r(u,367a^2+1)$} 
		{\label{algo22:Line2}
		\For{every $i_1,i_2,i_3,i_4\in [4]$ such that  $i_j\neq i_\ell$ for every $1\leq j<\ell \leq 4$} 
			{\label{algo22:Line3}
				\For{every $(366a^2+1)\times (366a^2+1)$-grid graph embedding $f'$ of $T[U]$ such that $f(u)=(183a+1,183a+1)$} 
				{\label{algo22:Line4}
						\For{every $k_1,k_2,k_3,k_4,r_1,r_2,r_3,r_3,r_4\in \mathbb{N}$ such that $k_1+k_3-\mathsf{length}(f')\leq k$ and $r_2+r_4-\mathsf{width}(f')\leq r$}
						{\label{algo22:Line5}
						Find a $k_1\times r_1$-grid embedding $f$ of $T[U_{i_1}]$ such that $f$ agrees with $f'$ and $P^{i_1}$ moves up in $f$\;
						Find a $k_2\times r_2$-grid embedding $f$ of $T[U_{i_2}]$ such that $f$ agrees with $f'$ and $P^{i_2}$ moves left in $f$\; 
						Find a $k_3\times r_3$-grid embedding $f$ of $T[U_{i_3}]$ such that $f$ agrees with $f'$ and $P^{i_3}$ moves down in $f$\;
						Find a $k_4\times r_4$-grid embedding $f$ of $T[U_{i_4}]$ such that $f$ agrees with $f'$ and $P^{i_4}$ moves right in $f$\;
						\If{found such embeddings}
						{\Return ``yes-instance''\;}
						}
				}
		
		}
	
		}
		
		{\Return ``no-instance''\;}

    \caption{$k\times r \mathsf{Grid Recognition on Trees with Double-Split Vertex}$}
    \label{alg:Grid Recognition IterationTree}
\end{algorithm}

\begin{lemma} \label{lem:caseDouble}
There exists an algorithm that gets as input positive integers $a,k,r\in \mathbb{N}$, a tree $T=(V,E)$ that has exactly one $81a^2$-split vertex $u$, which is an $81a^2$-double-split vertex, the four subsets of vertices of connected components in $T\setminus \{u\}$, $U_1,U_2,U_3$ and $U_4$. If the algorithm returns ``yes-instance'', then $T$ is a $k\times r$ grid graph. Otherwise, there is no $k\times r$ grid graph embedding $f$ of $T$ such that $a_f\leq a$. Moreover, the algorithm works in $\OO(|V(G)|^{\OO(1)}a^{a^{\OO(1)}})$.
\end{lemma}

\begin{proof}
We show that Algorithm \ref{alg:Grid Recognition IterationTree} satisfies the condition of the lemma. The correctness of the algorithm is derived from Lemmas \ref{lem:algo1Dir} and \ref{lem:algo1Dir2}.

As for the running time, by Lemma \ref{lem:noSplitIsPTPath}, the computation in Line \ref{algo22:Line1} takes $\OO(|V(T)|^2)$. Now, in Line \ref{algo22:Line2}, notice that the size of $r(u,367a^2+1)$ is bounded by $(367a^2+1)^2$. The reason for that, is that, in a grid graph, all the vertices with graph distance, less or equal to $r\in \mathbb{N}$, must be embedded at locations, with grid distance from $u$, less or equal to $r$, and there are $\OO(r^2)$ such possibilities. So, we get that there are $\OO(2^{367a^4})$ iterations in the loop in Line \ref{algo22:Line2}. The number of iterations in the loop in Line \ref{algo22:Line3} is bounded by $\OO(1)$. In Line \ref{algo22:Line4}, we have at most $(367a^4)!=\OO({a}^{a^{\OO(1)}})$ possibilities for $f'$. In Line \ref{algo22:Line5}, we have at most $k^4r^4=\OO(|V(T)|^8)$ iterations, where each iteration, by Lemma \ref{algoExPath}, takes $\OO(|V(G)|^{13}(a^{10})^{\OO(a^{10})})$. Therefore, the overall time for the algorithm is bounded by $\OO(|V(G)|^{\OO(1)}a^{a^{\OO(1)}})$.   
\end{proof}

Now, Algorithm \ref{alg:Grid Recognition IterationTree} deals with the case where $T$ has exactly one $81a^2$-split vertex $u$, which is an $81a^2$-double-split vertex. By Lemma \ref{lem:splitVer}, there are a few more cases. We state the slight changes to be made for them. Assume that $T$ has two $81a^2$-one-split vertices, $u_1$ and $u_2$. In this case, we do similar operations. We guess a ``small environment'' around each of the split vertices. Then, as in the case of the one double-split vertex, we look for grid graph embeddings, of each of the connected components in $T\setminus \{u_1,u_2\}$, that agrees with our guesses, and ``glue'' the embeddings in order to get a grid graph embedding of $T$.

Now, assume that $T$ has exactly one $81a^2$-split vertex $u$, which is an $81a^2$-one-split vertex. Here, we do the same operations, as in the case of the one double-split vertex, with the following change. Observe that, in this case, we have three ``main'' $(P^i,81a^2)$-paths, such that, as proved in the case of the one double-split vertex, each `moves'' to a different direction. Therefore, since there are four directions, there might be a direction which is ``not in use''. Thus, it might not be clear, as it was clear in the double-split vertex case, where are the ``extreme'' vertices, that is, vertices that are located in a maximum or minimum number of row or column. Therefore, we also need to track on the location of these extremes. In the algorithm we described before Lemma \ref{algoExPath}, we also save the location of these vertices.

Observe that these changes do not change the running time of Algorithm \ref{alg:Grid Recognition IterationTree} by much, and we still have an \FPT\ algorithms.

In Lemma \ref{lem:caseDouble}, we show how to handle one of the cases from Lemma \ref{lem:splitVer}, where we wish to construct a grid graph embedding, $f$, of $T$, where $a_f\leq a$. In addition, we state that, for every other case, there exists also an algorithm that solves that case, in a similar running time. Therefore, for a given $a$, first we need to see in what case we are, that is, how many split-vertices $T$ has, and what kind of split-vertices. This clearly can be done in $|V(G)^{\OO(1)}|$. Then, we only need an algorithm that iterates over the possibilities for $a$, in order to get an \FPT\ that solves the \bGridEm\ problem, with respect to $a_T$, similarly to Algorithm \ref{alg:Grid Recognition} of the previous subsection. We remark that we do not need to know the value of $a_T$ in advance in order to use our algorithm, as we iterate over all the potential values of $a_T$. 

We summarize the result of this subsection in the next theorem.

\distanceTrees*

\section{\textsf{Para-NP}-hardness Results}\label{sec:hardness}
In this section, we show that the \bGridEm\ problem is \textsf{para-NP}-hard with respect to several parameters. We start with a \textsf{para-NP}-hardness result with respect to the parameter $k + \pw$.

\hardnessK*

\begin{proof}
We give a reduction from \partition. In this problem, we are given a multiset $W$ of $3m$ positive integers and the objective is to decide whether $W$ can be partitioned into $m$ triplets $W_1, W_2,\ldots, W_m$ such that the sum of each triplet is the same. 

\smallskip\noindent {\bf Reduction:} Given an instance $W$ of \partition, we construct an instance $G$ of \bGridEm\ as follows. Let $W = \{w_1, w_2 \ldots, w_{3m}\}$ be the multiset of $3m$ positive integers and $B$ be the required sum for each subset of the sought partition, i.e. $B = \sum W/m$. Without loss of generality, we can assume that every element in $W$ is greater than $2$ and that only sums of exactly three elements can be equal to $B$, as otherwise we can get an equivalent instance of \partition\ with this property by adding say, $\sum W$ to each of the elements of $W$ and $3\sum W$ to $B$. For every $i \in [3m]$, we create a path $P_i$ of size $w_i$. Let $V(P_i) = \{p^i_1, p^i_2, \ldots, p^i_{w_i}\}$ and $E(P_i) = \bigcup_{j \in [w_i-1]}\big\{\{p^i_j, p^i_{j+1}\}\big\}$. Let $\cal P$ be the set of all these paths. We create $m$ copies of some graph $G'$ which, intuitively, corresponds to the sum $B$ to be attained by each subset of the partition, and are defined as follows (see Figure~\ref{fi:hardnessK}). For every $j \in [m]$, denote the vertex set $V_j(G')$ of the $j$-th copy of $G'$ by $\{c^j_1, c^j_2,\ldots, c^j_{B+4}, s^j_1, s^j_2, \ldots, s^j_{B+4}\}$. Then, we add the following edges between the vertices of $V_j(G')$ to get the corresponding edge set $E_j(G')$: (i) for all $i \in [B+3], \{c^j_i, c^j_{i+1}\}$, (ii) for all  $i \in [2, B+3], \{c^j_i, s^j_i\}$, (iii) for all $i \in [2, B+2], \{s^j_i, s^j_{i+1}\}$, and (iv) $\{c^j_2, s^j_1\}, \{c^j_{B+3}, s^j_{B+4}\}$.
We finally create a graph $G$ which is the disjoint union of $m$ copies of $G'$ and the path $P$, i.e., $V(G) = \cup_{j \in [m]} V_j(G') \bigcup \cup_{i \in [3m]} V(P_i)$ and $E(G) = \cup_{j \in [m]} E_j(G') \bigcup \cup_{i \in [3m]} E(P_i)$. 

\begin{figure}[!t]
\centering
\includegraphics[width=0.5\textwidth, page=1]{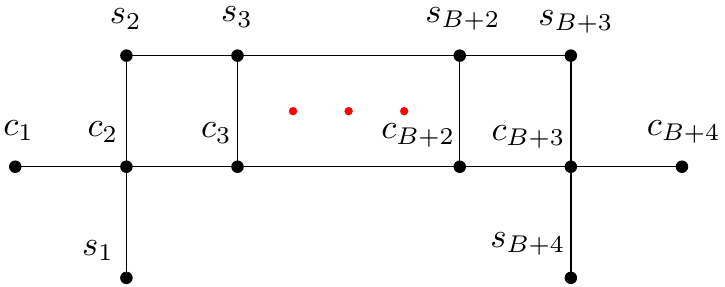}
\caption{Example of the graph $G'$ built in the proof of Theorem~\ref{thm:hardnessK}.}\label{fi:hardnessK}
\end{figure}

We now prove that $\pw(G') = 2$. Consider the sequence $PD = \{\{c_1, s_1, c_2\}, \{c_2, s_2, c_3\},\{s_2,$ $ c_3, s_3\}, \ldots,\{c_i, s_i, c_{i+1}\},$ $\{s_i, c_{i+1}, s_{i+1}\}, \ldots, \{c_{B+2}, s_{B+2}, c_{B+3}\}, \{s_{B+2}, c_{B+3}, s_{B+3}\},\{c_{B+3},$ $s_{B+4}, c_{B+4}\}\}$. Observe that $PD$ is a path decomposition of $G'$, by Definition~\ref{def:pathwidth}, $\pw(G') = 2$. Note that as the pathwidth of $G'$ is $2$ and that of a path is $1$, by Observation~\ref{obs:pathwidthDisjoint}, the pathwidth of $G$ is $2$.

We now prove that $W$ is a \yes\ instance of \partition\ if and only if $G$ is a $3 \times r$ grid graph, where $r=m(B+4)$.

\smallskip\noindent {\bf Forward Direction:} Let $W$ be a \yes\ instance of \partition. Let $W_1, W_2, \ldots, W_m$ be a corresponding partition of $W$. Let $W_j = \{w_{x_j}, w_{y_j}, w_{z_j}\}$, for every $j \in m$ and some $x_j, y_j, z_j \in [3m]$. Let $f:V(G) \rightarrow [3] \times [r$] be a function defined as follows (see Figure~\ref{fi:hardnessKEmbedding}):
\begin{enumerate}[(i)]
	\item for all $j \in [m],$ for all $i \in [B+4],$ let $f(c^j_i) = \big(2, (j-1)(B+4)+i\big)$,
	\item for all $j \in [m],$ for all $i \in [2, B+3],$ let $f(s^j_i) = \big(3, (j-1)(B+4)+i\big)$,
	\item for all $j \in [m],$ let $f(s^j_1)= \big(1, (j-1)(B+4)+2\big), f(s^j_{B+4}) = \big(1, (j-1)(B+4)+(B+3)\big)$,
	\item for all $j \in [m],$ for all $i \in [x_j],$ let $f(p^{x_j}_i) = \big(1, (j-1)(B+4)+2+i\big)$,
	\item for all $j \in [m],$ for all $i \in [y_j],$ let $f(p^{y_j}_i) = \big(1, (j-1)(B+4)+2+w_{x_j}+i\big)$,
	\item for all $j \in [m],$ for all $i \in [z_j],$ let $f(p^{z_j}_i) = \big(1, (j-1)(B+4)+2+w_{x_j}+w_{y_j}+i\big)$.
\end{enumerate}
Note that, as for every $j \in [m], \sum W_j = B$ and hence $w_{x_j}+w_{y_j}+w_{z_j} = B$, it follows that $f$ does not assign same value for two vertices of $V(G)$, i.e.~$f$ is an injective function. Moreover, for every $\{v,u\}\in E(G), d_f(u,v)=1$. So, by Definition~\ref{def:Grid graph embedding}, $G$ is a $3 \times r$ grid graph.

\smallskip\noindent {\bf Reverse Direction:} Let $G$ be a $3 \times r$ grid graph. Let $f$ be a corresponding injective function. As for every $j \in [m], \degr_G(c^j_2) = \degr_G(c^j_{B+3}) = 4$, it must be that $\fr(c^j_2) = \fr(c^j_{B+3}) = 2$. Also, as for every $j \in [m]$ and $i \in [2, B+2],$ $(c^j_2, c^j_3, \ldots, c^j_{B+3}, s^j_{B+3}, s^j_{B+2}, \ldots, s^j_2, c^j_2)$ and $(c^j_i, c^j_{i+1}, s^j_{i+1}, s^j_i, c^j_i)$ are cycles in $G$, it follows that for every $j \in [m]$ and $i \in [3, B+2],$ $\fr(c^j_i) = 2$ and $\fr(s^j_2) = \fr(s^j_3) = \ldots = \fr(s^j_{B+3}) = 1$ or $\fr(s^j_2) = \fr(s^j_3) = \ldots = \fr(s^j_{B+3}) = 3$. Without loss of generality, assume that for every $j \in [m]$ and $ i \in [2, B+3]$, $\fr(s^j_i) = 3$. As for every $j \in [m]$, $\fr(c^j_2) = 2, \fr(c^j_3) = 2$ and $ \fr(s^j_2) = 3$, it must be that either both $\fr(c^j_1) = 2$ and $\fr(s^j_1) =1$ or both $\fr(c^j_1) = 1$ and $\fr(s^j_1) =2$. A similar argument follows for $\fr(c^j_{B+4})$ and $\fr(s^j_{B+4})$. So, without loss of generality, we can also assume that for every $j \in [m]$, $\fr(c^j_1) = \fr(c^j_{B+4}) = 2$ and $ \fr(s^j_1) = \fr(s^j_{B+4}) = 1$. So, for every $j \in [m], |\fc(c^j_{B+4}) - \fc(c^j_1)| = B+3$. As $r = m(B+4)$, without loss of generality, we can assume that for every $j \in [m]$ and $i \in [B+4], \fc(c^j_i) = (j-1)(B+4)+i$. This necessarily implies that for every $j \in [m]$ and $i \in [2, B+3], \fc(s^j_i) = (j-1)(B+4)+i, \fc(s^j_1)= (j-1)(B+4)+2$ and $ \fc(s^j_{B+4}) = (j-1)(B+4)+(B+3)$. Therefore, as every path $P \in \cal P$ is of size more than $2$, $\fc(V(P)) \in [\fc(c^j_3), \fc(c^j_{B+2})]$ for some $j \in [m]$. For every $j \in [m]$, let ${\cal P}_j$ be the set of paths $P \in {\cal P}$ with the property that $\fc(V(P)) \in [\fc(c^j_3), \fc(c^j_{B+2})]$. As $\sum_{P \in \cal P}|V(P)| = mB$, necessarily $\sum_{P \in {\cal P}_j}|V(P)| = B$ for every $j \in [m]$. By the property of $W$ that only three elements can sum up to $B$, $|{\cal P}_j| = 3$ for every $j \in [m]$. So, for every $j \in [m]$, it is well defined to denote ${\cal P}_j = \{P^j_1, P^j_2, P^j_3\}$. Then, we define the set $W_j = \{|P^j_1|, |P^j_2|, |P^j_3|\}$ for every $j \in [m]$. Observe that, $\{W_1, W_2, \ldots, W_m\}$ is a partition of $W$ satisfying the required property, so $W$ is a \yes\ instance of \partition.
\begin{figure}[!t]
\centering
\includegraphics[width=\textwidth, page=2]{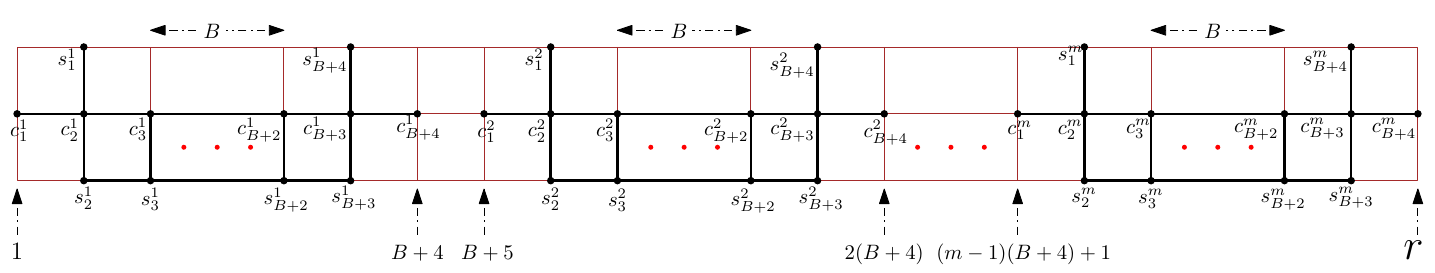}
\caption{Embedding of $m$ copies of $G'$ in a $3 \times r$ embedding of the graph $G$.}\label{fi:hardnessKEmbedding}
\end{figure}
\end{proof}

Next, we prove that the \gridEm\ problem is \paraH\ with respect to the parameter $\pw$. We first define few terms related to a grid graph embedding of a path which will be useful for the next proof. 
Given a path $P=v_1-v_2-\cdots-v_k$ and a grid graph embedding $f$ of it, we say that the path $P$ is \emph{straight} if all the vertices of $P$ are embedded on a straight line, i.e., either $\fr(v)$ or $\fc(v)$ is the same for all the vertices $v \in V(P)$. If $P$ is not straight, it is said of be \emph{bent}. Given a vertex $v_i \in V(P)$, $P$ is said to be \emph{bent at $v_i$} if the subpath $v_{i-1}-v_i-v_{i+1}$ is bent. Moreover, if we turn right (left) while going from $v_{i-1}$ to $v_{i+1}$ through $v_i$, we say the bend at $v_i$ is \emph{towards the right (left)}.

\hardnessUnrestricted*

\begin{proof}
We present a reduction from {\sc Not-All-Equal SAT (NAE-SAT)}. In this problem, we are given a formula in conjunctive normal form (CNF) and the objective is to decide whether there exists a truth assignment to the variables so that each clause has at least one true literal and at least one false literal.

\smallskip\noindent {\bf Reduction:} Given an instance $\varphi$ of {\sc NAE-SAT}, we construct an instance $G$ of \gridEm\ as follows (see also Figure~\ref{fi:hardnessUnR}). Let $X=\{x_1,x_2,\ldots,x_n\}$ be the set of variables of $\varphi$, and let $C=\{c_1,c_2,\ldots,c_m\}$ be the set of clauses of $\varphi$. Then, for every variable $x_i\in X$, we construct two caterpillars, $P_i$ and $\overline{P}_i$, as follows. 
First, let $P_i={v}^1_i-{v}^2_i-\cdots-{v}^{2m+1}_i$ and $\overline{P}_i={\overline{v}}^1_i-{\overline{v}}^2_i-\cdots-{\overline{v}}^{2m+1}_i$ be two paths on $2m+1$ vertices. Now, for every odd $j\in\{1,\ldots,2m-1\}$, if $x_i$ does {\em not} appear in $c_{\lceil j/2\rceil}$, then add the vertex $u^j_i$  and the edge $\{u^j_i,v^j_i\}$  to $P_i$, and if $\overline{x}_i$ does {\em not} appear in $c_{\lceil j/2\rceil}$, then add the vertex $\overline{u}^j_i$ and the edge $\{\overline{u}^j_i,\overline{v}^j_i\}$ to $\overline{P}_i$. Additionally, for every even $j\in\{2,\ldots,2m\}$ and $i \in [1,n-1]$, add the vertex $u^j_i$  and the edge $\{u^j_i,v^j_i\}$  to $P_i$, and add the vertex $\overline{u}^j_i$ and the edge $\{\overline{u}^j_i,\overline{v}^j_i\}$ to $\overline{P}_i$. Moreover, let $B=b_0-b_0'-b_1-b_1'-\cdots-b_{n}-b_n'-b_{n+1}$ be a path on $2n+1$ vertices. We connect $B$ to all aforementioned paths as follows. For every $i\in\{1,\ldots,n\}$, we add the edges $\{b_i,v^1_i\}$ and $\{b_i,\overline{v}^1_i\}$.  

\begin{figure}[!t]
\centering
\includegraphics[width=\textwidth, page=4]{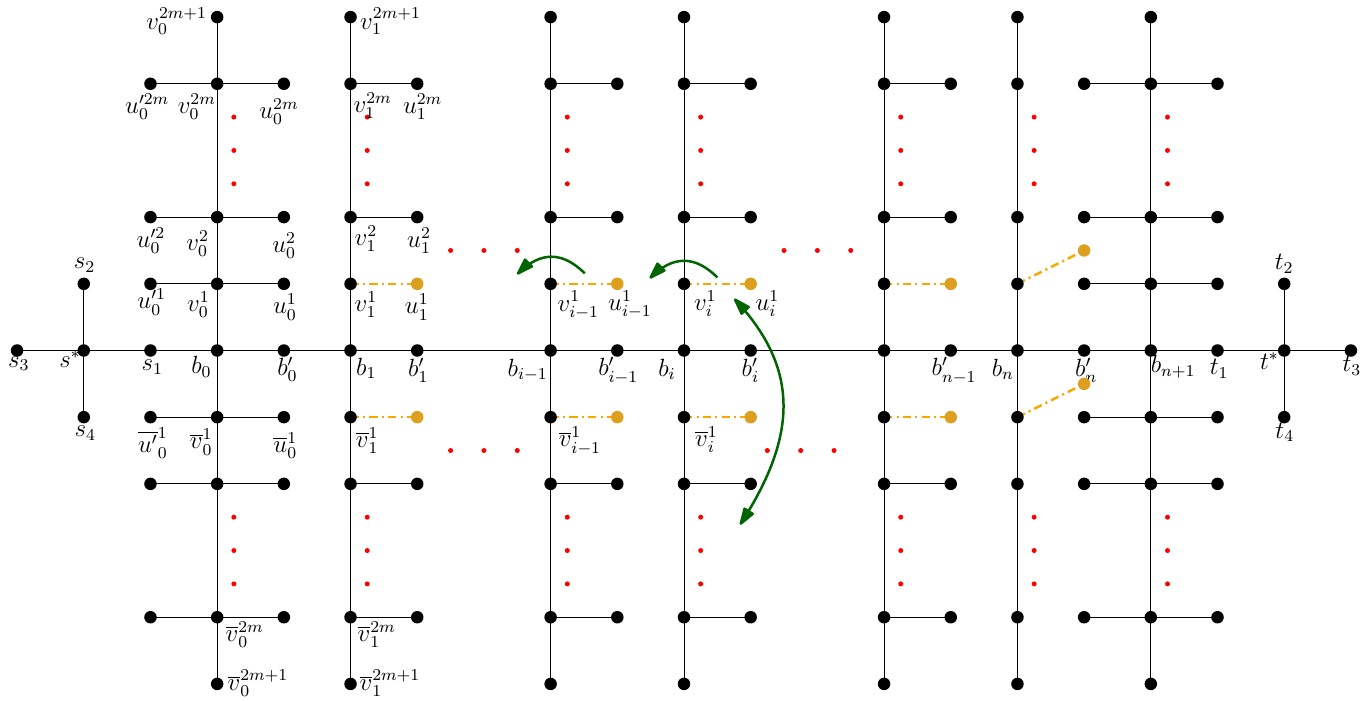}
\caption{Example of the graph $G$ built in the reduction of Theorem~\ref{thm:hardnessUnR}. The dashed edges may or may not be present in the graph. For every $i \in [n]$ and for every odd $j\in [2m-1]$, we add the dashed edge between the vertices $v^j_i$ and $u^j_i$ ($\overline{v}^j_i$ and $\overline{v}^j_i$) if the variable $x_i$ ($\overline{x}_i$) does not appear in clause $c_{\lceil j/2\rceil}$. The green arrows denote the possible embeddings for the respective components.}\label{fi:hardnessUnR}
\end{figure}

We further extend the graph constructed so far to obtain $G$ by adding four additional caterpillars as well as two stars as follows. First, let $P_0={v}^1_0-{v}^2_0-\cdots-{v}^{2m+1}_0$, $\overline{P}_0={\overline{v}}^1_0-{\overline{v}}^2_0-\cdots-{\overline{v}}^{2m+1}_0$, $P_{n+1}={v}^1_{n+1}-{v}^2_{n+1}-\cdots-{v}^{2m+1}_{n+1}$ and $\overline{P}_{n+1}={\overline{v}}^1_{n+1}-{\overline{v}}^2_{n+1}-\cdots-{\overline{v}}^{2m+1}_{n+1}$ be four paths on $2m+1$ vertices. Now, for every $j\in\{1,2,\ldots,2m\}$, add the two vertices $u^j_0,{u'}^j_0$  and the two edges $\{u^j_0,v^j_0\},\{{u'}^j_0,v^j_0\}$  to $P_0$, add the two vertices $\overline{u}^j_0,{\overline{u}'}^j_0$  and the two edges $\{\overline{u}^j_0,\overline{v}^j_0\},\{{\overline{u}'}^j_0,\overline{v}^j_0\}$  to $\overline{P}_0$, add the two vertices $u^j_{n+1},{u'}^j_{n+1}$  and the two edges $\{u^j_{n+1},v^j_{n+1}\},\{{u'}^j_{n+1},v^j_{n+1}\}$  to $P_{n+1}$, and add the two vertices $\overline{u}^j_{n+1},{\overline{u}'}^j_{n+1}$  and the two edges $\{\overline{u}^j_{n+1},\overline{v}^j_{n+1}\},\{{\overline{u}'}^j_{n+1},\overline{v}^j_{n+1}\}$ to $\overline{P}_{n+1}$. Then, we add the four edges $\{b_0,v^1_0\},\{b_0,\overline{v}^1_0\},\{b_{n+1},v^1_{n+1}\}$ and $\{b_{n+1},\overline{v}^1_{n+1}\}$. Lastly, we create a star $S$ with center $s^\star$ and four leaves $s_1,s_2,s_3$ and $s_4$, as well as a star  $T$ with center $t^\star$ and four leaves $t_1,t_2,t_3$ and $t_4$, and add the edges $\{b_0,s_1\}$ and $\{b_{n+1},t_1\}$. 

This completes the construction of $G$. Clearly, $G$ is a tree. Moreover, because $G$ is a collection of caterpillars connected to a base path, it has pathwidth $2$ (this can also be easily verified explicitly). For every caterpillar $P_i$, we call the path ${v}^1_i-{v}^2_i-\cdots-{v}^{2m+1}_i$ as the \emph{main path} of $P_i$ and the set of vertices $\{u^2_i, u^4_i, \cdots, u^{2m}_i\}$ as the \emph{even leaves} of $P_i$. Moreover, for any $t \in [2m]$, we call the set of even leaves $\{u^2_i, u^4_i, \ldots, u^{2\lfloor{t/2}\rfloor}_i\}$ as the \emph{even leaves of $P_i$ until $v^t_i$}. We define these terms for $\overline{P}_i$ as well in a similar fashion as that of $P_i$. Also, for any $t \in [0, n]$, we call the set of vertices $\{b_0, b'_0 , \ldots, b_t\}$ ($\{b_0, b'_0 , \ldots, b'_t\}$) as the \emph{vertices of $B$ until $b_i$ $(b'_i)$}. 

We now prove that $\varphi$ admits a solution if and only if $G$ is a grid graph.

\smallskip\noindent {\bf Forward Direction:} In the forward direction, we suppose that $\varphi$ admits a solution, which is an assignment $\alpha:X\rightarrow\{\mathrm{false},\mathrm{true}\}$ such that every clause has at least one literal assigned true and at least one literal assigned false. Then, we assert that $G$ is a grid graph by embedding it as follows. First, we embed $B$ on a straight horizontal line, flanked by $s_3-s^\star-s_1$ to the left, and $t_1-t^\star-t_3$ to the right, which also fixes the embedding of the rest of the vertices $S$ and $T$ (up to isomorphism, e.g., $s_2$ and $s_4$ can be swapped). Then, we embed the main paths of the caterpillars $P_0$ and $P_{n+1}$ on vertical lines above $B$ and the main paths of the caterpillars $\overline{P}_0$ and $\overline{P}_{n+1}$ on vertical lines below $B$, which also fixes the embedding of the leaves of these four caterpillars up to isomorphism. Now, for every variable $x_i$ assigned true, we embed the main path of the caterpillar $P_i$ on a vertical line above $B$, and the main path of the caterpillar $\overline{P}_i$ on a vertical line below $B$, which also fixes the embedding of the leaves attached to vertices of even indices (up to isomorphism). Symmetrically, for every variable $x_i$ assigned false, we embed the main path of the caterpillar $\overline{P}_i$ on a vertical line above $B$, and the main path of the caterpillar ${P}_i$ on a vertical line below $B$, which also fixes the embedding of the leaves attach to vertices of even indices (up to isomorphism). It remains to fix the embedding of each leaf attached to a vertex of an odd index---specifically, whether it is on the left or the right of the main path.  Think of these leaves attached to vertices of the same odd index $2j-1$ above $B$ as a ``row'', which corresponds to some clause $c_j$, and notice that because $\alpha$ is a solution, there is a variable $x_i$ such that either $x_i$ appears in $c_j$ and is assigned true by $\alpha$, in which case $P_i$ appears above $B$ and $v^j_i$ has no leaves attached to it, or $\overline{x}_i$ appears in $c_j$ and is assigned false by $\alpha$, in which case $\overline{P}_i$ appears above $B$ and $\overline{v}^j_i$ has no leaves attached to it. Thus, the row has a ``free'' position, and so up until this position the leaves can be placed to the right of their main paths, and from that position onwards they can be placed to the left. The symmetric argument holds for ``rows'' below $B$. For example, see Figure~\ref{fi:hardnessUnRExa}.

\begin{figure}[!t]
\centering
\includegraphics[width=0.9\textwidth, page=15]{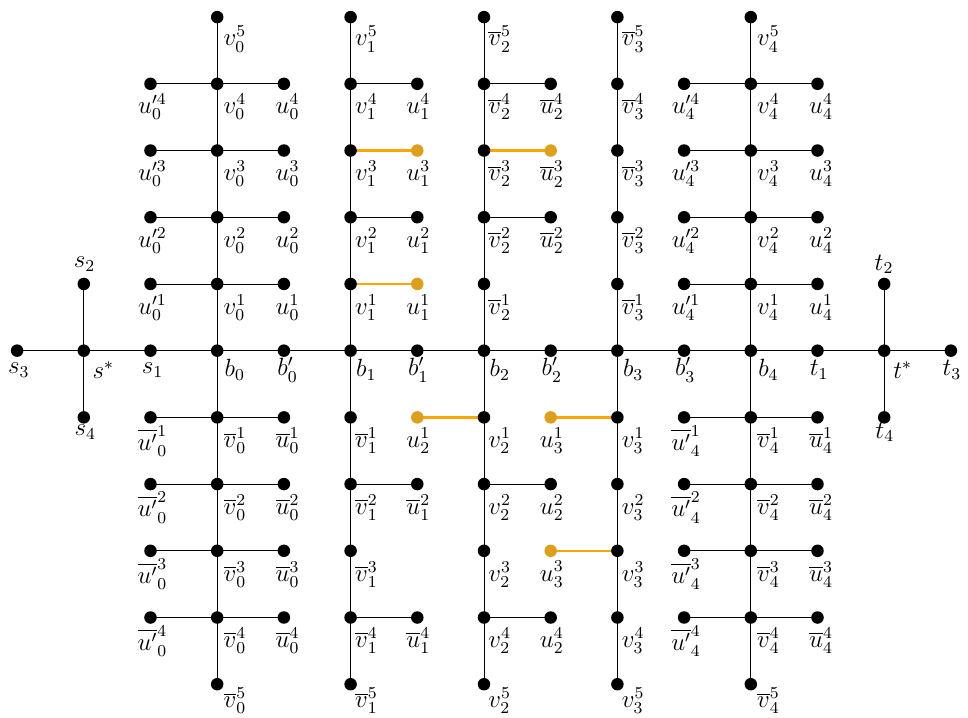}
\caption{Example of the graph $G$ and its grid graph embedding built for a NAE-SAT instance $\varphi = (\overline{x}_1 \vee \overline{x}_2 \vee \overline{x}_3) \wedge (\overline{x}_1 \vee x_2 \vee \overline{x}_3)$ and a corresponding solution \{$x_1$ = true, $x_2$ = false, $x_3$ = false\}.}\label{fi:hardnessUnRExa}
\end{figure}

\smallskip\noindent {\bf Reverse Direction:} In the reverse direction, we suppose that we have an embedding of $G$ in a grid. We first argue that the embedding of $P_0,\overline{P}_0,P_{n+1},\overline{P}_{n+1}$ and the stars is essentially fixed in the sense that the main paths of $P_0$ and $\overline{P}_0$ should be on the same straight line (without loss of generality, say, a vertical one), the star $S$ should be placed, without loss of generality, to the left of them, the main paths of $P_{n+1}$ and $\overline{P}_{n+1}$ should be on the same straight line (though, having fixed the previous pair as vertical, we cannot say, without loss of generality, that these are vertical too yet), and that the star $T$ should be place to the left or right of them. Observe that, the subgraph of $G$ induced by the vertex set $V' = V(P_0) \cup V(\overline{P}_0) \cup \{s_1, b_0, b'_0\}$ is the spine graph used in~\cite{bhatt1987complexity}. As shown in~\cite{bhatt1987complexity}, the embedding of $G[V']$ is essentially fixed as the vertices along the spine, that is the vertices of the joint main path of $P_0$ and $\overline{P}_0$ and the vertex $b_0$ has to lie on the same straight line. Similarly, vertices of the joint main path of $P_{n+1}$ and $\overline{P}_{n+1}$ and the vertex $b_{n+1}$ has to lie on the same straight line. Having shown this, we argue that $B$ should be embedded on a horizontal line, each pair of the main paths of $P_i$ and $\overline{P}_i$ (for any $i\in\{1,\ldots,n\}$) should be embedded on a single vertical line, and that $P_{n+1}$ and $\overline{P}_{n+1}$ as well should be embedded on a single vertical line, with the star $T$ to the right of it. In particular, this argument makes use of the even leaves of the caterpillars $P_i$ and $\overline{P}_i$. To prove this, we first give the following claim.

\begin{claim}\label{cla:hardnessUnRCases}
Let $i \in \mathbb{N}$ such that $i \in [1,n]$. Given any grid graph embedding of the graph $G$, if all the vertices of $B$ until $b'_{i-1}$ are embedded on a horizontal line, the main path of $P_{i-1}$ $(\overline{P}_{i-1})$ is embedded either on a vertical line or the first bend is towards the right (left) and the even leaves of $P_{i-1}$ $(\overline{P}_{i-1})$ until the first bend is towards the right (left), then the following conditions are satisfied.
\begin{enumerate}[(i)]
	\item \label{con:bend1} The path $B$ is bent at neither $b'_{i-1}$ nor $b_i$.
	\item \label{con:bend2} The main path of $P_i$ $(\overline{P}_i)$ is embedded either on a vertical line or the first bend is towards the right (left). 
	\item \label{con:bend3} The even leaves of $P_i$ $(\overline{P}_i)$ until the first bend are towards the right (left) of the main path of $P_i$ $(\overline{P}_i)$.
\end{enumerate}
\end{claim}

\begin{figure}[!t]
	\centering
	\begin{subfigure}{0.3\textwidth}
		\includegraphics[width = \textwidth, page = 5]{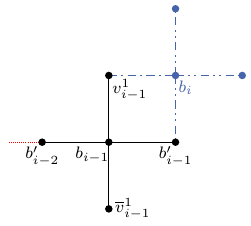}
		\subcaption{}
		\label{fi:hardnessUnRCase1}
	\end{subfigure}
	\hfil
	\begin{subfigure}{0.32\textwidth}
		\includegraphics[width = \textwidth, page = 6]{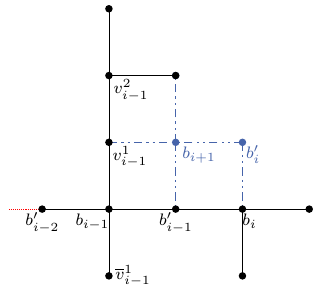}
		\subcaption{}
		\label{fi:hardnessUnRCase2}
	\end{subfigure}
	\begin{subfigure}{0.36\textwidth}
		\includegraphics[width = \textwidth, page = 7]{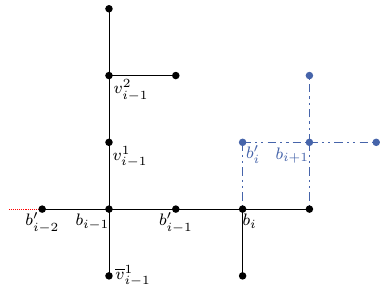}
		\subcaption{}
		\label{fi:hardnessUnRCase3}
	\end{subfigure}
	\begin{subfigure}{0.32\textwidth}
		\includegraphics[width = \textwidth, page = 8]{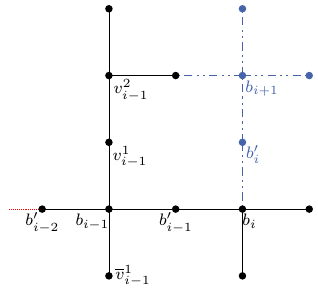}
		\subcaption{}
		\label{fi:hardnessUnRCase4}
	\end{subfigure}
	\hfil
	\begin{subfigure}{0.32\textwidth}
		\includegraphics[width = \textwidth, page = 9]{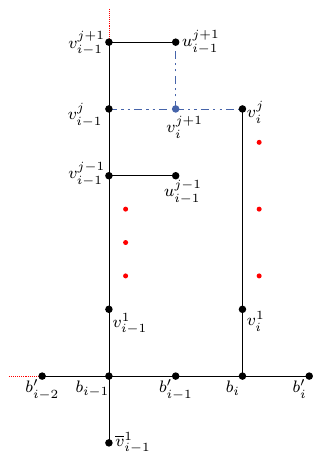}
		\subcaption{}
		\label{fi:hardnessUnRCase5}
	\end{subfigure}
	\begin{subfigure}{0.32\textwidth}
		\includegraphics[width = \textwidth, page = 10]{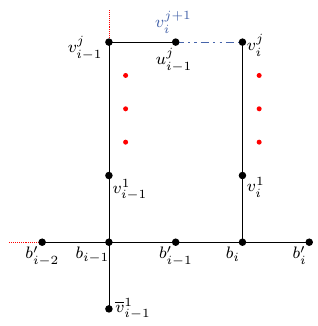}
		\subcaption{}
		\label{fi:hardnessUnRCase6}
	\end{subfigure}
	\begin{subfigure}{0.32\textwidth}
		\includegraphics[width = \textwidth, page = 11]{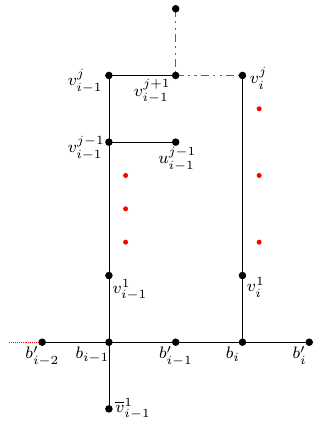}
		\subcaption{}
		\label{fi:hardnessUnRCase7}
	\end{subfigure}
	\hfil
	\begin{subfigure}{0.32\textwidth}
		\includegraphics[width = \textwidth, page = 12]{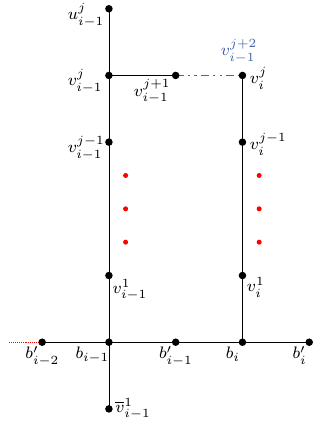}
		\subcaption{}
		\label{fi:hardnessUnRCase8}
	\end{subfigure}
	\begin{subfigure}{0.32\textwidth}
		\includegraphics[width = \textwidth, page = 13]{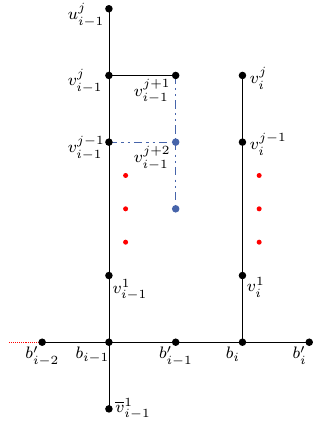}
		\subcaption{}
		\label{fi:hardnessUnRCase9}
	\end{subfigure}
	\caption{Various cases considered in the proof of Claim~\ref{cla:hardnessUnRCases}.} 
	\label{fi:hardnessUnRCases}
\end{figure}

\begin{figure}[!t]
	\ContinuedFloat
	\centering
	\begin{subfigure}{0.32\textwidth}
		\includegraphics[width = \textwidth, page = 14]{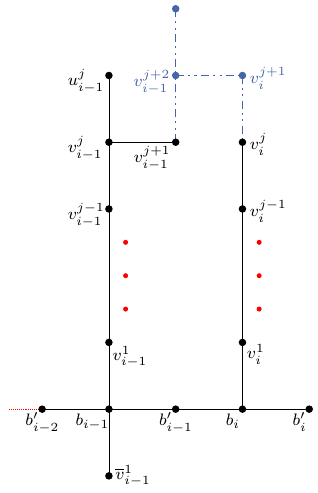}
		\subcaption{}
		\label{fi:hardnessUnRCase10}
	\end{subfigure}
	\caption{Various cases considered in the proof of Claim~\ref{cla:hardnessUnRCases} (cont.)} 
	\label{fi:hardnessUnRCasesCont}
\end{figure}
\begin{proof}
We prove the claim for $P_i$. The proof for $\overline{P}_i$ is similar. Note that the vertices with odd superscript of the main path of $P_i$ may or may not have a leaf attached to them. Thus, without loss of generality, for the proof of the claim we assume that there are no leaves attached to the vertices with odd superscript of the main path of $P_i$ (which is less restrictive for the embedding).

We first prove Condition (\ref{con:bend1}). Assume first, for contradiction, that the path $B$ is bent at $b'_{i-1}$. Then, the vertex $b_i$ is placed above or below $b'_{i-1}$. As $\deg_G(b_i) = 4$, one of its neighbors is embedded at the same position as that of $v^1_{i-1}$ or $\overline{v}^1_{i-1}$; see Figure~\ref{fi:hardnessUnRCase1}. This is a contradiction as in any grid graph embedding, any two distinct vertices are embedded at two different positions.

Assume now, for contradiction, that the path $B$ is bent at $b_i$. Then, the vertex $b'_i$ is placed above or below $b_i$. Assume without loss of generality that $b'_i$ is placed above $b_i$. Then, the vertex $b_{i+1}$ is placed to the left or right or above of $b'_i$. As $\deg_G(b_{i+1}) = 4$, if $b_{i+1}$ is placed to the left of $b'_i$, then one of its neighbors is embedded at the same position at that of $b'_{i-1}$; see Figure~\ref{fi:hardnessUnRCase2}. If $b_{i+1}$ is placed to the right of $b'_i$, then one of its neighbors is embedded at the same position as that of one of the neighbors of $b_i$; see Figure~\ref{fi:hardnessUnRCase3}. If $b_{i+1}$ is placed on the top of $b'_i$, then one of its neighbors is embedded at the same position at that of one of the neighbors of $v^2_{i-1}$ due to the fact that either the main path of $P_{i-1}$ is bent right at $v^2_{i-1}$ or the leaf $u^2_{i-1}$ is towards the right; see Figure~\ref{fi:hardnessUnRCase4}. In all these cases, this results in a contradiction as in any grid graph embedding, any two distinct vertices are embedded at two different positions.

We now prove Condition (\ref{con:bend2}). Assume first that the main path of $P_{i-1}$ is embedded on a vertical line. Towards a contradiction, suppose that the first bend of the main path of $P_i$ is at $v^j_i$ towards the left. If $j$ is odd, then one of the neighbors of $v^{j+1}_i$ is embedded at the same position as that of $v^j_{i-1}$ or $u^{j-1}_{i-1}$ or $u^{j+1}_{i-1}$, a contradiction; see Figure~\ref{fi:hardnessUnRCase5}. If $j$ is even, then $v^{j+1}_i$ is embedded at the same position as that of $u^j_{i-1}$, a contradiction; see Figure~\ref{fi:hardnessUnRCase6}. Thus, in this case, the main path of $P_i$ is embedded either on a vertical line or the first bend is towards the right. 

Assume now that the first bend of main path of $P_{i-1}$ is at $v^j_{i-1}$ towards the right. By the above argument, we can conclude that if the first bend of $P_i$ is at $v^{j'}_i$ where $j' < j$, then $P_i$ is bent at $v^{j'}_i$ towards right. Thus, we assume that if $P_i$ is bent, the first bend is at $v^{j''}_i$ where $j'' \geq j$. We prove that if $j'' \geq j$, then $j$ is even and $j'' = j$. Assume first, by contradiction, that $j$ is odd. Then, one of the neighbors of $v^{j+1}_{i-1}$ is embedded at the same position as that of $v^j_i$, a contradiction; see Figure~\ref{fi:hardnessUnRCase7}. Assume now that $j$ is even. Then, $v^{j+1}_{i-1}$ is placed to the right of $v^j_{i-1}$ and $v^{j+2}_{i-1}$ is placed above or below or to the right of $v^{j+1}_{i-1}$. If $v^{j+2}_{i-1}$ is placed to the right of $v^{j+1}_{i-1}$, then it is embedded at the same position as that of $v^j_i$, a contradiction; see Figure~\ref{fi:hardnessUnRCase8}. Similarly, if $v^{j+2}_{i-1}$ is placed below $v^{j+1}_{i-1}$, then one of the neighbors of $v^{j+2}_{i-1}$ is embedded at the same position as that of $v^{j-1}_{i-1}$ or $u^{j-2}_{i-1}$ or $v^{j-1}_i$, a contradiction; see Figure~\ref{fi:hardnessUnRCase9}. So, $v^{j+2}_{i-1}$ is placed above $v^{j+1}_{i-1}$. If $j'' > j$, one of the neighbors of $v^{j+2}_{i-1}$ is embedded at the same position as that of $v^{j+1}_i$, a contradiction; see Figure~\ref{fi:hardnessUnRCase10}. Thus, $j'' = j$.

We finally prove Condition (\ref{con:bend3}). By above arguments, we know that if the first bend of the main path of $P_{i-1}$ is at $v^j_{i-1}$ towards the right, then the first bend of the main path of $P_{i-1}$ is at $v^{j'}_{i-1}$ towards the right, where $j' \leq j$. Observe that, since the even leaves of $P_{i-1}$ until the first bend are towards the right, we cannot put any even leaf $u^k_i$, for any $k < j$, of $P_i$ towards the left as otherwise it will be embedded at the same position as that of $u^k_{i-1}$, a contradiction. Thus, the even leaves of $P_i$ until the first bend are towards the right. This completes the proof of the claim.
\end{proof}

We now return to the proof of the theorem. As we have already proved that the vertices of the joint main path of $P_0$ and $\overline{P}_0$ and the vertex $b_0$ lie on the same straight line, by inductive use of the above claim, we can conclude that the path $B$ is straight until $b_n$ and if the main path of $P_i$ $(\overline{P}_i)$ is bent towards the right (left), then the main path of every $P_j$ $(\overline{P}_j)$, for every $j \in [i+1, n]$, is bent towards right (left). Observe that, as the vertices of the joint main path of $P_{n+1}$ and $\overline{P}_{n+1}$ and the vertex $b_{n+1}$ lie on the same straight line, the main path of $P_n$ cannot be bent as otherwise one of the vertex of the main path of $P_n$ will be embedded at the same position as that of one of the leaves of $P_{n+1}$. This implies that, each pair of the main paths of $P_i$ and $\overline{P}_i$ (for any $i\in [n]$) are embedded on a single vertical line, and that $P_{n+1}$ and $\overline{P}_{n+1}$ as well are embedded on a single vertical line, with the star $T$ to the right of it.

Having proved this, we now define an assignment $\alpha:X\rightarrow\{\mathrm{false},\mathrm{true}\}$ as follows. For every variable $x_i\in X$, if $P_i$ is embedded above $B$, then we define $\alpha(x_i)$ to be true, and otherwise we define $\alpha(x_i)$ to be false. To argue that this assignment is a solution, consider some arbitrary clause $c_j$. Now, consider the ``row'' above $B$ that consists of the vertices $v_1^{2j-1},v_2^{2j-1},\ldots,v_n^{2j-1}$ possibly with negations above some of them. So that the embedding can be valid, it must be that at least one of these vertices, say, $v^{2k-1}_i$ (possibly negated), has no leaves attached to it, which means that  the literal corresponding to it (which is $x_i$ or $\overline{x}_i$, depending on whether $v^{2k-1}_i$ is negated)  appears in $c_j$, and as it is assigned true by $\alpha$, we get that the clause $c_j$ is satisfied. Symmetrically, by considering the ``row'' below $B$ that consists of the vertices $v_1^{2j-1},v_2^{2j-1},\ldots,v_n^{2j-1}$ possibly with negations above some of them, we can conclude that at least one literal of the clause $c_j$ is assigned false. For example, see Figure~\ref{fi:hardnessUnRExa}.
\end{proof}
\newpage
\bibliographystyle{plainurl}
\bibliography{Refs,MoreRefs}

\end{document}